\documentclass{amsart}
\usepackage{amssymb}
\usepackage{amsmath}
\usepackage{amsthm}
\usepackage{amscd}
\usepackage{amstext}
\usepackage{bm}
\usepackage[utf8]{inputenc}
\usepackage[dvipdfmx]{graphicx}
\usepackage{dcolumn}
\usepackage{longtable}
\usepackage[all]{xy}

\usepackage[bbgreekl]{mathbbol}
\usepackage{amsfonts}
\DeclareSymbolFontAlphabet{\mathbb}{AMSb}
\DeclareSymbolFontAlphabet{\mathbbl}{bbold}

\makeatletter
\newcommand{\colim@}[2]{%
  \vtop{\m@th\ialign{##\cr
    \hfil$#1\operator@font lim$\hfil\cr
    \noalign{\nointerlineskip\kern1.5\ex@}#2\cr
    \noalign{\nointerlineskip\kern-\ex@}\cr}}%
}
\newcommand{\colim}{%
  \mathop{\mathpalette\colim@{\rightarrowfill@\textstyle}}\nmlimits@
}
\makeatother

\theoremstyle{plain}
\newtheorem{theorem}{Theorem}[section]
\newtheorem{corollary}[theorem]{Corollary}
\newtheorem{lemma}[theorem]{Lemma}
\newtheorem{proposition}[theorem]{Proposition}
\newtheorem{assumption}[theorem]{Assumption}

\theoremstyle{definition}
\newtheorem{definition}[theorem]{Definition}
\theoremstyle{remark}
\newtheorem{remark}[theorem]{Remark}
\newtheorem{example}[theorem]{Example}
\numberwithin{equation}{section}

\newcommand{\Z}{\mathbb{Z}}
\newcommand{\Q}{\mathbb{Q}}
\newcommand{\R}{\mathbb{R}}
\newcommand{\C}{\mathbb{C}}
\newcommand{\T}{\mathbb{T}}

\newcommand{\K}{\mathcal{K}}
\newcommand{\B}{\mathcal{B}}

\newcommand{\cA}{\mathcal{A}}

\newcommand{\cD}{\mathcal{D}}

\newcommand{\cP}{\mathcal{P}}

\newcommand{\cR}{\mathcal{R}}

\newcommand{\cN}{\mathcal{N}}

\newcommand{\Fred}{\mathrm{Fred}}
\newcommand{\I}{\mathcal{I}}
\newcommand{\cpt}{\mathrm{cpt}}
\newcommand{\pt}{\mathrm{pt}}
\newcommand{\id}{\mathrm{id}}

\newcommand{\sa}{\mathrm{sa}}
\newcommand{\ska}{\mathrm{sk}}
\newcommand{\diag}{\mathrm{diag}}
\newcommand{\Corner}{\mathrm{Gapless}}
\newcommand{\BE}{\mathrm{Gapped}}
\newcommand{\gapped}{\mathrm{Gapped}}

\newcommand{\BL}{\mathrm{BL}}
\newcommand{\CA}{$C^*$-algebra}
\newcommand{\TA}{$C^{*, \tau}$-algebra}
\newcommand{\Cl}{\mathit{Cl}}

\newcommand{\bt}{{\bm t}}

\DeclareMathOperator{\ind}{\mathrm{ind}}
\DeclareMathOperator{\Ker}{\mathrm{Ker}}
\DeclareMathOperator{\Coker}{\mathrm{Coker}}
\DeclareMathOperator{\Image}{\mathrm{Im}}
\DeclareMathOperator{\rank}{\mathrm{rank}}
\DeclareMathOperator{\sign}{\mathrm{sign}}

\DeclareMathOperator{\csf}{\mathrm{sf}}
\DeclareMathOperator{\qsf}{\mathrm{sf_2}}
\DeclareMathOperator{\Ad}{\mathrm{Ad}}

\newcommand{\A}{\mathrm{A}}
\newcommand{\AIII}{\mathrm{A\hspace{-.05em}I\hspace{-.05em}I\hspace{-.05em}I}}
\newcommand{\AI}{\mathrm{A\hspace{-.05em}I}}
\newcommand{\BDI}{\mathrm{BDI}}
\newcommand{\DD}{\mathrm{D}}
\newcommand{\DIII}{\mathrm{D\hspace{-.05em}I\hspace{-.05em}I\hspace{-.05em}I}}
\newcommand{\AII}{\mathrm{A\hspace{-.05em}I\hspace{-.05em}I}}
\newcommand{\CII}{\mathrm{C\hspace{-.05em}I\hspace{-.05em}I}}
\newcommand{\CC}{\mathrm{C}}
\newcommand{\CI}{\mathrm{CI}}

\newcommand{\HH}{\mathcal{H}}
\newcommand{\HHa}{\HH^{\alpha}}
\newcommand{\HHb}{\HH^{\beta}}
\newcommand{\HHab}{\hat{\HH}^{\alpha,\beta}}
\newcommand{\HHs}{\check{\HH}^{\alpha,\beta}}

\newcommand{\TT}{\mathcal{T}}
\newcommand{\TTa}{\TT^{\alpha}}
\newcommand{\TTb}{\TT^{\beta}}
\newcommand{\TTg}{\TT^{\gamma}}
\newcommand{\TTz}{\TT^{0}}

\newcommand{\TTab}{\hat{\TT}^{\alpha,\beta}}
\newcommand{\TTs}{\check{\TT}^{\alpha,\beta}}

\newcommand{\Sab}{\mathcal{S}^{\alpha, \beta}}
\newcommand{\Szg}{\mathcal{S}^{0, \gamma}}
\newcommand{\sigmaa}{\sigma^\alpha}
\newcommand{\sigmab}{\sigma^\beta}
\newcommand{\sigmaz}{\sigma^0}
\newcommand{\sigmag}{\sigma^\gamma}

\newcommand{\Pa}{P^{\alpha}}
\newcommand{\Pb}{P^{\beta}}
\newcommand{\Pz}{P^{0}}
\newcommand{\Pg}{P^{\gamma}}
\newcommand{\Pab}{\hat{P}^{\alpha,\beta}}
\newcommand{\Ps}{\check{P}^{\alpha,\beta}}

\newcommand{\rTTa}{\TTa_0}

\newcommand{\fa}{\mathfrak{a}}
\newcommand{\fC}{\mathcal{C}}
\newcommand{\fJ}{\mathcal{J}}
\newcommand{\fq}{\mathfrak{q}}
\newcommand{\fr}{\mathfrak{r}}

\newcommand{\taua}{\tau_\alpha}
\newcommand{\taub}{\tau_\beta}
\newcommand{\taug}{\tau_\gamma}
\newcommand{\tauz}{\tau_0}
\newcommand{\tauab}{\hat{\tau}_{\alpha,\beta}}
\newcommand{\taus}{\check{\tau}_{\alpha,\beta}}

\begin{document}

\title[Classification of Topological Invariants Related to Corner States]{Classification of Topological Invariants Related to Corner States}

\author[S. Hayashi]{Shin Hayashi}

\thanks{AIST-TohokuU Mathematics for Advanced Materials - Open Innovation Laboratory, National Institute of Advanced Industrial Science and Technology, 2-1-1 Katahira, Aoba, Sendai 980-8577, Japan}
\thanks{JST, PRESTO, 4-1-8 Honcho, Kawaguchi, Saitama, 332-0012, Japan}
\email{shin-hayashi@aist.go.jp}

\subjclass{Primary 19K56; Secondary 47B35, 81V99.}
\keywords{Topologically protected corner states, higher-order topological insulators, $K$-theory and index theory}

\begin{abstract}
We discuss some bulk-surfaces gapped Hamiltonians on a lattice with corners, and propose a periodic table for topological invariants related to corner states aimed at studies of higher-order topological insulators.
Our table is based on four things: (1) the definition of topological invariants, (2) a proof of their relation with corner states (3) computations of $K$-groups and (4) a construction of explicit examples.
\end{abstract}

\maketitle

\setcounter{tocdepth}{1}
\tableofcontents


\section{Introduction}
Recent developments in condensed matter physics have greatly generalized the bulk-boundary correspondence for topological insulators to include corner states.
Topological insulators have a gapped bulk, which incorporates some topology that do not change unless the spectral gap of the bulk Hamiltonian closes under deformations.
Examples include the TKNN number for quantum Hall systems \cite{TKNN82} and the Kane-Mele $\Z_2$ index for quantum spin Hall systems \cite{KM05b}.
It is known that, corresponding to these bulk invariants, gapless edge states appear, which is called the {\em bulk-boundary correspondence} \cite{Hat93b}.
After Schnyder--Ryu--Furusaki--Ludwig's classification of topological insulators \cite{SRFL08} for ten Altland--Zirnbauer classes \cite{AZ97}, Kitaev noted the role of $K$-theory and Bott periodicity in the classification problem and obtained the famous periodic table \cite{Kit09}.
Recently, some (at least bulk) gapped systems possessing in-gap or gapless states localized around a higher codimensional part of the boundary (corners or hinges) are studied \cite{HWK17,BBH17a,Khalaf18a}, which are called {\em higher-order topological insulators} (HOTIs) \cite{Schindler18}.
For example, for second-order topological insulators, not only is the bulk gapped but also the codimension-one boundaries (edges, surfaces), and an in-gap or a gapless state appears around codimension-two corners or hinges.
In this framework, conventional topological insulators are regarded as first-order topological insulators.
HOTIs are now actively studied and the classification of HOTIs has also been proposed \cite{Geier18,Khalaf18b,OSS19}.
Generalizing the bulk-boundary correspondence, relations between some gapped topology and corner states are much discussed \cite{TB19,AMH20,TTM20}.

Initiated by Bellissard, $K$-theory and index theory are known to provide a powerful tool to study topological insulators.
Bellissard--van~Elst--Schulz-Baldes studied quantum Hall effects by means of noncommutative geometry \cite{Be86,BvES94}, and Kellendonk--Richter--Schulz-Baldes went on to prove the bulk-boundary correspondence by using index theory for Toeplitz operators \cite{KRSB02}.
The study of topological insulators, especially regarding its classification and the bulk-boundary correspondence for each of the ten Altland--Zirnbauer classes by using $K$-theory and index theory has been  much developed \cite{KRSB02,FM13,Thi16,BCR1,GS16,MT16b,PSB16,Thi16,BKR17,Kel17,Ku15,AMZ20}.
In \cite{Hayashi2}, three-dimensional (3-D) class A bulk periodic systems are studied on one piece of a lattice cut by two specific hyperplanes, which is a model for systems with corners.
Based on the index theory for quarter-plane Toeplitz operators \cite{Sim67, DH71,Pa90}, a topological invariant is defined assuming the spectral gap both on the bulk Hamiltonian and two half-space compressions of it.
This gapped topological invariant is topological in the sense that it does not change under continuous deformation of the bulk Hamiltonians unless the spectral gap of one of the two surfaces closes.
It is proved that, corresponding to this topology gapless corner states appear.
A construction of nontrivial examples by using two first-order topological insulators (of 2-D class A and 1-D class AIII) is also proposed.
Class AIII codimension-two systems are also studied through this method in \cite{Hayashi3} and,
as an application to HOTIs, the appearance of topological corner states in Benalcazar--Bernevig--Hughes' 2-D model \cite{BBH17a} is explained based on the chiral symmetry.
The construction of examples in \cite{Hayashi3} leads to a proposal of second-order semimetallic phase protected by the chiral symmetry \cite{OHN19}.

The purpose of this paper is to expand the results in \cite{Hayashi2} to all Altland--Zirnbauer classes and systems with corners of arbitrary codimension.
Since class A and class AIII systems (with codimension-two corners) were already discussed in \cite{Hayashi2,Hayashi3} by using complex $K$-theory, we focus on the remaining eight cases, for which we use real $K$-theory.
For our expansion, a basic scheme has already been well developed in the above previous studies, which we mainly follow:
some gapped Hamiltonian defines an element of a $KO$-group of a real \CA, and its relation with corner states are clarified by using index theory \cite{KRSB02,FM13,BCR1,GS16,Ku15,Thi16,BKR17,Kel17}.
Although many techniques have already been developed in studies of topological insulators, in our higher-codimensional cases, we still lack some basic results at the level of $K$-theory and index theory; hence, the first half of this paper is devoted to these $K$-theoretic preliminaries, that is, the computation of $KO$-groups for real \CA s associated with the quarter-plane Toeplitz extension and the computation of boundary maps for the $24$-term exact sequence of $KO$-theory associated with it, which are carried out in Sect.~\ref{Sect.3}.
Since the quarter-plane Toeplitz extension \cite{Pa90} is a key tool in our study of codimension-two corners, such a variant for Toeplitz operators associated with higher-codimensional corners should be clarified, which are carried out in Sect.~\ref{Sect.4}.
These variants of Toeplitz operators were discussed in \cite{DH71,Do73}, and the contents in Sect.~\ref{Sect.4} will be well-known to experts.
Since the author could not find an appropriate reference, especially concerning Theorem~\ref{multiexact} which will play a key role in Sect.~\ref{Sect.5}, the results are included for completeness.
Note that the idea there to use tensor products of the ordinary Toeplitz extension for the study of these variants is based on the work of Douglas--Howe \cite{DH71}, where these higher-codimensional generalizations are briefly mentioned.
The study of some gapped phases for systems with corners in Altland--Zirnbauer's classification is carried out in Sect.~\ref{Sect.5}.
In the framework of the one-particle approximation, we consider $n$-D systems with a codimension $k$ corner and take compressions of the bulk Hamiltonian onto infinite lattices with codimension $k-1$ corners\footnote{In standard terminologies, they will be called {\em edges}, {\em surfaces}, {\em hinges} or {\em edge of edges} depending on $n$ and $k$. In this paper, we may simply call them {\em corners} but state its codimensions.}
whose intersection makes the codimension $k$ corner.
We assume that they are gapped.
Note that, under this assumption, bulk, surfaces and corners up to codimension $k-1$ which constitute the codimension $k$ corner are also gapped.
For such a system, we define two topological invariants as elements of some $KO$-groups: one is defined for these gapped Hamiltonians while the other one is related to in-gap or gapless codimension $k$ corner states.
We then show a relation between these two which states that topologically protected corner states appear reflecting some gapped topology of the system.
We first study codimension-two cases (Sect.~\ref{Sect.5.1} to Sect.~\ref{Sect.5.4}) and then discuss higher-codimensional cases (Sect.~\ref{Sect.5.5}).
This distinction is made because many detailed results have been obtained for codimension-two cases by virtue of previous studies of quarter-plane Toeplitz operators \cite{Pa90,Ji95}
(the shape of the corner we discuss is more general than in higher-codimensional cases, and a relation between convex and concave corners is also obtained in \cite{Hayashi3}).
Based on these results, we propose a classification table for topological invariants related to corner states (Table~\ref{ptHOTI}).
Note that the codimension-one case of Table~\ref{ptHOTI} is Kitaev's table \cite{Kit09} and Table~\ref{ptHOTI} is also periodic by the Bott periodicity.
In order further to clarify a relation between our invariants and corner states, in Sect.~\ref{Sect.5.6}, we introduce $\Z$ or $\Z_2$-valued numerical corner invariants when the dimension of the corner is zero or one.
They are defined by (roughly speaking) counting the number of corner states.
A construction of examples is discussed in Sect.~\ref{Sect.5.7}.
As in \cite{Hayashi2}, this is given by using pairs of Hamiltonians of two lower-order topological insulators.
In the real classes, there are $64$ pairs of them and the results are collected in Table~\ref{producttalbe}.
By using this method, we can construct nontrivial examples of each entry of Table~\ref{ptHOTI}, starting from Hamiltonians of first-order topological insulators.
The corner invariant for the constructed Hamiltonian is expressed by corner (or edge) invariants of constituent two Hamiltonians.
This is given by using an exterior product of some $KO$-groups in general, though, as in \cite{Hayashi2,Hayashi3}, the formula for numerical invariants introduced in Sect.~\ref{Sect.5.6} is also included.
\begin{table}
\caption{Classification of (strong) topological invariants related to corner state in Altland--Zirnbauer (AZ) classification.
In this table, $n$ is the dimension of the bulk, and $k$ is the codimension of the corner.}
\label{ptHOTI}
\begin{tabular}{|c|ccc||cccccccc|c|} \hline
  	\multicolumn{4}{|c||}{Symmetry} & \multicolumn{8}{|c|}{$n-k \mod 8$} \\
    AZ & $\Theta$ & $\Xi$ & $\Pi$ & $0$ & $1$ & $2$ & $3$ & $4$ & $5$ & $6$ & $7$ \\ \hline \hline
    $\A$  & $0$ & $0$ & $0$ & $0$ & $\Z$ & $0$ & $\Z$ & $0$ & $\Z$ & $0$ & $\Z$ \\
    $\AIII$ & $0$ & $0$ & $1$ &$\Z$ & $0$ & $\Z$ & $0$ & $\Z$ & $0$ & $\Z$ & $0$ \\ \hline
    $\AI$ & $1$ & $0$ & $0$ & $0$ & $0$ & $0$ & $2\Z$ & $0$ & $\Z_2$ & $\Z_2$ & $\Z$ \\
    $\BDI$ & $1$ & $1$ & $1$ & $\Z$ & $0$ & $0$ & $0$ & $2\Z$ & $0$ & $\Z_2$ & $\Z_2$ \\
    $\DD$ & $0$ & $1$ & $0$ & $\Z_2$ & $\Z$ & $0$ & $0$ & $0$ & $2\Z$ & $0$ & $\Z_2$ \\
    $\DIII$ & $-1$ & $1$ & $1$ & $\Z_2$ & $\Z_2$ & $\Z$ & $0$ & $0$ & $0$ & $2\Z$ & $0$ \\
    $\AII$ & $-1$ & $0$ & $0$ & $0$ & $\Z_2$ & $\Z_2$ & $\Z$ & $0$ & $0$ & $0$ & $2\Z$ \\
    $\CII$ & $-1$ & $-1$ & $1$ & $2\Z$ & $0$ & $\Z_2$ & $\Z_2$ & $\Z$ & $0$ & $0$ & $0$ \\
    $\CC$ & $0$ & $-1$ & $0$ & $0$ & $2\Z$ & $0$ & $\Z_2$ & $\Z_2$ & $\Z$ & $0$ & $0$ \\ 
    $\CI$ & $1$ & $-1$ & $1$ & $0$ & $0$ & $2\Z$ & $0$ & $\Z_2$ & $\Z_2$ & $\Z$ & $0$ \\ \hline
\end{tabular}
\end{table}
For computations of $KO$-groups and classification of such gapped systems, we employ Boersema--Loring's unitary picture for $KO$-theory \cite{BL16} whose definitions are collected in Sect.~$2$.
Basic results for some Toeplitz operators are also included there.
In Appendix~\ref{Sect.A}, we revisit Atiyah--Singer's study of spaces of skew-adjoint Fredholm operators \cite{AS69} and collect necessary results from the viewpoint of Boersema--Loring's $K$-theory.
Definitions of some $\Z_2$-spaces, maps between them, expression of boundary maps of $24$-term exact sequences used in this paper are collected there.

Finally, let us point out a relation with our results and the current rapidly developing studies on HOTIs.
In \cite{Geier18}, the HOTIs are divided into two classes: {\em intrinsic} HOTIs, which basically originate from the bulk topology protected by a point group symmetry, and others {\em extrinsic} HOTIs.
Our study will be for extrinsic HOTIs since no point group symmetry is assumed and our classification table (Table~\ref{ptHOTI}) is consistent with that of Table~$1$ in \cite{Geier18}.

\section{Preliminaries}\label{Sect.2}
In this section, we collect the necessary results and notations.

\subsection{Boersema--Loring's $KO$-Groups via Unitary Elements}\label{Sect.2.1}
In this subsection, we collect Boersema--Loring's definition of $KO$-groups by using unitaries satisfying some symmetries \cite{BL16}.
The basics of real \CA s and $KO$-theory can be found in \cite{Goo82,Sch93}, for example.

A {\em \TA} is a pair $(\mathcal{A}, \tau)$ consisting of (complex) \CA \ $\mathcal{A}$ and an anti-automorphism\footnote{i.e., a complex linear automorphism of $A$ that preserves $\ast$ and satisfies $\tau(ab) = \tau(b)\tau(a)$.} $\tau$ of $\mathcal{A}$ satisfying $\tau^2=1$.
We call $\tau$ the {\em transposition} and write $a^\tau$ for $\tau(a)$.
There is a category equivalence between the category of \TA s and the category of real \CA s:
for a \TA \ $(A,\tau)$, the corresponding real \CA \ is $\mathcal{A}^\tau = \{ a \in \mathcal{A} \ | \ a^\tau = a^* \}$, and its inverse is given by the complexification.
A {\em real structure} on a (complex) \CA \ $A$ is an antilinear $*$-automorphism $\fr$ satisfying $\fr^2 = 1$.
For a real structure $\fr$, there is an associated transposition $\tau$ given by $\tau(a) = \fr(a^*)$, which gives a one-to-one correspondence between transpositions and real structures on the \CA\footnote{Boersema--Loring called $\tau$ the real structure in \cite{BL16}. In this paper, we distinguish these two since the antilinear structure naturally appears in our application. We call $\tau$ the transposition following \cite{Kel17}.}.
We extend the transposition $\tau$ on $\mathcal{A}$ to the transposition (for which we simply write $\tau$) on the matrix algebra $M_n(\mathcal{A})$ by $(a_{ij})^\tau = (a_{ji}^\tau)$ where $a_{ij} \in \mathcal{A}$ and $1 \leq i,j \leq n$.
This induces a transposition $\tau_\K$ on $\K \otimes \mathcal{A}$ where $\K = \K(\mathcal{V})$ is the \CA \ of compact operators on a separable complex Hilbert space $\mathcal{V}$.
Let $\sharp \otimes \tau$ be a transposition on $M_2(\mathcal{A})$ defined by\footnote{
For notations of the transpositions introduced here, we follow \cite{BL16}.}
\begin{equation*}
\left(
    \begin{array}{cc}
           a_{11}&a_{12}\\
           a_{21}&a_{22}
    \end{array}
\right)^{\sharp \otimes \tau}
	=
\left(
    \begin{array}{cc}
           a_{22}^\tau & -a_{12}^\tau \\
           -a_{21}^\tau & a_{11}^\tau
    \end{array}
\right).
\end{equation*}
If we identify the quaternions $\mathbb{H}$ with $\C^2$ by $x+yj \mapsto (x,y)$, the left multiplication by $j$ corresponds to $j(x,y) = (-\bar{y},\bar{x})$.
Then, we have $\sharp \otimes \id = \Ad_j \circ \ast$ where $\ast$ denotes the operation of taking conjugation of matrices and the \TA \ $(M_2(\C), \sharp \otimes \id)$ corresponds to the real \CA \ $\mathbb{H}$ of quaternions.
We extend this transposition to $M_{2n}(\mathcal{A})$ by $(b_{ij})^{\sharp \otimes \tau} = (b_{ji}^{\sharp \otimes \tau})$ where $1 \leq i,j \leq n$ and $b_{ij} \in M_2(\mathcal{A})$.
On $M_{2n}(\mathcal{A})$, we also consider a transposition $\widetilde{\sharp} \otimes \tau$ defined by
\begin{equation*}
\left(
    \begin{array}{cc}
           c_{11}&c_{12}\\
           c_{21}&c_{22}
    \end{array}
\right)^{\widetilde{\sharp} \otimes \tau}
	=
\left(
    \begin{array}{cc}
           c_{22}^\tau & -c_{12}^\tau \\
           -c_{21}^\tau & c_{11}^\tau
    \end{array}
\right),
\end{equation*}
where $c_{ij} \in M_n(\mathcal{A})$.
For an $m \times m$ matrix $X$, we write $X_n$ for the $mn \times mn$ block diagonal matrix $\diag(X, \ldots, X)$.
For example, we write $1_n$ for the $n \times n$ diagonal matrix $\diag(1, \ldots, 1)$.

\begin{definition}[Boersema--Loring \cite{BL16}]
Let $(\mathcal{A},\tau)$ be a unital \TA.
For $i = -1, 0, \ldots, 6$, let $n_i$ be a positive integer, $\mathcal{R}_i$ be a relation and $I^{(i)}$ be a matrix, as indicated in Table~\ref{BL}.
Let $U^{(i)}_k(\mathcal{A},\tau)$ be the set of all unitaries in $M_{n_i \cdot k}(\mathcal{A})$ satisfying the relation $\mathcal{R}_i$.
On the set $U^{(i)}_\infty(\mathcal{A},\tau) = \cup_{k=1}^\infty U^{(i)}_k(\mathcal{A},\tau)$, we consider the equivalence relation $\sim_i$ generated by homotopy and stabilization given by $I^{(i)}$.
We define $KO_i(\mathcal{A}, \tau) = U^{(i)}_\infty(\mathcal{A},\tau) / \sim_i$ which is a group by the binary operation given by $[u] + [v] = [\diag(u,v)]$.
\end{definition}
For a nonunital \TA \ $(A,\tau)$, the $i$-th $KO$-group $KO_i(\mathcal{A}, \tau)$ is defined as the kernel of $\lambda_* \colon KO_i(\tilde{\mathcal{A}}, \tau) \to KO_i(\C, \id)$, where $\tilde{\mathcal{A}}$ is the unitization of $\mathcal{A}$ and $\lambda \colon \tilde{\mathcal{A}} \to \C$ is the natural projection.
In \cite{BL16}, they also describe the boundary maps of the $24$-term exact sequence for $KO$-theory associated with a short exact sequence of \TA s.
In Appendix~\ref{Sect.B}, we discuss an alternative description for some of them through exponentials.
\begin{table}
\caption{Boersema--Loring's unitary picture for $KO$-theory \cite{BL16}}
\label{BL}
\centering
\begin{tabular}{|c|c|c|c|c|}
\hline
$i$ & $KO$-group & $n_i$ & $\mathcal{R}_i$ & $I^{(i)}$ \\ \hline \hline
$-1$ & $KO_{-1}(\mathcal{A}, \tau)$ & $1$ & $u^\tau = u$ & $1$ \\ \hline
$0$ & $KO_0(\mathcal{A}, \tau)$ & $2$ & $u = u^*$, $u^{\tau} = u^*$ & $\diag(1, -1)$ \\ \hline
$1$ & $KO_1(\mathcal{A}, \tau)$ & $1$ & $u^{\tau} = u^*$ & $1$ \\ \hline
$2$ & $KO_2(\mathcal{A}, \tau)$ & $2$ & $u = u^*$, $u^\tau = -u$ & $\left( \begin{array}{cc} 0&i \cdot 1 \\ -i \cdot 1&0 \end{array} \right)$ \\ \hline
$3$ & $KO_3(\mathcal{A}, \tau)$ & $2$ & $u^{\sharp \otimes \tau} = u $ & $1_2$ \\ \hline
$4$ & $KO_4(\mathcal{A}, \tau)$ & $4$ & $u = u^*$, $u^{\sharp \otimes \tau} = u^*$ & $\diag(1_2, -1_2)$ \\ \hline
$5$ & $KO_5(\mathcal{A}, \tau)$ & $2$ & $u^{\sharp \otimes \tau} = u^*$ & $1_2$ \\ \hline
$6$ & $KO_6(\mathcal{A}, \tau)$ & $2$ & $u = u^*$, $u^{\sharp \otimes \tau} = -u$ & $\left( \begin{array}{cc} 0&i \cdot 1 \\ -i \cdot 1&0 \end{array} \right)$ \\ \hline
\end{tabular}
\end{table}

\subsection{Toeplitz Operators}\label{Sect.2.2}
In this subsection, we collect the definitions and basic results for some Toeplitz operators used in this paper \cite{Do73,Pa90}.

Let $\T$ be the unit circle in the complex plane $\C$, and let $c$ be the complex conjugation on $\C$, that is, $c(z) = \bar{z}$.
Let $n$ be a positive integer.
On the $n$-dimensional torus $\T^n$, we consider an involution $\zeta$ defined as the $n$-fold product of $c$.
This induces a transposition $\tau_{\T}$ on $C(\T^n)$ by $(\tau_{\T} f)(t) = f(\zeta(t))$.
Let $\Z_{\geq 0}$ be the set of non-negative integers and $P_n$ be the orthogonal projection of $l^2(\Z^n)$ onto $l^2((\Z_{\geq 0})^n)$.
For a continuous function $f \colon \T^n \to \C$, let $M_f$ be the multiplication operator on $l^2(\Z^n)$ generated by $f$.
We consider the operator $P_n M_f P_n$ on $l^2((\Z_{\geq 0})^n)$, which is the Toeplitz operator associated with the subsemigroup $(\Z_{\geq 0})^n$ of $\Z^n$ of symbol $f$.
We write $\TT^n$ for the $C^*$-subalgebra of $\B(l^2((\Z_{\geq 0})^n))$ generated by these Toeplitz operators.
The algebra $\TT^1$ is the ordinary Toeplitz algebra and we simply write $\TT$.
Note that the algebra $\TT^n$ is isomorphic to the $n$-fold tensor product of $\TT$.
The complex conjugation $c$ on $\C$ induces an antiunitary operator\footnote{An operator $A$ on a complex Hilbert space $\mathcal{V}$ is called the {\em antiunitary} operator if $A$ is an antilinear bijection on $\mathcal{V}$ satisfying $\langle A v, A w \rangle = \overline{\langle v,w \rangle}$ for any $v$ and $w$ in $\mathcal{V}$.} of order two on the Hilbert space $l^2(\Z^n)$ by the pointwise operation, for which we also write $c$.
This induces a real structure $\mathfrak{c}$ on $\B(l^2((\Z_{\geq 0})^n))$  by $\mathfrak{c}(a) = \Ad_c(a) = c a c^*$.
We write $\tau_\TT$ for the transposition on $\TT^n$ given by its restriction onto $\TT^n$.

We next focus on the case of $n=2$.
We consider the Hilbert space $l^2(\Z^2)$ and take its orthonormal basis $\{ \delta_{m,n} \ | \ (m,n) \in \Z^2 \}$, where $\delta_{m,n}$ is the characteristic function of the point $(m,n)$ on $\Z^2$.
When $f \in C(\T^2)$ is given by $f(z_1,z_2) = z_1^m z_2^n$, we write $M_{m,n}$ for the multiplication operator $M_f$.
Let $\alpha < \beta$ be real numbers, and let $\HHa$, $\HHb$, $\HHab$ and $\HHs$ be closed subspaces of $l^2(\Z^2)$ spanned by
$\{ \delta_{m,n} \ | -\alpha m + n \geq 0 \}$,
$\{ \delta_{m,n} \ | -\beta m + n \leq 0 \}$,
$\{ \delta_{m,n} \ | -\alpha m + n \geq 0 \ \text{and}\ -\beta m + n \leq 0 \}$, and
$\{ \delta_{m,n} \ | -\alpha m + n \geq 0 \ \text{or}\ -\beta m + n \leq 0 \}$, respectively.
In the following, we may take $\alpha = -\infty$ or $\beta = \infty$, but not both.
Let $\Pa$, $\Pb$, $\Pab$ and $\Ps$ be the orthogonal projection of $l^2(\Z^2)$ onto $\HHa$, $\HHb$, $\HHab$ and $\HHs$, respectively.
For $f \in C(\T^2)$, the operators $\Pa M_f \Pa$ on $\HHa$ and $\Pb M_f \Pb$ on $\HHb$ are called {\em half-plane Toeplitz operators}.
The operator $\Pab M_f \Pab$ on $\HHab$ is called the {\em quarter-plane Toeplitz operator}, and $\Ps M_f \Ps$ on $\HHs$ is its concave corner analogue.
We write $\TTa$ and $\TTb$ for $C^*$-algebras generated by these half-plane Toeplitz operators and $\TTab$ and $\TTs$ for $C^*$-algebras generated by the quarter-plane and concave corner Toeplitz operators, respectively.
There are $*$-homomorphisms $\sigmaa \colon \TTa \to C(\T^2)$ and $\sigmab \colon \TTb \to C(\T^2)$, which map $\Pa M_f \Pa$ and $\Pb M_f \Pb$ to the symbol $f$, respectively.
We define the \CA \ $\Sab$ as a pullback \CA \ of these two $*$-homomorphisms.
The real structure $c$ on $\HH = l^2(\Z^2)$ induces real structures $\mathfrak{c}$ on $\TTa$, $\TTb$, $\TTab$, $\TTs$, and $\Sab$
and thus induces transpositions $\taua$, $\taub$, $\tauab$, $\taus$ and $\tau_\mathcal{S}$ on $\TTa$, $\TTb$, $\TTab$, $\TTs$ and $\Sab$, respectively.
For transpositions, we may simply write $\tau$ when it is clear from the context.
The maps $\sigmaa$ and $\sigmab$ preserve the real structures and we have the following pull-back diagram:
\begin{equation}\label{diag2}
\vcenter{
\xymatrix{
(\Sab, \tau_\mathcal{S}) \ar[r]^{p^\beta} \ar[d]_{p^\alpha} & (\TTb, \taub) \ar[d]^{\sigmab} \\
(\TTa, \taua) \ar[r]^{\sigmaa} & (C(\T^2),\tau_{\T})}}
\end{equation}
We write $\sigma$ for the composition $\sigmaa \circ p^\alpha = \sigmab \circ p^\beta$.
Let $\hat{\gamma}$ be a $*$-homomorphism from $\TTab$ to $\Sab$ which maps $\Pab M_f \Pab$ to the pair $(\Pa M_f \Pa, \Pb M_f \Pb)$.
This map $\hat{\gamma}$ preserves the real structures, and there is the following short exact sequence of \TA s (Park \cite{Pa90}):
\begin{equation}\label{seq1}
0 \to(\K(\HHab), \tau_\K) \to (\TTab, \tauab) \overset{\hat{\gamma}}{\to} (\Sab, \tau_\mathcal{S}) \to 0,
\end{equation}
where the map from $(\K(\HHab), \tau_\K)$ to $(\TTab, \tauab)$ is the inclusion.
Its concave corner analogue is studied in \cite{Hayashi3} and the following exact sequence is obtained:
\begin{equation}\label{seq2}
0 \to (\K(\HHs),\tau_\K) \to (\TTs,\taus) \overset{\check{\gamma}}{\to} (\Sab, \tau_\mathcal{S}) \to 0,
\end{equation}
where $\check{\gamma}$ is a $*$-homomorphism mapping $\Ps \hspace{-0.15mm} M_f \hspace{-0.15mm} \Ps$ to $(\hspace{-0.1mm} \Pa \hspace{-0.1mm} M_f \hspace{-0.1mm} \Pa \hspace{-0.15mm}, \hspace{-0.15mm} \Pb \hspace{-0.15mm} M_f \hspace{-0.15mm} \Pb \hspace{-0.15mm})$.

\section{$KO$-Groups of \CA s Associated with Half-Plane and Quarter-Plane Toeplitz Operators}\label{Sect.3}
In this section, the $KO$-theory for half-plane and quarter-plane Toeplitz operators are discussed.
In Sect.~\ref{Sect.3.1}, $KO$-groups of the half-plane Toeplitz algebra is computed.
Quarter-plane Toeplitz operators are discussed in the following sections, and the $KO$-groups of the \TA \ $(\Sab, \tau_\mathcal{S})$ are computed in Sect.~\ref{Sect.3.2}.
In Sect~\ref{Sect.3.3}, the boundary maps of the $24$-term exact sequence for $KO$-theory associated with the sequence (\ref{seq1}) are discussed and the $KO$-groups of the quarter-plane Toeplitz algebra $(\TTab,\tauab)$ are computed.

\subsection{$KO$-Groups of $(\TTa, \taua)$}\label{Sect.3.1}
We compute the $KO$-groups of the \TA \ $(\TTa, \taua)$.
The discussion is divided into two cases whether $\alpha$ is rational (or $-\infty$) or irrational.

We first consider the case when $\alpha$ is a rational number or $- \infty$.
When $\alpha \in \Q$, we write $\alpha = \frac{p}{q}$ where $p$ and $q$ are relatively prime integers and $q$ is positive.
Let $m$ and $n$ be integers such that $-pm + qn = 1$ and let
\begin{equation}\label{SL}
\Gamma = \left(
    \begin{array}{cc}
           n&-m\\
           -p&q
    \end{array}
\right) \in SL(2,\Z).
\end{equation}
Then, the action of $\Gamma$ on $\Z^2$ induces the Hilbert space isomorphism $\HHa \cong \HH^{0}$ and an isomorphism of \TA s $(\TTa, \taua) \cong (\TT^0, \tau_0)$.
Thus, the \TA \ $(\TTa, \taua)$ is isomorphic to $(\TT, \tau_\TT) \otimes (C(\T), \tau_\T)$,
and its $KO$-groups are computed as
$KO_i(\TTa, \taua) \cong KO_i(C(\T), \tau_\T) \cong KO_i(\C, \id) \oplus KO_{i-1}(\C,\id)$.
For the first isomorphism, see Proposition~$1.5.1$ of \cite{Sch93}.
Generators of the group $KO_i(C(\T), \tau_\T)$ are obtained in Example~9.2 of \cite{BL16}, and the unital $*$-homomorphism $\iota \colon \C \to \TT$ induces an isomorphism $(\id \otimes \iota)_* \colon KO_i(C(\T), \tau_\T) \to KO_i((C(\T), \tau_\T) \otimes (\TT, \tau_\TT)) \cong KO_i(\TT^0, \tau_0)$.
Combined with them, $KO$-group $KO_i(\TTa, \taua)$ and its generators are given as follows.
\begin{itemize}
\item $KO_0(\TTa, \taua) \cong \Z$ and its generator is $[1_2]$.
\item $KO_1(\TTa, \taua) \cong \Z_2 \oplus \Z$. A generator of $\Z_2$ is $[-1]$ and that of $\Z$ is $[\Pa M_{q,p} \Pa]$.
\item $KO_2(\TTa, \taua) \cong (\Z_2)^2$. A generator of one $\Z_2$ is $[-I^{(2)}]$,
and that of another $\Z_2$ is
$\left[ \left(
    \begin{array}{cc}
           0&i \Pa M_{q,p}\Pa \\
           -i \Pa M_{-q,-p} \Pa&0
    \end{array}
\right)\right]$.
\item $KO_3(\TTa, \taua) \cong \Z_2$ and its generator is $[\diag(\Pa M_{q,p}\Pa,\Pa M_{-q,-p}\Pa)]$.
\item $KO_4(\TTa, \taua) \cong \Z$ and its generator is $[1_4]$.
\item $KO_5(\TTa, \taua) \cong \Z$ and its generator is
$[\diag(\Pa M_{q,p}\Pa,\Pa M_{q,p}\Pa)]$.
\item $KO_6(\TTa, \taua) = KO_{-1}(\TTa, \taua) = 0.$
\end{itemize}
The case of $\alpha = -\infty$ is computed similarly, and its generators are given by replacing $p$ and $q$ above with $-1$ and $0$, respectively.

We next consider the cases of irrational $\alpha$.
In this case, complex $K$-groups of $\TTa$ are computed by Ji--Kaminker and Xia in \cite{JK88, Xia88}.
\begin{lemma}\label{lemma3.1}
For irrational $\alpha$ and for each $i$, we have $KO_i(\TTa, \taua) \cong KO_i(\C, \id)$, where the isomorphism is given by $\lambda^\alpha_*$.
\end{lemma}
\begin{proof}
As for complex $K$-groups, we have $K_0(\TTa) = \Z$ and $K_1(\TTa) = 0$ by \cite{JK88,Xia88}.
We consider a split $*$-homomorphism of \TA s $\lambda^\alpha \colon (\TTa, \taua) \to (\C, \id)$ given by the composition of $\sigmaa \colon (\TTa,\taua) \to (C(\T^2), \tau_{\T})$ and the pull-back onto a fixed point of the involution $\zeta$ on $\T^2$.
Let $\rTTa = \Ker \lambda^\alpha$.
By the\vspace{-0.5mm}
six-term exact sequence associated with the extension $0 \to \rTTa \to \TTa \overset{\lambda^\alpha}{\to} \C \to 0$, complex $K$-groups of $\TTa_0$ are trivial.
For a \TA \ $(\mathcal{A}, \tau)$, it follows from Theorem~$1.12$, Proposition~$1.15$ and Theorem~$1.18$ of \cite{Boe02} that $KO_*(\mathcal{A},\tau) = 0$ if and only if $K_*(\mathcal{A})=0$.
Therefore, $KO_*(\TTa_0, \taua) = 0$.
The result follows from the $24$-term exact sequence \vspace{-0.5mm}
 of $KO$-theory for \TA s associated with the short exact sequence
$0 \to (\rTTa, \taua) \to (\TTa, \taua) \overset{\lambda^\alpha}{\to} (\C,\id) \to 0$.
\end{proof}

\subsection{$KO$-Groups of $(\Sab, \tau_\mathcal{S})$}\label{Sect.3.2}
In this subsection, we compute the $KO$-groups of the \TA \ $(\Sab, \tau_\mathcal{S})$.
The basic tool is the following Mayer--Vietoris exact sequence associated with the pull-back diagram (\ref{diag2}) (see Theorem~$1.4.15$ of \cite{Sch93}, for example):
\vspace{-2mm}
\begin{equation}
\label{MV}
{\small
\vcenter{
\xymatrix{
	& \cdots \ar[r] & KO_{i+1}(C(\T^2),\tau_{\T}) \ar[dll]_{\partial_{i+1}} \\
KO_i(\Sab,\tau_\mathcal{S}) \ar[r]_{(p^\alpha_*,p^\beta_*) \hspace{9mm}} & KO_i(\TTa,\taua) \oplus KO_i(\TTb,\taub) \ar[r]^{\hspace{7mm} \sigmab_* - \sigmaa_*} & KO_i(C(\T^2),\tau_{\T}) \ar[dll]_{\partial_{i}} \\
KO_{i-1}(\Sab,\tau_\mathcal{S}) & \cdots \ar[r] &
}}}
\end{equation}
As in \cite{Pa90}, the computation of the group $KO_*(\Sab, \tau_\mathcal{S})$ is divided into three cases corresponding to whether $\alpha$ and $\beta$ are rational (or $\pm \infty$) or irrational.
As in Sect.~\ref{Sect.3}, we have a unital $*$-homomorphism $\lambda^\alpha \circ p^\alpha \colon (\Sab,\tau_\mathcal{S}) \to (\C, \id)$ which splits.
Correspondingly, the $KO$-group $KO_*(\Sab,\tau_\mathcal{S})$ have a direct summand corresponding to $KO_*(\C,\id)$.
Noting this, these $KO$-groups are computed by Lemma~\ref{lemma3.1} and the sequence (\ref{MV}) when at least one of $\alpha$ and $\beta$ is irrational.
The results are collected in Tables~\ref{KOSab2} and \ref{KOSab3}.
In the rest of this subsection, we focus on the cases when both $\alpha$ and $\beta$ are rational (or $\pm \infty$).

When $\alpha, \beta \in \Q$, we write $\alpha = \frac{p}{q}$ and $\beta = \frac{r}{s}$ by using mutually prime integers where $q$ and $s$ are positive.
In the following discussion, the case of $\alpha = -\infty$ corresponds to the case where $p = -1$ and $q = 0$, and the case of $\beta = +\infty$ corresponds to the case where $r = 1$ and $s = 0$.
By using the action of $\Gamma \in SL(2, \Z)$ in (\ref{SL}) on $\Z^2$, there are isomorphisms $(\TTa, \taua) \cong (\TT^0, \tau_0)$ and $(\TTb, \taub) \cong (\TTg, \taug)$, where $\gamma = \frac{t}{u}$ for $u = ns - mr$ and $t = -ps + qr$.
Note that  $t$ is positive since $\alpha < \beta$.
We have the following commutative diagram:
\vspace{-1mm}
\[\xymatrix{
	KO_i(\TTa,\taua) \oplus KO_i(\TTb,\taub) \ar[r]^{\hspace{10mm} \sigmab_* - \sigmaa_*} \ar[d]_\cong & KO_i(C(\T^2),\tau_{\T}) \ar[d]^\cong\\
	KO_i(\TTz,\tauz) \oplus KO_i(\TTg,\taug) \ar[r]^{\hspace{10mm} \sigma^\gamma_* - \sigma^0_*} & KO_i(C(\T^2),\tau_{\T})
}\]
where the vertical isomorphisms are induced by the action of $\Gamma$.
In the following, we discuss the lower part of the diagram, which is enough for our purpose since the isomorphism $KO_i(\Sab, \tau_\mathcal{S}) \cong KO_i(\Szg, \tau_\mathcal{S})$ is also induced.
We write $\varphi_i$ for the above map $\sigma^\gamma_* - \sigma^0_*$.
By the exact sequence (\ref{MV}), we have the following short exact sequence.
\begin{equation}\label{varphi}
	0 \to \Coker (\varphi_{i+1}) \to KO_i(\Szg, \tau_\mathcal{S}) \to \Ker (\varphi_i) \to 0.
\end{equation}
We first compute kernels and cokernels of $\varphi_i$.
Cases for $i = -1, 0, 4, 6$ is easy, thus we consider the other cases.

When $i = 1$, groups $KO_1(\TTz,\tauz)$ and $KO_1(\TTg,\taug)$ are both isomorphic to $\Z_2 \oplus \Z$.
The $\Z_2$ direct summand is generated by $[-1]$, and the other $\Z$ direct summand is generated by $[\Pz M_{1,0} \Pz]$ and $[\Pg M_{u,t} \Pg]$, respectively. They map to $[M_{1,0}]$ and $[M_{u,t}]$ in $KO_1(C(\T^2),\tau_{\T})$ by $\sigmaz_*$ and $\sigmag_*$, respectively.
We have $KO_1(C(\T^2),\tau_{\T}) \cong \Z_2 \oplus \Z^2$, where the $\Z_2$ direct summand is generated by $[-1]$.
For $(m, n) \in \Z^2$, the element $[M_{m,n}] \in KO_1(C(\T^2),\tau_{\T})$ corresponds to $(0, m, n) \in \Z_2 \oplus \Z^2$.
Therefore, $\Ker(\varphi_1) \cong \Z_2$ which is generated by $([-1],[-1])$, and\footnote{When $\alpha = - \infty$ and $\beta \in \Q$, we have $\Coker (\varphi_1) \cong \Z_s$. This is the case when $p = -1$ and $q =0$ and $t = -ps + qr = s$ in this case. A similar remark also holds for $i = 2,3,5$.} $\Coker(\varphi_1) \cong \Z_t$.

We next consider the case of $i =2$.
We have $KO_2(C(\T^2),\tau_{\T}) \cong \Z_2 \oplus (\Z_2)^2 \oplus \Z$, where the first $\Z_2$ direct summand is generated by $[-I^{(2)}]$.
For $(m, n) \in \Z^2$, the element
$\left[ \left(
    \begin{array}{cc}
           0&i M_{m,n} \\
           -i M_{-m,-n}&0
    \end{array}
\right)\right]$ in $KO_2(C(\T^2),\tau_{\T})$
corresponds to $(0, m \bmod 2, n \bmod 2, 0) \in \Z_2 \oplus (\Z_2)^2 \oplus \Z$ (Example~9.2 of \cite{BL16}).
The groups $KO_2(\TTz,\tauz)$ and $KO_2(\TTg,\taug)$ and their generators are obtained in Sect.~\ref{Sect.3} and we have
\begin{equation*}
\Ker (\varphi_2) \cong
\left\{
\begin{aligned}
(\Z_2)^2 & \hspace{3mm} \text{when $t$ is even,}\\
\Z_2  \hspace{3mm} & \hspace{3mm} \text{when $t$ is odd,}
\end{aligned}
\right.
\ \ \
\Coker (\varphi_2) \cong
\left\{
\begin{aligned}
\Z_2 \oplus \Z & \hspace{3mm} \text{when $t$ is even,}\\
\Z  \hspace{3.5mm} & \hspace{3mm} \text{when $t$ is odd.}
\end{aligned}
\right.
\end{equation*}

When $i = 3$, we have $KO_3(C(\T^2),\tau_{\T}) \cong (\Z_2)^3$.
For $(m, n) \in \Z^2$, the element $[\diag(M_{m,n}, M_{-m,-n})] \in KO_3(C(\T^2),\tau_{\T})$ corresponds to $(m \bmod 2, n \bmod 2, 0) \in (\Z_2)^3$.
By Sect.~\ref{Sect.3}, we have
\begin{equation*}
\Ker (\varphi_3) \cong
\left\{
\begin{aligned}
\Z_2 & \hspace{3mm} \text{when $t$ is even,}\\
0  \hspace{1mm} & \hspace{3mm} \text{when $t$ is odd,}
\end{aligned}
\right.
\ \ \
\Coker (\varphi_3) \cong
\left\{
\begin{aligned}
(\Z_2)^2 & \hspace{3mm} \text{when $t$ is even,}\\
\Z_2  \hspace{2mm} & \hspace{3mm} \text{when $t$ is odd.}
\end{aligned}
\right.
\end{equation*}

When $i=5$, we have $KO_5(C(\T^2),\tau_{\T}) \cong \Z^2$.
For $(m, n) \in \Z^2$, the element $[\diag(M_{m,n}, M_{m,n})] \in KO_5(C(\T^2),\tau_{\T})$ corresponds to $(m,n) \in \Z^2$.
By Sect.~\ref{Sect.3}, we have $\Ker(\varphi_5) = 0$ and $\Coker(\varphi_5) = \Z_t$.

Combined with the above computation and the exact sequence (\ref{varphi}), $KO$-group $KO_i(\Sab, \tau_\mathcal{S})$ are computed, though some complication appears when $i=2,3$.
We discuss these two cases in the following subsections.

\subsubsection{The Group $KO_2(\Sab, \tau_\mathcal{S})$}\label{Sect.3.2.1}
We compute the group $KO_2(\Szg, \tau_\mathcal{S})$, which is isomorphic to $KO_2(\Sab, \tau_\mathcal{S})$.
The computation is divided into two cases depending on whether $t$ is even or odd.
Note that $u$ is odd when $t$ is even since $r$ and $s$ are mutually prime.

When $t$ is odd, 
$\Ker(\varphi_2) \cong \Z_2$ is generated by $([-I^{(2)}],[-I^{(2)}])$ and the sequence (\ref{varphi}) splits.
Therefore, $KO_2(\Sab, \tau_\mathcal{S}) \cong (\Z_2)^2$.

We next discuss the cases of even $t$.
In this case, both of the kernel and the cokernel of $\varphi_3$ are isomorphic to $(\Z_2)^2$.
Let $\widetilde{KO}_2(\Szg, \tau_\mathcal{S})$ be the kernel of the map $\lambda^\alpha_* \circ p^\alpha_* \colon KO_2(\Sab,\tau_\mathcal{S}) \to KO_2(\C, \id) \cong \Z_2$ which splits.
Then, the sequence (\ref{varphi}) reduces to the following extension:
\begin{equation}\label{KO2ext}
	0 \to (\Z_2)^2 \to \widetilde{KO}_2(\Szg, \tau_\mathcal{S}) \to \Z_2 \to 0.
\end{equation}
In the following, we show that this sequence (\ref{KO2ext}) splits.
We find a lift of the generator of $\Z_2$ in $\widetilde{KO}_2(\Szg, \tau_\mathcal{S})$ and show this lift has order two.
For $(m,n) \in \Z^2$ and $\kappa = 0$ and $\gamma$, we write $T_{m,n}^\kappa$ for $P^\kappa M_{m,n} P^\kappa$, and let $Q$ be the projection $T^\gamma_{u,0} T^\gamma_{-u,0}$.
Note that $1-Q$ is the projection onto the closed subspace spanned by $\{ \delta_{m,n} \ | \ 0 \leq \gamma m - n < t \ \}$.
For $j=1, \ldots, t$, let $P_j$ be a projection in $\TTg$, defined inductively as follows:
\begin{equation*}
	P_1 = (1-Q)M_{0,-t+1}(1-Q)M_{0,t-1}(1-Q),
\end{equation*}
\vspace{-5mm}
\begin{equation*}
	P_j = (1-Q)M_{0,-t+j}(1-Q)M_{0,t-j}(1-Q) - \sum^{j-1}_{k=1}P_k.
\end{equation*}
Specifically, $P_j$ is the orthogonal projection of $\HH^\gamma$ onto the closed subspace spanned by $\{ \delta_{n, tn-j+1} \}_{n \in \Z}$.
\begin{figure}
\centering
\includegraphics[width=7.0cm]{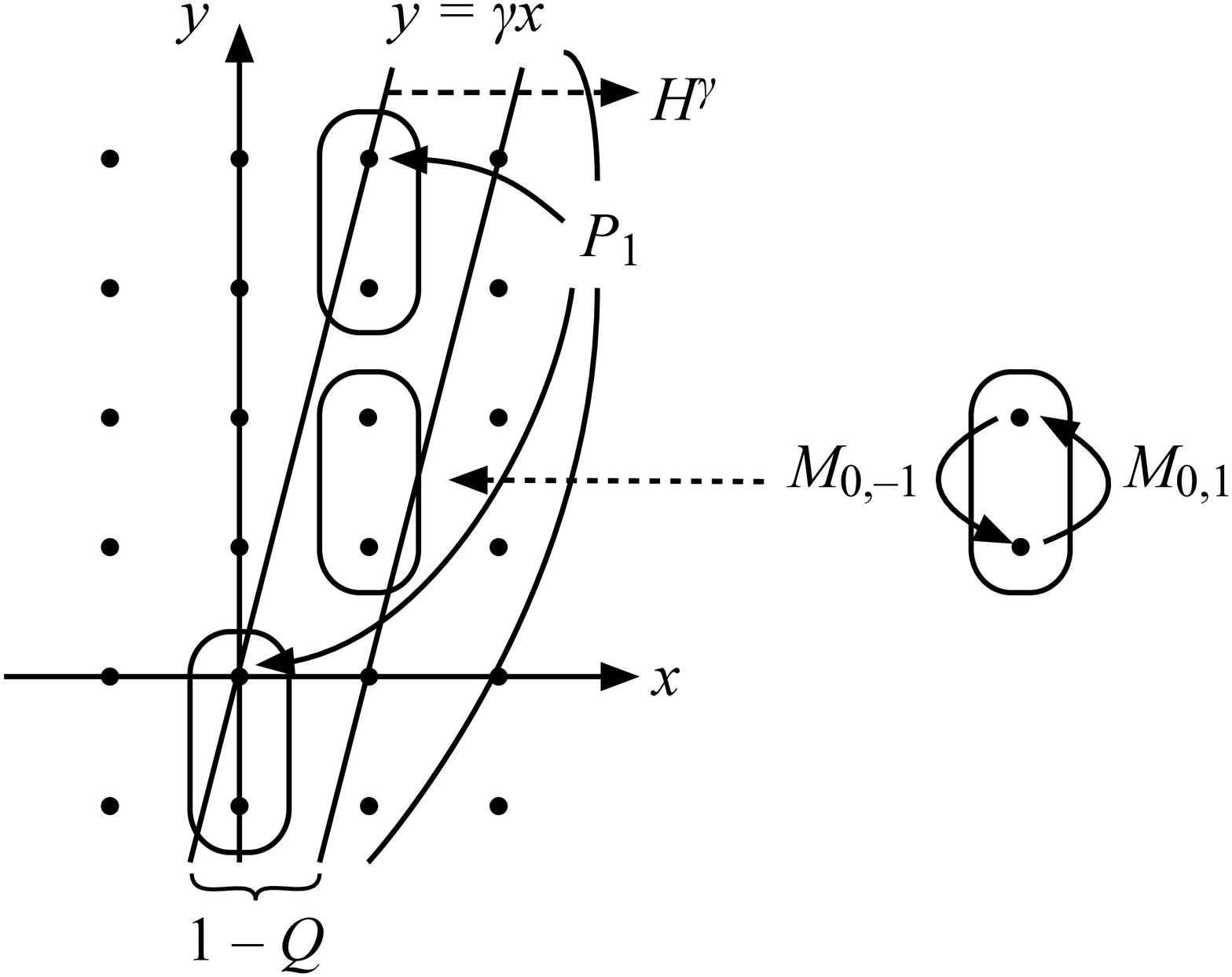}
\caption{The case of $u=1$ and $t = 4$. $1-Q$ is the projection onto the closed subspace corresponding to lattice points in between two lines (lattice points on the line $y = \gamma x$ are included, while that on $y = \gamma(x-1)$ are not). $P_j$ is the projection onto the closed subspace spanned by $\{\delta_{n, 4n-j+1} \ | \ n \in \Z \}$. $s_j$ interchanges two points in a pair up to the sign.}
\label{even}
\end{figure}
Note that $\sum^{t}_{j=1} P_j = 1- Q$.
For odd $j=1, 3, \ldots, t-1$, let $s_j = P_{j} M_{0,1} P_{j+1} - P_{j+1} M_{0,-1} P_{j}$ and $s = \sum_{j=1, \text{odd}}^{t-1} s_j$.
The element $s$ satisfies the relations
$\rm(\hspace{.18em}i\hspace{.18em})$ $s^* = -s$,
$\rm(\hspace{.08em}ii\hspace{.08em})$ $s^\tau = -s$,
$\rm(i\hspace{-.08em}i\hspace{-.08em}i)$ $s^2 = -1+Q$,
$\rm(i\hspace{-.08em}v\hspace{-.06em})$ $Q s = s Q = 0$ and
$\rm(\hspace{.06em}v\hspace{.06em})$ $s T^\gamma_{u,0} = T^\gamma_{-u,0} s = 0$.
Note that $\sigmag(s)=0$ since $\sigmag(1-Q)=0$.
We first consider the following elements:
\begin{equation*}
a =
\left(
    \begin{array}{cc}
           0 & i \cdot 1_{\TTz}  \\
           -i \cdot 1_{\TTz} & 0 \\
    \end{array}
\right) \in M_2(\TTz),
\ \ \
b_\pm =
\left(
    \begin{array}{cc}
           \pm i s &iQ  \\
           -iQ & i s \\
    \end{array}
\right) \in M_2(\TTg),
\end{equation*}
where the double-sign corresponds.
Elements $a$ and $b_\pm$ are self-adjoint unitaries satisfying $a^\tau = -a$ and $b_\pm^\tau = -b_\pm$, and pairs $(a, b_\pm)$ are elements of $M_2(\Szg)$; therefore, they define the elements of $KO_2(\Szg, \tau_\mathcal{S})$.

\begin{lemma}\label{lemtriv}
As elements of $KO_2(\Szg, \tau_\mathcal{S})$, we have
$[(a, b_+)]= [(a, b_-)] = 0$.
\end{lemma}
\begin{proof}
We first show that $[(a, b_+)] = 0$.
For $j=1, 3, \ldots, t-1$, let $r_j = P_{j} M_{0,1} P_{j+1} + P_{j+1} M_{0,-1} P_{j}$ and $r = \sum_{j=1, \text{odd}}^{t-1} r_j$.
The element $r$ satisfies
$\rm(\hspace{.18em}i\hspace{.18em})$ $r^* = r$,
$\rm(\hspace{.08em}ii\hspace{.08em})$ $r^\tau = r$,
$\rm(i\hspace{-.08em}i\hspace{-.08em}i)$ $r^2 = 1 - Q$,
$\rm(i\hspace{-.08em}v\hspace{-.06em})$ $Q r = r Q = 0$,
$\rm(\hspace{.06em}v\hspace{.06em})$ $r T^\gamma_{u,0} = T^\gamma_{-u,0} r = 0$ and
$\rm(\hspace{-.06em}v\hspace{-.08em}i)$ $r$ anticommutes with $s$.
For $0 \leq \theta \leq \frac{\pi}{2}$, let
\begin{equation*}
b_{\theta} = \left(
    \begin{array}{cc}
           i s \cos \theta &iQ + i r \sin \theta\\
           -iQ -i r \sin \theta & i s \cos \theta\\
    \end{array}
\right), \ \
d =
\left(
    \begin{array}{cc}
           0 &i \cdot 1_{\TTg} \\
           -i \cdot 1_{\TTg} & 0 \\
    \end{array}
\right).
\end{equation*}
This $b_\theta$ is a self-adjoint unitary satisfying $b_\theta^\tau = - b_\theta$ and $b_0 = b_+$.
Therefore, $b_+$ and $b_{\frac{\pi}{2}}$ are homotopic in $U^{(2)}_1(\TTg,\taug)$.
We further discuss $b_{\frac{\pi}{2}}$.
Let consider lattice points $(m,n) \in \Z^2$ satisfying $0 \leq \gamma m - n < t$, as indicated in Figure~\ref{even} for the case where $u=1$ and $t=4$.
As in Figure~\ref{even}, we divide these points to $\frac{t}{2}$ pairs of lattice points: for $n \in \Z$ and odd $j = 1,3, \ldots, t-1$, a pair consists of $\{(n, tn-j), (n, tn-j+1)\}$.
The action of $b_{\frac{\pi}{2}}$ is closed on each pair of lattice points and is expressed by a $4 \times 4$ matrix (acting on $\C^2 \otimes \C^2$; one $\C^2$ corresponds to a pair of lattice points, and the other $\C^2$ corresponds to the $2 \times 2$ matrix we consider).
Let $V$ be the following matrix.
\begin{equation*}
V = \frac{1}{2}\left(
    \begin{array}{cccc}
           1 & 1 & 1 & -1 \\
           1 & 1 & -1 & 1 \\
           1 & -1 & 1 & 1 \\
           -1 & 1 & 1 & 1 \\
    \end{array}
\right).
\end{equation*}
Then $V \in SO(4)$ and satisfies
\begin{equation*}
	V \left(
    \begin{array}{cccc}
           0 & 0 & 0 & i \\
           0 & 0 & i & 0 \\
           0 & -i & 0 & 0 \\
           -i & 0 & 0 & 0 \\
    \end{array}
\right) V^*
	=
\left(
    \begin{array}{cccc}
           0 & 0 & i & 0 \\
           0 & 0 & 0 & i \\
           -i & 0 & 0 & 0 \\
           0 & -i & 0 & 0 \\
    \end{array}
\right),
\end{equation*}
where the left matrix inside the conjugation is the restriction of $b_{\frac{\pi}{2}}$ onto the closed subspace spanned by generating functions of these two lattice points tensor $\C^2$ and the right matrix is that of $d$
(note that $Q=0$ on these lattice points).
Let $W$ be the unitary on $\HH^\gamma \otimes \C^2$ defined by applying $V$ to these pair of lattice points satisfying $0 \leq \gamma m - n < t$ and the identity on the lattice points satisfying $t \leq \gamma m - n$, then we have $W b_{\frac{\pi}{2}} W^* = d$.
Since $SO(4)$ is path-connected, there is a path of self-adjoint unitaries from $b_{\frac{\pi}{2}}$ to $d$ preserving the relation of the $KO_2$-group.
Summarizing, we have a path in $U^{(2)}_1(\TTg,\taug)$ from $b_+$ to $d$.
By its construction, the pair of the constant path at $a \in M_2(\TT^0)$ and this path gives a path in $U^{(2)}_1(\Szg,\tau_\mathcal{S})$ from $(a,b_+)$ to $(a,d)$.
Therefore, we have $[(a, b_+)] = [(a, d)] = [I^{(2)}] = 0$ in $KO_2(\Szg, \tau_\mathcal{S})$.

We next discuss the class $[(a, b_-)]$.
For $0 \leq \theta \leq \frac{\pi}{2}$, let
\begin{equation*}
b'_\theta = \left(
    \begin{array}{cc}
           -i s \cos \theta &iQ + i(1-Q) \sin \theta \\
           -iQ -i(1-Q) \sin \theta & i s \cos \theta \\
    \end{array}
\right).
\end{equation*}
Then, $b'_\theta$ is a self-adjoint unitary satisfying $(b'_\theta)^\tau = -b'_\theta$.
We have $b'_0 = b_-$ and
$b'_{\frac{\pi}{2}} = I^{(2)}$.
Therefore, $[(a, b_-)] = [(a, b'_{\frac{\pi}{2}})] = [I^{(2)}] = 0$ in $KO_2(\Szg, \tau_\mathcal{S})$.
\end{proof}
Let consider the following elements:
\begin{equation*}
v_2 =
\left(
    \begin{array}{cc}
           0 &\hspace{-1mm} i T^0_{u,0}  \\
           -i T^0_{-u,0} & 0 \\
    \end{array}
\right) \in M_2(\TTz),
\ \ \
w_\pm =
\left(
    \begin{array}{cc}
           \pm i s &\hspace{-1mm} i T^\gamma_{u,0}  \\
           -i T^\gamma_{-u,0} & 0 \\
    \end{array}
\right) \in M_2(\TTg),
\end{equation*}
which are self-adjoint unitaries satisfying $(v_2)^\tau = -v_2$ and $w_\pm^\tau = -w_\pm$.
Since $\sigma^0(v_2) = \sigma^\gamma(w_\pm)$, pairs $(v_2, w_\pm)$ are elements of $M_2(\Szg)$ satisfying $(v_2, w_\pm)^\tau \hspace{-1mm} = -(v_2, w_\pm)$ and give elements $[(v_2, w_\pm)]$ of the group $KO_2(\Szg, \tau_\mathcal{S})$.

\begin{lemma}\label{equal2}
In $KO_2(\Szg, \tau_\mathcal{S})$, we have $[(v_2, w_+)] = [(v_2, w_-)]$.
\end{lemma}
\begin{proof}
For $0 \leq \theta \leq \pi$, let consider the following element in $M_4(\TTg)$:
\begin{equation*}
R_\theta =
\left(
    \begin{array}{cccc}
           i s \cos \theta & iT^\gamma_{u,0} & -ir \sin \theta & 0 \\
           -iT^\gamma_{-u,0} & 0 & 0 & 0 \\
           ir \sin \theta & 0 & is \cos \theta & i Q \\
           0 & 0 & -i Q  & is
    \end{array}
\right),
\end{equation*}
Then, we have $R_0 = w_+ \oplus b_+$, $R_\pi = w_- \oplus b_-$ and $R_\theta$ is a self-adjoint unitary satisfying $R_\theta^\tau = - R_\theta$.
Since $\sigma^\gamma(R_\theta) = \sigma^0(v_2 \oplus a)$, the pair $(v_2 \oplus a, R_\theta)$ is contained in $U^{(2)}_2(\Szg,\tau_\mathcal{S})$ and gives a path from $(v_2 \oplus a, w_+ \oplus b_+)$ to $(v_2 \oplus a, w_- \oplus b_-)$.
By using Lemma~\ref{lemtriv}, we obtain the following equality in $KO_2(\Szg, \tau_\mathcal{S})$:
\begin{equation*}
	[(v_2,w_+)] = [(v_2 \oplus a, w_+ \oplus b_+)] = [(v_2 \oplus a, w_- \oplus b_-)] = [(v_2, w_-)].
\qedhere
\end{equation*}
\end{proof}

\begin{lemma}\label{order2}
In $KO_2(\Szg, \tau_\mathcal{S})$, the element $[(v_2,w_+)]$ has order two.
\end{lemma}
\begin{proof}
For $0 \leq \theta \leq \frac{\pi}{2}$, let
\begin{equation*}
A^0_\theta =
\left(
    \begin{array}{cccc}
           0 & i T^0_{u,0}\cos \theta & i \sin \theta & 0 \\
           -i T^0_{-u,0}\cos \theta & 0 & 0 & -i \sin \theta \\
          -i \sin \theta & 0 & 0 & i T^0_{u,0}\cos \theta \\
           0 & i \sin \theta & -i T^0_{-u,0}\cos \theta  & 0
    \end{array}
\right) \in M_4(\TT^0),
\end{equation*}
\begin{equation*}
A^\gamma_\theta =
\left(
    \begin{array}{cccc}
           i s \cos \theta & i T^\gamma_{u,0}\cos \theta & i \sin \theta & 0 \\
           -i T^\gamma_{-u,0}\cos \theta & 0 & 0 & -i \sin \theta \\
          -i \sin \theta & 0 & - i s\cos \theta & i T^\gamma_{u,0}\cos \theta \\
           0 & i \sin \theta & -i T^\gamma_{-u,0}\cos \theta  & 0
    \end{array}
\right) \in M_4(\TTg).
\end{equation*}
Then, $A^0_\theta$ and $A^\gamma_\theta$ are self-adjoint unitaries satisfying $(A^0_\theta)^\tau = -A^0_\theta$ and $(A^\gamma_\theta)^\tau = -A^\gamma_\theta$, and their pair $(A^0_\theta, A^\gamma_\theta)$ is contained in $M_4(\Szg)$.
Note that $(A^0_0, A^\gamma_0) = (v_2 \oplus v_2, w_+ \oplus w_-)$.
Therefore, by Lemma~\ref{equal2}, the following equality holds in $KO_2(\Szg, \tau_\mathcal{S})$:
\begin{equation*}
	2 \cdot [(v_2, w_+)] = [(v_2, w_+)] + [(v_2, w_-)] = [(A^0_0, A^\gamma_0)] = [(A^0_{\frac{\pi}{2}}, A^\gamma_{\frac{\pi}{2}})] = 0.
\qedhere
\end{equation*}
\end{proof}

\begin{proposition}\label{KO2}
When $\alpha$ and $\beta$ are rational numbers and $t=-ps + qr$ is even, we have
$KO_2(\Sab, \tau_\mathcal{S}) \cong (\Z_2)^4$.
\end{proposition}
\begin{proof}
Since $u$ is odd when $t$ is even, the pair $([v_2], [w_+]) \in KO_2(\TTz,\tau_0) \oplus KO_2(\TTg,\taug)$ constitutes a nontrivial element of the right $\Z_2 \subset \Ker(\varphi_2)$ in the sequence (\ref{KO2ext}).
The element $[(v_2, w_+)] \in KO_2(\Szg, \tau_\mathcal{S})$ is a lift of it.
Therefore, $[(v_2, w_+)]$ is nontrivial and has order two by Lemma~\ref{order2}.
This element belongs to $\widetilde{KO}_2(\Szg, \tau_\mathcal{S})$ and, by mapping $1 \in \Z_2$ to $[(v_2, w_+)]$, we obtain a splitting of the the sequence (\ref{KO2ext}).
Therefore, $\widetilde{KO}_2(\Szg, \tau_\mathcal{S}) \cong (\Z_2)^3$ and the result follows.
\end{proof}

\subsubsection{The Group $KO_3(\Sab, \tau_\mathcal{S})$}\label{Sect.3.2.2}
We next compute $KO_3(\Sab, \tau_\mathcal{S})$.
Note that $\Ker (\varphi_3)$ depends on whether $t$ is even or odd.
When $t$ is odd, $\Ker (\varphi_3)$ is zero and, from the sequence (\ref{varphi}), we have $KO_3(\Sab, \tau_\mathcal{S}) \cong \Z_2$.

We next discuss the cases of even $t$.
In this case, the extension (\ref{varphi}) is of the following form:
\begin{equation}\label{KO3ext}
	0 \to \Z_2 \to KO_3(\Szg, \tau_\mathcal{S}) \to \Z_2 \to 0.
\end{equation}
As in Sect.~\ref{Sect.3.2.1}, we show that this sequence splits by finding a lift of the generator of the right $\Z_2$ in $KO_3(\Szg, \tau_\mathcal{S})$ of order two.
Let consider the following elements:
\begin{equation*}
v_3 =
\left(
    \begin{array}{cc}
           T^0_{u,0} & 0 \\
           0 & T^0_{-u,0} \\
    \end{array}
\right) \in M_2(\TTz), \ \ \
z_\pm
	=
\left(
    \begin{array}{cc}
           T^\gamma_{u,0} & \pm s \\
           0 & T^\gamma_{-u,0} \\
    \end{array}
\right) \in M_2(\TTg),
\end{equation*}
where the double-sign in the second equality corresponds.
Pairs $(v_3, z_\pm)$ are unitaries in $M_2(\Szg)$ satisfying $(v_3, z_\pm)^{\# \otimes \tau} = (v_3, z_\pm)$ and define elements $[(v_3, z_\pm)]$ of the $KO$-group $KO_3(\Szg, \tau_\mathcal{S})$.
\begin{lemma}\label{equal3}
In $KO_3(\Szg, \tau_\mathcal{S})$, we have $[(v_3, z_+)] = [(v_3, z_-)]$.
\end{lemma}
\begin{proof}
For $0 \leq \theta \leq \pi$, let
$z_\theta =
\left(
    \begin{array}{cc}
           T^\gamma_{u,0} & e^{i\theta}s \\
           0 & T^\gamma_{-u,0} \\
    \end{array}
\right) \in M_2(\TTg)$
which gives a path $\{ z_\theta \}_{0 \leq \theta \leq \pi}$ of unitaries satisfying $(z_\theta)^{\# \otimes \tau} = z_\theta$.
Its endpoints are $z_0 = z_+$ and $z_{\pi} = z_-$.
The pair $(v_3, z_\theta)$ satisfies $(v_3, z_\theta)^{\# \otimes \tau} = (v_3, z_\theta)$ and gives a homotopy between $(v_3, z_+)$ and $(v_3, z_-)$ in $U^{(3)}_1(\Szg,\tau_{\mathcal{S}})$.
\end{proof}

\begin{lemma}\label{order3}
The element $[(v_3, z_+)]$ in $KO_3(\Szg, \tau_\mathcal{S})$ has order two.
\end{lemma}

\begin{proof}
For $0 \leq \theta \leq \frac{\pi}{2}$, let
\begin{equation*}
B^0_\theta =
\left(
    \begin{array}{cccc}
           T^0_{u,0}\cos \theta & 0 & 0 & \sin \theta \\
           0 & T^0_{-u,0}\cos \theta & \sin \theta & 0 \\
           0 & -\sin \theta & T^0_{u,0}\cos \theta & 0 \\
           -\sin \theta & 0 & 0 & T^0_{-u,0}\cos \theta
    \end{array}
\right),
\end{equation*}
\begin{equation*}
B^\gamma_\theta =
\left(
    \begin{array}{cccc}
           T^\gamma_{u,0}\cos \theta & s \cos \theta & 0 & \sin \theta \\
           0 & T^\gamma_{-u,0}\cos \theta & \sin \theta & 0 \\
           0 & -\sin \theta & T^\gamma_{u,0}\cos \theta & -s\cos \theta \\
           -\sin \theta & 0 & 0 & T^\gamma_{-u,0}\cos \theta
    \end{array}
\right).
\end{equation*}
For each $\theta$, matrices $B^0_\theta$ and $B^\gamma_\theta$ are unitaries satisfying $(B^0_\theta)^{\# \otimes \tau} = B^0_\theta$ and $(B^\gamma_\theta)^{\# \otimes \tau} = B^\gamma_\theta$.
We have $B^\gamma_0 = v_3 \oplus v_3$ and $B^\gamma_0 = z_+ \oplus z_-$.
Note that matrices $B^0_{\frac{\pi}{2}}$ and $B^\gamma_{\frac{\pi}{2}}$ are contained in $M_4(\C)$, where they coincide.
Since this unitary satisfies the symmetry of the $KO_3$-group, this is an element of the quaternionic unitary group $U(2, \mathbb{H})$.
Since $U(2, \mathbb{H})$ is path-connected, there is a path of unitaries in $U_2^{(3)}(\Szg, \tau_{\mathcal{S}})$ connecting $(B^0_{\frac{\pi}{2}}, B^\gamma_{\frac{\pi}{2}})$ to $(1_{\mathcal{S}})_4$.
By using Lemma~\ref{equal3}, we obtain the following equality in  $KO_3(\Szg, \tau_\mathcal{S})$:
\begin{equation*}
	2 \cdot [(v_3, z_+)] = [(v_3 \oplus v_3, z_+ \oplus z_-)] = [(B^0_{\frac{\pi}{2}}, B^\gamma_{\frac{\pi}{2}})] = [(1_{\mathcal{S}})_4] = 0.
	\qedhere
\end{equation*}
\end{proof}

\begin{proposition}\label{KO3}
When $\alpha$ and $\beta$ are rational numbers and $t=-ps + qr$ is even, we have $KO_3(\Sab, \tau_\mathcal{S}) \cong (\Z_2)^2$.
\end{proposition}
\begin{proof}
The pair $([v_3], [z_+]) \in KO_3(\TTz, \tau_0) \oplus KO_3(\TTg, \taug)$ is contained in $\Ker (\varphi_3) \cong \Z_2$ and is nontrivial.
The element $[(v_3, z_+)] \in KO_3(\Szg, \tau_\mathcal{S})$ is its lift.
Therefore, the class $[(v_3, z_+)]$ is nontrivial and has order two by Lemma~\ref{order3}.
We thus obtain a splitting of the sequence (\ref{KO3ext}) and the group $KO_3(\Szg, \tau_\mathcal{S})$ is isomorphic to $(\Z_2)^2$.
\end{proof}
The results in this subsection are summarized in Tables~\ref{KOSab}, \ref{KOSab2} and \ref{KOSab3}.
\begin{table}
\caption{$KO_*(\Sab,\tau_\mathcal{S})$ when both $\alpha$ and $\beta$ are rational (or $\pm \infty$).}
\label{KOSab}
\centering
\vspace{-2mm}
\begin{tabular}{|c||c|c|c|c|c|c} \hline
  \multicolumn{1}{|c||}{$i$} & \multicolumn{1}{c|}{$0$} & \multicolumn{2}{c|}{$1$} & \multicolumn{2}{c|}{$2$} & \\ \hline
   $t=-ps + qr$ & --- & even & odd & even & odd & \\ \hline
  $KO_i(\Sab,\tau_\mathcal{S})$ & $\Z \oplus \Z_t$  & $(\Z_2)^2 \oplus \Z$ & $\Z_2 \oplus \Z$ & $(\Z_2)^4$ & $(\Z_2)^2$ & \\ \hline
  \end{tabular}
  \begin{tabular}{c|c|c|c|c|c|} \hline
  \multicolumn{2}{c|}{$3$} & \multicolumn{1}{c|}{$4$} & \multicolumn{1}{c|}{$5$} & \multicolumn{1}{c|}{$6$} & \multicolumn{1}{c|}{$7$} \\ \hline
   even & odd & --- & --- & --- & --- \\ \hline
   $(\Z_2)^2$ & $\Z_2$ &  $\Z \oplus \Z_t$ & $\Z$ & $0$ & $0$ \\ \hline
\end{tabular}

\caption{$KO_*(\Sab,\tau_\mathcal{S})$ when one of $\alpha$ and $\beta$ is rational (or $\pm \infty$) and the other is irrational.}
  \label{KOSab2}
  \vspace{-2mm}
  \begin{tabular}{|c||c|c|c|c|c|c|c|c|} \hline
  $i$ & $0$ & $1$ & $2$ & $3$ & $4$ & $5$ & $6$ & $7$ \\ \hline
  $KO_i(\Sab,\tau_\mathcal{S})$ & $\Z^2$ & $(\Z_2)^2 \oplus \Z$ & $(\Z_2)^3$ & $\Z_2$ & $ \Z^2$ & $\Z$ & $0$ & $0$ \\ \hline
  \end{tabular}

\caption{$KO_*(\Sab,\tau_\mathcal{S})$ when both $\alpha$ and $\beta$ are irrational.}
  \label{KOSab3}
  \vspace{-2mm}
  \begin{tabular}{|c||c|c|c|c|c|c|c|c|} \hline
  $i$ & $0$ & $1$ & $2$ & $3$ & $4$ & $5$ & $6$ & $7$ \\ \hline
  $KO_i(\Sab,\tau_\mathcal{S})$ & $\Z^3$ & $(\Z_2)^3 \oplus \Z$ & $(\Z_2)^4$ & $\Z_2$ & $\Z^3$ & $\Z$ & $0$ & $0$ \\ \hline
  \end{tabular}
\end{table}

\subsection{Boundary Maps Associated with Quarter-Plane Toeplitz Extensions and $KO$-Groups of $(\TTab,\tauab)$}\label{Sect.3.3}
We next consider the boundary maps\footnote{We write $\hat{\partial}_i$ for boundary maps associated with (\ref{seq1}) and write $\check{\partial}_i$ for that with (\ref{seq2}).}
of the $24$-term exact sequence for $KO$-theory associated with the sequence (\ref{seq1}):
\begin{equation}\label{b}
	\hat{\partial}_i \colon KO_i(\Sab, \tau_\mathcal{S}) \to KO_{i-1}(\K(\HHab), \tau_\K).
\end{equation}

\begin{proposition}\label{surjectivity1}
For each $i$, the boundary map $\hat{\partial}_i$ in (\ref{b}) is surjective.
\end{proposition}
\begin{proof}
When $i=-1,0,4,6$, the group $KO_{i-1}(\K(\HHab), \tau_\K)$ is trivial and the statement is obvious.
We discuss the other cases.
The proof is given by constructing explicit elements of the group $KO_i(\Sab, \tau_\mathcal{S})$, which maps to a generator of the group $KO_{i-1}(\K(\HHab), \tau_\K)$.
As in \cite{Ji95}, by using the action of $SL(2,\Z)$ on $\Z^2$, we assume $0 < \alpha \leq \frac{1}{2}$ and $1 \leq \beta < + \infty$ without loss of generality.
Let $\hat{\cP}_{m,n} = \Pab M_{m,n} \Pab M_{-m,-n} \Pab$.
As in \cite{Ji95}, we consider the following element in $\TTab$:
\begin{equation}\label{index1}
	\hat{A} = \hat{\cP}_{0,1} + M_{1,1}(1 - \hat{\cP}_{-1,0}) + M_{1,0}(\hat{\cP}_{-1,0} - \hat{\cP}_{0,1}).
\end{equation}
The operator $\hat{A}$ is Fredholm whose kernel is trivial and has one dimensional cokernel \cite{Ji95}.
We also have the following.
\begin{itemize}
\item $\hat{\gamma}(\hat{A})$ is a unitary in $\Sab$.
\item $\hat{A}$ is a real operator, that is $\fr(\hat{A}) = \hat {A}$, and $\hat{A}^\tau = \fr(\hat{A}^*) = \hat{A}^*$ holds.
\end{itemize}
From these preliminaries, the proof of Proposition~\ref{surjectivity1} is parallel to the computation in Example 9.4 of \cite{BL16}.
We summarize the results here.
\begin{itemize}
\item Let $u_1 = \hat{\gamma}(\hat{A})$.
$u_1$ is a unitary satisfying $u_1^\tau = u_1^*$ and gives an element $[u_1] \in KO_1(\Sab, \tau_\mathcal{S})$.
$\hat{\partial}_1([u_1])$ is a generator of $KO_0(\K(\HHab),\tau_\K) \cong \Z$.

\item Let $u_2 = \left(
    \begin{array}{cc}
           0&i\hat{\gamma}(\hat{A})\\
           -i\hat{\gamma}(\hat{A})^*&0
    \end{array}
\right)$.
$u_2$ is a self-adjoint unitary satisfying $u_2^\tau = -u_2$ and gives $[u_2] \in KO_2(\Sab, \tau_\mathcal{S})$.
Its image $\hat{\partial}_2 [u_2]$ is the generator of $KO_1(\K(\HHab),\tau_\K) \cong \Z_2$.

\item Let $u_3 = \diag(\hat{\gamma}(\hat{A}), \hat{\gamma}(\hat{A})^*).$
$u_3$ is a unitary satisfying $u_3^{\sharp \otimes \tau} = u_3$ and gives $[u_3] \in KO_3(\Sab, \tau_\mathcal{S})$.
Its image $\hat{\partial}_3([u_3])$ is the generator of the group $KO_2(\K(\HHab),\tau_\K) \cong \Z_2$.

\item Let $u_5 = \diag(\hat{\gamma}(\hat{A}), \hat{\gamma}(\hat{A}))$.
$u_5$ is a unitary satisfying $u_5^{\sharp \otimes \tau} = u_5^*$ and gives $[u_5] \in KO_5(\Sab, \tau_\mathcal{S})$.
Its image $\hat{\partial}_5([u_5)]$ is a generator of the group $KO_4(\K(\HHab),\tau_\K) \cong \Z$.
\qedhere
\end{itemize}
\end{proof}

\begin{remark}\label{Remark4.11}
In the case when $\alpha$, $\beta$ are both rational (or $\pm \infty$) and $t = -ps +qr$ is even, the group $KO_2(\Szg, \tau_\mathcal{S}) \cong (\Z_2)^4$ is generated by
$[-I^{(2)}]$, $[(v_2,w_+)]$, $[u_2]$ and $[(w', I_2^{(2)})]$, where\footnote{The matrix $Y^{(3)}_4$ is introduced in Appendix~\ref{Sect.B}.}
\begin{equation*}
w' = Y^{(3)}_4 \diag(1, 1-2T^0_{0,1}T^0_{0,-1}, 2T^0_{0,1}T^0_{0,-1}-1, -1){Y^{(3)*}_4}.
\end{equation*}
By the map $\sigma_* \colon KO_i(\Sab, \tau_\mathcal{S}) \to KO_i(C(\T^2), \tau_{\T})$, components generated by $[(v_2,w_+)]$ (when $i=2$) and $[(v_3,z_+)]$ (when $i=3$) maps injectively.
\end{remark}

The $KO$-groups of $(\TTab,\tauab)$ are computed by the $24$-term exact sequence of $KO$-theory associated with (\ref{seq1}) and Proposition~\ref{surjectivity1}.
The results are collected in Tables~\ref{13}, \ref{14} and \ref{15}.

\begin{table}
\caption{$KO_*(\TTab,\tauab)$ when both $\alpha$, $\beta$ are rational (or $\pm \infty$).}
  \label{13}
\centering
\vspace{-2mm}
  \begin{tabular}{|c||c|c|c|c|c|c|c|c} \hline
  \multicolumn{1}{|c||}{$i$} & \multicolumn{1}{c|}{$0$} & \multicolumn{2}{c|}{$1$} & \multicolumn{2}{c|}{$2$} & \multicolumn{2}{c|}{$3$} & \\ \hline
   $t=-ps + qr$ & --- & even & odd & even & odd & even & odd &  \\ \hline
  $KO_*(\TTab,\tauab)$ & $\Z \oplus \Z_t$  & $(\Z_2)^2$ & $\Z_2$ & $(\Z_2)^3$ & $\Z_2$ & $\Z_2$ & $0$ & \\ \hline
  \end{tabular}
  \begin{tabular}{c|c|c|c|} \hline
	\multicolumn{1}{c|}{$4$} & \multicolumn{1}{c|}{$5$} & \multicolumn{1}{c|}{$6$} & \multicolumn{1}{c|}{$7$} \\ \hline
	--- & --- & --- & --- \\ \hline
	$\Z \oplus \Z_t$ & $0$ & $0$ & $0$ \\ \hline
  \end{tabular}

\caption{$KO_*(\TTab,\tauab)$ when one of $\alpha$ and $\beta$ is rational (or $\pm \infty$) and the other is irrational.}
\label{14}
\centering
\vspace{-2mm}
\begin{tabular}{|c||c|c|c|c|c|c|c|c|} \hline
  $i$ & $0$ & $1$ & $2$ & $3$ & $4$ & $5$ & $6$ & $7$ \\ \hline
  $KO_i(\TTab,\tauab)$ & $\Z^2$ & $(\Z_2)^2$ & $(\Z_2)^2$ & $0$ & $\Z^2$ & $0$ & $0$ & $0$ \\ \hline
\end{tabular}

\caption{$KO_*(\TTab,\tauab)$ when both $\alpha$ and $\beta$ are irrational.}
\label{15}
\centering
\vspace{-2mm}
\begin{tabular}{|c||c|c|c|c|c|c|c|c|} \hline
  $i$ & $0$ & $1$ & $2$ & $3$ & $4$ & $5$ & $6$ & $7$ \\ \hline
  $KO_i(\TTab,\tauab)$ & $\Z^3$ & $(\Z_2)^3$ & $(\Z_2)^3$ & $0$ & $\Z^3$ & $0$ & $0$ & $0$ \\ \hline
\end{tabular}
\end{table}

\begin{remark}\label{remconcave}
Similar results in this section also hold for convex corners.
Let $\check{A} \in \TTs$ be an operator defined by replacing $\Pab$ in the definition of $\hat{A}$ by $\Ps$.
This $\check{A}$ is a Fredholm operator satisfying $\check{A}^\tau = \check{A}^*$ which have one dimensional kernel and trivial cokernel \cite{Hayashi3}.
As in Proposition~\ref{surjectivity1}, by using this example, we see that the boundary maps $\check{\partial}_i$ of $KO$-theory associated with the sequence (\ref{seq2}) is surjective.
The $KO$-groups of $(\TTs,\taus)$ is computed by the $24$-term exact sequence associated with (\ref{seq2}), and the results are the same as in Tables~\ref{13}, \ref{14} and \ref{15}.
Through the stabilization isomorphism, we have two boundary maps $\hat{\partial}_i$ and $\check{\partial}_i$ from $KO_i(\Sab, \tau_\mathcal{S})$ to $KO_{i-1}(\C, \id)$ associated with (\ref{seq1}) and (\ref{seq2}).
Since $\hat{\gamma}(\hat{A}) = \check{\gamma}(\check{A})$, the relation $\hat{\partial}_i = -\check{\partial}_i$ holds, as in Corollary~$1$ of \cite{Hayashi3}.
\end{remark}

\section{Toeplitz Operators Associated with Subsemigroup $(\Z_{\geq 0})^n$ of $\Z^n$}\label{Sect.4}
In this section, Toeplitz operators associated with the subsemigroup $(\Z_{\geq 0})^n$ of $\Z^n$ for $n \geq 3$ are discussed.
They are an $n$-variable generalization of the ordinary Toeplitz and quarter-plane Toeplitz operators and are briefly discussed in \cite{DH71,Do73}, where a necessary and sufficient condition for these Toeplitz operators to be Fredholm is obtained.
We revisit these operators since, in our application to condensed matter physics, models of higher-codimensional corners are given by using these $n$-variable generalizations.
Since the Toeplitz extension (\ref{Toeplitz}) and the quarter-plane Toeplitz extension (\ref{seq2}) provide a framework for these applications, we seek this extension for our $n$-variable cases (Theorem~\ref{multiexact}).
Note that we consider corners of arbitrary codimension, though of a specific shape compared to the codimension-two case \cite{Pa90}.
In this section, let $n$ be a positive integer bigger than or equals to three.

To study such Toeplitz operators, we follow Douglas--Howe's idea \cite{DH71} to use the tensor product of the Toeplitz extension,
\begin{equation}\label{Toeplitz}
	0 \rightarrow \K \overset{\iota}{\rightarrow} \TT \overset{\gamma}{\rightarrow} C(\T) \rightarrow 0,
\end{equation}
where $\K = \K(l^2(\Z_{\geq 0}))$.
There is a linear splitting of the $*$-homomorphism $\gamma$ given by the compression onto $l^2(\Z_{\geq 0})$.
For a subset $\mathcal{A} \subset \{1,\ldots,n\}$, let $\TT^n_{\mathcal{A}} = A_1 \otimes \cdots \otimes A_n$, where $A_i$ is $C(\T)$ when $i \in \mathcal{A}$ and is $\TT$ when $i \notin \mathcal{A}$.
Note that $\TT^n_\emptyset$ is isomorphic to $\TT^n$ introduced in Sect.~\ref{Sect.2.2}.
For subsets $\cD \subset \cR \subset \{1,\ldots,n\}$, let $\pi^{\cD}_{\cR} \colon \TT^n_\cD \to \TT^n_\cR$ be the $*$-homomorphism induced by $\gamma$.
Specifically, $\pi^{\cD}_{\cR} = a_1 \otimes \cdots \otimes a_n$, where $a_i$ is $\id_{C(\T)}$ when $i \in \cD$, is $\gamma$ when $i \in \cR \setminus \cD$ and is $\id_{\TT}$ otherwise.
Note that $\pi^{\cD}_{\cR}$ is a surjection and $\pi^{\emptyset}_{\emptyset} = \id$.
In the following, we use a subset $\mathcal{A}$ of $\{1,\ldots,n\}$ as a label to distinguish \CA s and the morphisms between them, which we may abbreviate brackets $\{ \cdot \}$ in our notation.
For example, we write $\TT^n_{1,2}$ for $\TT^n_{\{1,2\}}$, $\pi_{i}$ for $\pi^{\emptyset}_{\{i\}}$ and $\pi^{1}_{1,2}$ for $\pi^{\{1\}}_{\{1,2\}}$.
For each $\mathcal{A} \subset \{ 1, \ldots, n \}$, the map $\pi_\mathcal{A}$ has a linear splitting $\rho_\mathcal{A} \colon \TT^n_\mathcal{A} \to \TT^n$ given by the compression onto $l^2((\Z_{\geq 0})^n)$.
By these preliminaries, we consider the following $C^*$-subalgebra of $\TT^n_1 \oplus \cdots \oplus \TT^n_n$.
\begin{equation*}
\mathcal{S}^n =
	\left\{ (T_1, \cdots, T_n) \ \Biggl\vert
\begin{array}{ll}
   \text{For} \ 1 \leq i \leq n, \ T_i \in \TT^n_i, \\
   \text{For} \ 1 \leq i < j \leq n, \ \pi^i_{ij}(T_i) = \pi^j_{ij}(T_j)
 \end{array}
\right\}.
\end{equation*}
Let $(T_1, \ldots, T_n) \in \mathcal{S}^n$.
For a nonempty subset $\mathcal{A} \subset \{ 1, \ldots, n \}$, we take $i \in \mathcal{A}$ and consider the element $\pi^i_\mathcal{A}(T_i) \in \TT^n_\mathcal{A}$.
This element does not depend on the choice of $i \in \mathcal{A}$, and we write $T_\mathcal{A} = \pi^i_\mathcal{A}(T_i)$.
Let $\rho' \colon \mathcal{S}^n \to \TT^n$ be a linear map defined by
\vspace{-1mm}
\begin{equation*}
	\rho'(T_1, \ldots, T_n) = \sum_{k=1}^n \sum_{|\mathcal{A}|=k} (-1)^{k + 1}\rho_\mathcal{A}(T_\mathcal{A}) \vspace{-1mm}
\end{equation*}
for $(T_1, \ldots, T_n) \in \mathcal{S}^n$, where the second summation is taken over all subsets $\mathcal{A} \subset \{1, \ldots, n\}$ consisting of $k$ elements.
Let $\K^n = \K(l^2((\Z_{\geq 0})^n))$, and let $\gamma_n \colon \TT^n \to \mathcal{S}^n$ be an $*$-homomorphism given by $\gamma_n (T) = (\pi_1(T), \ldots, \pi_n(T))$.
Let $\iota_n$ be the $n$-fold tensor product of $\iota$.
\begin{theorem}\label{multiexact}
There is the following short exact sequence of \CA s:
\begin{equation}\label{multimain}
	0 \to \K^n \overset{\iota_n}{\to} \TT^n \overset{\gamma_n}{\to} \mathcal{S}^n \to 0,
\end{equation}
where the map $\gamma_n$ has a linear splitting given by $\rho'$.
\end{theorem}
\begin{proof}
The map $\iota_n$ is injective since $\iota$ is injective.
We first show the exactness at $\TT^n$.
Since $\gamma \circ \iota = 0$, we have $\gamma_n \circ \iota_n = 0$, and thus, $\Image(\iota_n) \subset \Ker(\gamma_n)$.
Let $T \in \Ker(\gamma_n)$.
Since $\pi_1(T) = (\gamma \otimes 1 \otimes \cdots \otimes 1)(T) = 0$, there exists some $S_1 \in \K \otimes \TT \otimes \cdots \otimes \TT$ such that $(\iota \otimes 1 \otimes \cdots \otimes 1)(S_1) = T$.
Since
\begin{align*}
	0 &= (1 \otimes \gamma \otimes 1 \otimes \cdots \otimes 1)(T)
	= (1 \otimes \gamma \otimes 1 \otimes \cdots \otimes 1)(\iota \otimes 1 \otimes \cdots \otimes 1)(S_1)\\
	&= (\iota \otimes 1 \otimes \cdots \otimes 1)(1 \otimes \gamma \otimes 1 \otimes \cdots \otimes 1)(S_1)
\end{align*}
and $\iota \otimes 1 \otimes \cdots \otimes 1$ is injective, $(1 \otimes \gamma \otimes 1 \otimes \cdots \otimes 1)(S_1) = 0$.
Therefore, there exists some $S_2 \in \K \otimes \K \otimes \TT \otimes \cdots \otimes \TT$ such that $S_1 = (1 \otimes \iota \otimes 1 \otimes \cdots \otimes 1)(S_2)$.
By continuing this argument, we see that there exists some $S_n \in \K \otimes \cdots \otimes \K \cong \K^n$ such that $(\iota \otimes \cdots \otimes \iota)(S_n) = T$.
Thus, we have $\Ker(\gamma_n) \subset \Image(\iota_n)$.

For the surjectivity of $\gamma_n$, we see that $\rho$ is a linear splitting of $\gamma_n$, that is,
for $(T_1,\ldots,T_n) \in \mathcal{S}^n$ and $1 \leq i \leq n$, the relation $\pi_i \circ \rho' (T_1, \ldots, T_n) = T_i$ holds.
In the following, we show $\pi_1 \circ \rho' (T_1, \ldots, T_n) = T_1$ and the other case is proved similarly.
Note that
\vspace{-1mm}
\begin{equation}\label{sum}
	\pi_1 \circ \rho'(T_1, \ldots, T_n) = \sum_{k=1}^n \sum_{|\mathcal{A}|=k} (-1)^{k+1} \pi_1 \circ \rho_\mathcal{A}(T_\mathcal{A})
\vspace{-1mm}
\end{equation}
and that $\pi_1 \circ \rho_1(T_1) = T_1$.
Thus, it is sufficient to show that the sum over $\mathcal{A}$ $(\neq \{ 1 \})$ is zero.
Note that for $2 \leq i_1 < \cdots < i_{k-1} \leq n$, we have
\begin{equation*}
	\pi_1 \circ \rho_{i_1, \ldots, i_{k-1}}(T_{i_1, \ldots, i_{k-1}})
		=
	\pi_1 \circ \rho_{1, i_1, \ldots, i_{k-1}}(T_{1, i_1, \ldots, i_{k-1}}).
\end{equation*}
By using this relation, we compute the sum on the right-hand side of (\ref{sum}).
For $k=1$ and $k=2$ of the sum, we have the following:
\begin{align*}
	&\sum_{|\mathcal{A}|=1, \ \mathcal{A} \neq \{1\}}\pi_1 \circ \rho_\mathcal{A}(T_\mathcal{A}) - \sum_{|\mathcal{A}|=2}\pi_1 \circ \rho_\mathcal{A}(T_\mathcal{A})\\ \vspace{-4mm}
	&= \sum_{i=2}^n \pi_1 \circ \rho_i (T_i) - \sum_{1 \leq i < j \leq n} \pi_1 \circ \rho_{i,j}(T_{ij})
	= - \sum_{2 \leq i < j \leq n} \pi_1 \circ \rho_{i,j}(T_{ij}).
\end{align*}
For $2 < k < n$, we have
\begin{align*}
	&(-1)^k \sum_{2 \leq i_1 < \cdots < i_{k-1} \leq n}\pi_1 \circ \rho_{i_1, \ldots, i_{k-1}}(T_{i_1, \ldots, i_{k-1}})\\
		& \ \ + (-1)^{k+1} \sum_{1 \leq i_1 < \cdots < i_k \leq n}\pi_1 \circ \rho_{i_1, \ldots, i_k}(T_{i_1, \ldots, i_k})\\
	&= (-1)^{k+1} \sum_{2 \leq i_1 < \cdots < i_k \leq n}\pi_1 \circ \rho_{i_1, \ldots, i_k}(T_{i_1, \ldots, i_k}),
\end{align*}
and since
\begin{equation*}
	(-1)^n\pi_1 \circ \rho_{2,\ldots, n}(T_{2,\ldots,n}) + (-1)^{n+1}\pi_1 \circ \rho_{1,\ldots, n}(T_{1,\ldots,n}) = 0,
\end{equation*}
we have $\pi_1 \circ \rho' (T_1, \ldots, T_n) = T_1$.
\end{proof}
Theorem~\ref{sum} leads to the necessary and sufficient condition for Toeplitz operators associated with these codimension-$n$ corners to be Fredholm.
\begin{corollary}[Theorem~18 of \cite{Do73}]\label{nFredholm}
Let $k$ be a positive integer.
An operator $T \in M_k(\TT^n)$ is a Fredholm operator if and only if $\gamma_n(T)$ is invertible in $M_k(\mathcal{S}^n)$ or, equivalently, if and only if $\pi_i(T)$ is invertible for any $1 \leq i \leq n$.
\end{corollary}
As in Sect.~\ref{Sect.2.2}, the real structure $c$ on $l^2(\Z^n)$ induces real structures on $\TT^n_i$ and $\mathcal{S}^n$.
We write $\tau_\mathcal{S}$ for the transposition on $\mathcal{S}^n$ associated with this real structure.
The map $\gamma_n$ preserves the real structure, and we obtain the following exact sequence of \TA s:
\begin{equation}\label{multimainreal}
	0 \to (\K^n, \tau_\K) \overset{\iota_n}{\to} (\TT^n, \tau_\TT) \overset{\gamma_n}{\to} (\mathcal{S}^n, \tau_\mathcal{S}) \to 0.
\end{equation}

We next compute the $K$-groups of the \CA \ $\mathcal{S}^n$ and $KO$-groups of the \TA \ $(\mathcal{S}^n, \tau_\mathcal{S})$.
\begin{proposition}
$K_i(\mathcal{S}^n) \cong \Z$ for $i = 0,1$.
\end{proposition}
\begin{proof}
Note that $K_i(\TT^n) \cong K_i(\C)$. The result follows from the six-term exact sequence of $K$-theory associated with the sequence (\ref{multimain}) in Theorem~\ref{multiexact}.
\end{proof}

\begin{proposition}\label{prop2}
For each $i$, we have
\begin{equation*}
	KO_i(\mathcal{S}^n, \tau_\mathcal{S}) \cong KO_i(\C,\id) \oplus KO_{i-1}(\C, \id).
\end{equation*}
The results are collected in Table~\ref{multiKO}.
\end{proposition}
\begin{proof}
Note that $KO_i(\TT^n, \tau_\TT) \cong KO_i(\C, \id)$. The result follows from the $24$-term exact sequence of $KO$-theory associated with the sequence (\ref{multimainreal}).
\end{proof}

\begin{table}
\caption{$KO$-groups of $(\mathcal{S}^n ,\tau_\mathcal{S})$.}
\label{multiKO}
\centering
  \begin{tabular}{|c||c|c|c|c|c|c|c|c|} \hline
  $i$ & $0$ & $1$ & $2$ & $3$ & $4$ & $5$ & $6$ & $7$ \\ \hline
  $KO_i(\mathcal{S}^n ,\tau_\mathcal{S})$ & $\Z$ & $\Z \oplus \Z_2$ & $(\Z_2)^2$ & $\Z_2$ & $\Z$ & $\Z$ & $0$ & $0$ \\ \hline
\end{tabular}
\end{table}

A Fredholm Toeplitz operator associated with a codimension-$n$ corner whose Fredholm index is one is constructed as follows.
\begin{example}[A Fredholm Operator of Index One]\label{multiexam}
Let $T_z$ be the Toeplitz operator whose symbol $z \colon \T \to \C$ is the inclusion.
Its adjoint $T_z^*$ is a Fredholm operator on $l^2(\Z_{\geq 0})$ of index one.
Let $p = T_z T_z^*$ and $q = 1_{\TT} - p$, then $p, q \in \TT$ and are projections onto $l^2(\Z_{\geq 1})$ and $\C \delta_{0}$, respectively, where $\delta_0$ is the characteristic function of the point $0 \in \Z$.
For a subset $\mathcal{A} \subset \{ 1, \ldots, n\}$, let $P^n_\mathcal{A} = r_1 \otimes \cdots \otimes r_n$,
where $r_i$ is $p$ when $i \in \mathcal{A}$ and is $q$ otherwise.
The operator $P^n_\mathcal{A}$ is a projection which satisfies $\sum_{\mathcal{A}} P^n_\mathcal{A} = 1_{\TT^n}$.
Let $\tilde{T} = T_z^* \otimes q \otimes \cdots \otimes q$ and consider the following element in $\TT^n$:
\begin{equation*}
	G = \tilde{T} + \sum_{\mathcal{A} \neq \{1\}} P^n_\mathcal{A},
\end{equation*}
where the sum is taken over all subsets of $\{ 1, \ldots, n\}$ except $\{ 1 \}$.
Then, we can see that $\Ker (G) \cong \C$ and $\Coker (G) = 0$, that is, $G$ is a Fredholm Toeplitz operator associated with codimension-$n$ corners whose Fredholm index is one.
\end{example}
This example leads to the following result.
\begin{proposition}\label{surjectivity2}
The boundary maps of the six-term exact sequence for $K$-theory associated with $(\ref{multimain})$ are surjective.
Moreover, the boundary maps of the 24-term exact sequence for $KO$-theory associated with $(\ref{multimainreal})$ are surjective.
\end{proposition}
\begin{proof}
The result for complex $K$-theory is immediate from Example~\ref{multiexam}.
For $KO$-theory, since the operator $G$ in Example~\ref{multiexam} satisfies $G^{\tau} = G^*$, the result follows as in Proposition~\ref{surjectivity1}.
\end{proof}
Note that, we have $\gamma_n(G) = (\pi_1(G), 1, \cdots, 1) \in \mathcal{S}^n$ by using
\begin{equation*}
	\pi_1(G) = M_{z}^* \otimes q \otimes \cdots \otimes q + \sum_{\emptyset \neq \mathcal{A} \subset \{1, \ldots, n-1 \}} 1_{C(\T)} \otimes P^{n-1}_\mathcal{A},
\end{equation*}
The element $\gamma_n(G)$ is a unitary and defines an element $[\gamma_n(G)]$ of the group $K_1(\mathcal{S}^n)$.
Since the Fredholm index of $G$ is one, this gives a generator of $K_1(\mathcal{S}^n) \cong \Z$.
As in the proof of Proposition~\ref{surjectivity1}, generators of the $KO$-groups $KO_i(\mathcal{S}^n, \tau_\mathcal{S})$ are also obtained by using $G$.

\begin{remark}\label{relationhigher}
Let $1 \leq j \leq n$. We have the following $*$-homomorphisms:
\begin{equation*}
	\mathcal{S}^n \longrightarrow \TT^{n-1} \otimes C(\T) \overset{\gamma_{n-1} \otimes 1}{\longrightarrow}\mathcal{S}^{n-1} \otimes C(\T),
\end{equation*}
where the first map maps $(T_1, \ldots, T_n)$ to $T_j$.
We write $\sigma^{n,n-1}$ for the composite of the above maps which induces the map
$\sigma^{n,n-1}_* \colon K_i(\mathcal{S}^n) \to K_i(\mathcal{S}^{n-1} \otimes C(\T))$.
When $i=0$, $K_0(\mathcal{S}^n) \cong \Z$ is generated by $[1_2]$ and $\sigma^{n,n-1}_*[1_2] = [1_2]$.
When $i=1$, the map $\sigma^{n,n-1}_*$ is zero since, by Example~\ref{multiexam}, the element $[\gamma_n(G)]$ is a generator of $K_1(\mathcal{S}^n) \cong \Z$ and $\sigma^{n,n-1}_*[\gamma_n(G)] = [1] = 0$.
A similar observation also holds in real cases.
The map $\sigma^{n,n-1}_*$ from $KO_i(\mathcal{S}^n, \tau_\mathcal{S})$ to $KO_i(\mathcal{S}^{n-1} \otimes C(\T), \tau)$ maps direct summands corresponding to $KO$-groups of a point injectively and the other components to zero.
\end{remark}

\section{Topological invariants and corner states in Altland--Zirnbauer classification}\label{Sect.5}
In this section, some gapped Hamiltonians on a lattice with corners are discussed in each of the Altland--Zirnbauer classes.
Since two of them (class $\A$ and $\AIII$) are already studied in \cite{Hayashi2,Hayashi3}, we consider the remaining cases here.
The codimension of the corner will be arbitrary, though we mainly discuss codimension-two cases, with many detailed results being obtained by \cite{Pa90,Ji95,Hayashi3} and the results in Sect.~\ref{Sect.3}.
Higher-codimensional cases are discussed in a similar way, whose results are collected in Sect.~\ref{Sect.5.5}.

\subsection{Setup}\label{Sect.5.1}
Let $V$ be a finite rank Hermitian vector space of complex rank $N$.
Let $n$ be a positive integer.
Let $\Theta$ and $\Xi$ be antiunitary operators on $V$ whose squares are $+1$ or $-1$.
Let $\Pi$ be a unitary operator on $V$ whose square is one.
These operators $\Theta$, $\Xi$ and $\Pi$ are naturally extended to the operator on $l^2(\Z^n; V)$ by the fiberwise operation; we also denote them as $\Theta$, $\Xi$ and $\Pi$, respectively.
Let $\mathrm{Herm}(V)$ be the space of Hermitian operators on $V$.
We consider a continuous map $\T^n \to \mathrm{Herm}(V)$, $t \mapsto H(t)$, where $t = (t_1, t_2, \ldots, t_n)$ in $\T^n$.
Through the Fourier transform $L^2(\T^n; V) \cong l^2(\Z^n; V)$, the multiplication operator generated by this continuous map defines a bounded linear self-adjoint operator $H$ on the Hilbert space $l^2(\Z^n; V)$.
We consider the lattice $\Z^n$ as a model of the bulk and call $H$ the {\em bulk Hamiltonian}.
The Hamiltonian is said to preserve
{\em time-reversal symmetry} (TRS) if it commutes with $\Theta$ (i.e., $\Theta H \Theta^* = H$),
{\em particle-hole symmetry} (PHS) if it anticommutes with $\Xi$ (i.e., $\Xi H \Xi^* = -H$) and
{\em chiral symmetry} if it anticommutes with $\Pi$ (i.e., $\Pi H \Pi^* = -H$).
Furthermore, TRS or PHS is called {\em even} (resp. {\em odd}) if $\Theta^2 = 1$ or $\Xi^2 = 1$ (resp. $\Theta^2 = -1$ or $\Xi^2 = -1$).
Hamiltonians may preserve both TRS and PHS.
In that case, $\Theta$ and $\Xi$ are assumed to commute, and $\Pi$ is identified with $\Theta \Xi$ or $i\Theta \Xi$ such that $\Pi^2 = 1$ is satisfied.

By taking the partial Fourier transform in the variables $t_1$ and $t_2$, we obtain a continuous family of bounded linear self-adjoint operators $\{ H(\bt)\}_{\bt \in \T^{n-2}}$ on $\HH \otimes V$.
By taking a compression onto $\HHa \otimes V$, $\HHb \otimes V$ and $\HHab \otimes V$, we obtain a family of operators $H^\alpha(\bt)$, $H^\beta(\bt)$ and $\hat{H}^{\alpha,\beta}(\bt)$ parametrized by $\bt = (t_3, \ldots, t_n) \in \T^{n-2}$.
$H^\alpha(\bt)$ and $H^\beta(\bt)$ are our models for two surfaces (codimension-one boundaries), and $\hat{H}^{\alpha,\beta}(\bt)$ is our model of the corner (codimension-two corner).
We assume the following spectral gap condition.
\begin{assumption}[Spectral Gap Condition]\label{sgc1}
We assume that both $H^\alpha$ and $H^\beta$ are invertible.
\end{assumption}
Under this assumption, the bulk Hamiltonian $H$ is also invertible since, when we take a basis of $V$ and identify $V$ with $\C^N$, there is a unital $*$-homomorphism $M_N(\Sab \otimes C(\T^{n-2})) \to M_N(C(\T^{n}))$ that maps $(H^\alpha, H^\beta)$ to $H$.
In classes $\AI$ and $\AII$, we further assume that the spectrum of $H$ is not contained in either $\R_{> 0}$ or $\R_{< 0}$.
Note that in other classes where Hamiltonians preserve PHS or chiral symmetry, this condition follows from Assumption~\ref{sgc1}.
Let $h$ be the pair $(H^\alpha, H^\beta)$.
Under Assumption~\ref{sgc1}, we set
\begin{equation}\label{sign}
	\sign(h) = h | h |^{-1}.
\end{equation}
When the bulk Hamiltonian $H$ satisfies TRS, PHS or chiral symmetry, the operators $H^\alpha$, $H^\beta$, $\hat{H}^{\alpha,\beta}$ and $\sign(h)$ also satisfy the symmetry, that is, commutes with $\Theta$ or anticommutes with $\Xi$ or $\Pi$.

\subsection{Gapped Topological Invariants}\label{Sect.5.2}
In the following, starting from a Hamiltonian satisfying Assumption~\ref{sgc1} in each class $\AI$, $\BDI$, $\DD$, $\DIII$, $\AII$, $\CII$, $\CC$ and $\CI$, we construct a unitary and see that this unitary satisfies the relation $\mathcal{R}_i$ in Table~\ref{BL}.
By using this unitary, we define a topological invariant as an element of some $KO$-group.

In class $\AI$, the Hamiltonian has even TRS.
We take an orthonormal basis of $V$ to identify $V$ with $\C^N$ and express $\Theta$ as $\fC = \diag(c, \ldots, c)$.
Under our spectral gap condition, let
\begin{equation}\label{diagonal}
	u= \left(
    \begin{array}{cc}
           \sign(h)&0\\
           0&1_N
    \end{array}
\right).
\end{equation}
This $u$ is a self-adjoint unitary satisfying
$u^\tau = \Ad_{\fC \oplus \fC}(u^*) = u^*$
by the TRS.

In class BDI, the Hamiltonian has both even TRS and even PHS.
Note that the chiral symmetry is given by $\Pi = \Theta \Xi$ and commutes with $\Theta$ and $\Xi$.
For a Hamiltonian satisfying chiral symmetry and Assumption~\ref{sgc1} to exist, the even/odd decomposition $V \cong V^0 \oplus V^1$ with respect to $\Pi$ should satisfy $\rank_\C V^0 = \rank_\C V^1$, and we assume that.
Then, there is an orthonormal basis of $V$ to identify $V$ with $\C^N$ such that $\Pi$ and $\Theta$ are expressed as follows:
\begin{equation}\label{symBDI}
	\Pi =
\left(
    \begin{array}{cc}
           1&0\\
           0&-1
    \end{array}
\right),
\ \
\Theta =
\left(
    \begin{array}{cc}
           \fC&0\\
           0&\fC
    \end{array}
\right),
\end{equation}
where $\fC = \diag(c, \ldots, c)$.
Since the Hamiltonian $H$ anticommutes with $\Pi$, the operator $\sign(h)$ in (\ref{sign}) is written in the following off-diagonal form:
\begin{equation}\label{offdiag}
	\sign(h) =
\left(
    \begin{array}{cc}
           0&u^*\\
           u&0
    \end{array}
\right),
\end{equation}
where $u$ is a unitary.
By the TRS, we have $u^\tau = \fC u^* \fC^* = u^*$.

In class D, the Hamiltonian has even PHS.
We take an orthonormal basis of $V$ to identify $V \cong \C^N $ and express $\Xi$ as $\fC = \diag(c, \ldots, c)$.
Let $u= \sign(h)$, then we have
$u^\tau = \Xi u \Xi^* = -u$
by the PHS.

In class $\DIII$, the Hamiltonian has both odd TRS and even PHS.
Note that the chiral symmetry is given by $\Pi = i\Theta \Xi$ and anticommutes with $\Theta$ and $\Xi$.
For such a Hamiltonian $H$ satisfying Assumption~\ref{sgc1} to exist, the complex rank of $V$ must be a multiple of $4$ since $\mathrm{sign}(H(\bt))$, $i\Pi$, $i$ and $\Theta$ provides a $\Cl_{1,1} \otimes \Cl_{2,0} \cong \mathbb{H}(2)$-module structure on $V$.
We assume $\rank_\C V = 4M$ for some positive integer $M$.
\begin{lemma}\label{lemma5.2}
If a Hamiltonian $H$ satisfying Assumption~\ref{sgc1} exists, there is an orthonormal basis of $V$ such that $\Pi$ and $\Theta$ are expressed as follows.
\begin{equation*}
	\Pi =
\left(
    \begin{array}{cc}
           1&0\\
           0&-1
    \end{array}
\right),
\ \
\Theta =
\left(
    \begin{array}{cc}
           0& \fJ \\
          \fJ &0
    \end{array}
\right).
\end{equation*}
We write $\fJ = \diag(j,\ldots,j)$, where $j$ is the quaternionic structure on $\mathbb{H}$.
\end{lemma}
\begin{proof}
	By using $\Pi$, we decompose $V \cong V^0 \oplus V^1$.
	We identify $V^0 \cong V^1 \cong \C^{2M} \cong \mathbb{H}^M$, on which we consider $\fJ = \diag(j,\ldots,j)$.
	Let
$U = \Theta
\left(
    \begin{array}{cc}
           0& \fJ\\
           \fJ&0
    \end{array}
\right)$.
Since $U$ is a unitary and commutes with $\Pi$, we have $U = \diag(u_0,u_1)$, where $u_0$ and $u_1$ are unitaries on $\C^{2M}$.
Since $\Theta^2 = -1$, we have $u_1 = \Ad_{\fJ}(u_0^*)$.
Let $W = \diag(-u_0^*,1)$, then 
$
W \Theta W^* =
\left(
    \begin{array}{cc}
           0& \fJ\\
           \fJ&0
    \end{array}
\right)$.
\end{proof}
We take this basis on $V$ and express $\Pi$ and $\Theta$ as above.
By the chiral symmetry, we take $u$ in (\ref{offdiag}).
By the TRS, we have $u^{\sharp \otimes \tau} = \fJ u^* \fJ^* = u$.

In class $\AII$, the Hamiltonian has odd TRS.
The space $V$ has a quaternionic structure given by $\Theta$, and the complex rank of $V$ is even, for which we write $2M$.
There is an orthonormal basis of $V$ for identifying $V$ with $\C^{2M} \cong \mathbb{H}^M$ and expressing $\Theta$ as $\fJ = \diag(j, \ldots, j)$.
Let $u$ be a self-adjoint unitary in (\ref{diagonal}).
By the TRS, we have
$u^{\sharp \otimes \tau} = \Ad_{\fJ \oplus \fJ}(u^*) = u^*$.

In class $\CII$, the Hamiltonian has both odd TRS and odd PHS.
The chiral symmetry is given by $\Pi = \Theta \Xi$ and commutes with $\Theta$ and $\Xi$.
As in the class BDI case, we take an orthonormal basis of $V$ to identify $V$ with $\C^N$ and express $\Pi$ and $\Theta$ as
\begin{equation}\label{symCII}
	\Pi =
\left(
    \begin{array}{cc}
           1&0\\
           0&-1
    \end{array}
\right),
\ \
\Theta =
\left(
    \begin{array}{cc}
          \fJ& 0 \\
          0 & \fJ
    \end{array}
\right),
\end{equation}
where $\fJ = \diag(j, \ldots, j)$.
By the chiral symmetry, we take $u$ in (\ref{offdiag}).
By the TRS, we have $u^{\sharp \otimes \tau} = \fJ u^* \fJ^* = u^*$.

In class C, the Hamiltonian has odd PHS.
Since $\Xi$ provides a quaternionic structure on $V$, its complex rank is even, for which we write $2M$.
We take an orthonormal basis of $V$ to identify $V$ with $\C^{2M} \cong \mathbb{H}^M$ and express $\Xi$ as $\fJ = \diag(j, \ldots, j)$.
Let $u = \sign(h)$, then we have
$u^{\sharp \otimes \tau} = \fJ u^* \fJ^* = -u$ by the PHS.

In class CI, the Hamiltonian has both even TRS and odd PHS.
The chiral symmetry is given by $\Pi = i\Theta \Xi$ and anticommutes with $\Theta$ and $\Xi$.
As in Lemma~\ref{lemma5.2}, we take an orthonormal basis of $V$ to express,
\begin{equation*}
	\Pi =
\left(
    \begin{array}{cc}
           1&0\\
           0&-1
    \end{array}
\right),
\ \
\Theta =
\left(
    \begin{array}{cc}
           0& \fC \\
          \fC &0
    \end{array}
\right),
\end{equation*}
where $\fC = \diag(c,\ldots,c)$.
By the chiral symmetry, we take $u$ in (\ref{offdiag}).
By the TRS, we have $u^\tau = \fC u^* \fC^* = u$.

\begin{definition}\label{gappdinv1}
For a Hamiltonian in class $\spadesuit =$ $\AI$, $\BDI$, $\DD$, $\DIII$, $\AII$, $\CII$, $\CC$ or $\CI$ satisfying Assumption~\ref{sgc1},
let $u$ be the unitary defined as above.
As we have seen, this unitary $u$ satisfies the relation $\mathcal{R}_{i(\spadesuit)}$ where $i(\spadesuit)$ is as indicated in Table~\ref{label}.
We denote its class\footnote{We simply write $\tau$ in place of $\tau_\mathcal{S}\otimes \tau_\T$. In the following, these abbreviations for tensor products of transpositions are employed, though the meaning will be clear from the context.} $[ u ]$ in the $KO$-group $KO_{i(\spadesuit)}(\Sab \otimes C(\T^{n-2}), \tau)$ by $\I_\gapped^{n,2, \spadesuit}(H)$.
\end{definition}
\begin{table}
\caption{$i(\spadesuit)$ and $c(\spadesuit)$ for each of the Altland--Zirnbauer classes $\spadesuit$.}
\label{label}
\centering
\begin{tabular}{|c||c|c|c|c|c|c|c|c|} \hline
  $\spadesuit$ & $\AI$ & $\BDI$ & $\DD$ & $\DIII$ & $\AII$ & $\CII$ & $\CC$ & $\CI$ \\ \hline
  $i(\spadesuit)$ & $0$ & $1$ & $2$ & $3$ & $4$ & $5$ & $6$ & $-1$ \\ \hline
  $c(\spadesuit)$ & $1$ & $1$ & $i$ & $1$ & $1$ & $1$ & $i$ & $1$ \\ \hline
\end{tabular}
\end{table}
The groups $KO_{*}(\Sab \otimes C(\T^{n-2}), \tau)$ are computed by results in Sect.~\ref{Sect.3.2}.
\begin{remark}
We expressed the symmetry operators in a specific way, though we may choose another one.
In class DIII, for example, the operator $\Theta$ can also be expressed as
$\left(
    \begin{array}{cc}
           0& -\fC \\
          \fC &0
    \end{array}
\right)$,
where $\fC = \diag(c,\ldots,c)$.
Then, we obtain unitaries satisfying $u^\tau = -u$, which are treated in \cite{Kel17}.
\end{remark}

\subsection{Gapless Topological Invariants}\label{Sect.5.3}
We next define another topological invariant by using our model for the corner $\hat{H}^{\alpha,\beta}$.
By Assumption~\ref{sgc1} and Theorem~$2.6$ in \cite{Pa90}, $\{ \hat{H}^{\alpha,\beta}(\bt) \}_{\bt \in \T^{n-2}}$ is a continuous family of self-adjoint Fredholm operators.
Corresponding to its Altland--Zirnbauer classes, this family provides a $\Z_2$-map from $(\T^{n-2}, \zeta)$ to some $\Z_2$-spaces of self-adjoint or skew-adjoint Fredholm operators introduced in Appendix~\ref{Sect.A} as follows.

\begin{itemize}
\item Class AI, $\Z_2$-map $\hat{H}^{\alpha,\beta} \colon (\T^{n-2}, \zeta) \to (\Fred^{(0,\underline{1})}_*, \fr_\Theta)$.
\item Class BDI, let $\epsilon_1 = \Pi$. $\Z_2$-map $\hat{H}^{\alpha,\beta} \colon (\T^{n-2}, \zeta) \to (\Fred^{(0,\underline{2})}_*, \fr_\Theta)$.
\item Class D, $\Z_2$-map $i\hat{H}^{\alpha,\beta} \colon (\T^{n-2}, \zeta) \to (\Fred^{(\underline{1},0)}_*, \fr_\Xi)$.
\item Class $\DIII$, let $e_1 = i\Pi$. $\Z_2$-map $\hat{H}^{\alpha,\beta} \colon (\T^{n-2}, \zeta) \to (\Fred^{(1,\underline{1})}_*, \fq_\Theta)$.
\item Class $\AII$, $\Z_2$-map $\hat{H}^{\alpha,\beta} \colon (\T^{n-2}, \zeta) \to (\Fred^{(0,\underline{1})}_*, \fq_\Theta)$.
\item Class $\CII$, let $\epsilon_1 = \Pi$. $\Z_2$-map $\hat{H}^{\alpha,\beta} \colon (\T^{n-2}, \zeta) \to (\Fred^{(0,\underline{2})}_*, \fq_\Theta)$.
\item Class C, $\Z_2$-map $i\hat{H}^{\alpha,\beta} \colon (\T^{n-2}, \zeta) \to (\Fred^{(\underline{1},0)}_*, \fq_\Xi)$.
\item Class CI, let $e_1 = i\Pi$. $\Z_2$-map $\hat{H}^{\alpha,\beta} \colon (\T^{n-2}, \zeta) \to (\Fred^{(1,\underline{1})}_*, \fr_\Theta)$.
\end{itemize}
Here, we write $\fr_\Theta = \Ad_\Theta$ when $\Theta^2=1$ and $\fq_\Theta = \Ad_\Theta$ when $\Theta^2 = -1$.
Involutions $\fr_\Xi$ and $\fq_\Xi$ are defined as $\Ad_\Xi$ in the same way.
By Corollary~\ref{htpyomega}, the $\Z_2$-homotopy classes of $\Z_2$-maps from $(\T^{n-2}, \zeta)$ to the above $\Z_2$-space of self-adjoint or skew-adjoint Fredholm operators is isomorphic to the $KO$-group $KO_{i}(C(\T^{n-2}), \tau_\T)$ of some degree $i$.

\begin{definition}\label{gapless1}
For $\spadesuit=$ AI, BDI, D, DIII, AII, CII, C or CI, let $i(\spadesuit)$, $c(\spadesuit)$ and $\Fred^{\spadesuit}$ be numbers and the $\Z_2$-space as in Tables~\ref{label} and \ref{Table1}.
For a Hamiltonian $H$ in class $\spadesuit$ satisfying Assumption~\ref{sgc1}, we denote $\I_\Corner^{n,2, \spadesuit}(H)$ for the class $[c(\spadesuit) \hat{H}^{\alpha,\beta}]$ in $KO_{i(\spadesuit)-1}(C(\T^{n-2}), \tau_\T)$.
We call $\I_\Corner^{n,2, \spadesuit}(H)$ the {\em gapless corner invariant}.
\end{definition}
If the gapless corner invariant is nontrivial, zero is contained in the spectrum of $\hat{H}^{\alpha, \beta}$.
In Sect.~\ref{Sect.5.6}, we discuss more refined relations between the gapless corner invariant and corner states when $k=n-1$ and $n$.

\begin{table}[htb]
\caption{In each Altland--Zirnbauer class $\spadesuit$, gapped invariants and gapless invariants are defined as elements of some $K$- and $KO$-groups of some degree, as indicated in this table.
Classifying spaces for topological $K$- and $KR$-groups through self-adjoint or skew-adjoint Fredholm operators and unitaries are also included.
($\Z_2$-)spaces $\Fred^{\spadesuit}$ and $U^{\spadesuit}$ are introduced in Appendix~\ref{Sect.A}.}
\label{Table1}
\centering
\scalebox{0.95}{
  \begin{tabular}{|c||c|c|c|c|} \hline
  	\multicolumn{1}{|c||}{AZ} & \multicolumn{1}{|c|}{Gapped} & \multicolumn{3}{|c|}{Gapless} \\
    $\spadesuit$ & $K$-group & $K$-group & $\Fred^{\spadesuit}$ & $U^{\spadesuit}$ \\ \hline \hline
    $\A$  & $K_0$ & $K_1$ & $\Fred^{(0,\underline{1})}_*$ & $U_\cpt$ \\
    $\AIII$ & $K_1$ & $K_0$ & $\Fred$ & $U^{(0,\underline{1})}_\cpt$ \\ \hline
    $\AI$ & $KO_0$ & $KO_{-1}$ & $(\Fred^{(0,\underline{1})}_*, \fr_\Theta)$ & \hspace{-2mm} $(U_\cpt, \fr \circ*)$ \hspace{-2mm} \rule[0mm]{0mm}{4mm}\\
    $\BDI$ & $KO_1$ & $KO_0$ & $(\Fred^{(0,\underline{2})}_*, \fr_\Theta) \ (\cong (\Fred^{(\underline{1}, 1)}_*, \fr_\Xi))$ & $(U^{(0,\underline{1})}_\cpt, \fr)$ \rule[0mm]{0mm}{4mm}\\
    $\DD$ & $KO_2$ & $KO_1$ & $(\Fred^{(\underline{1}, 0)}_*, \fr_\Xi)$ & $(U_\cpt, \fr)$ \rule[0mm]{0mm}{4mm}\\
    $\DIII$ & $KO_3$ & $KO_2$ & $(\Fred^{(1,\underline{1})}_*, \fq_\Theta) \ (\cong (\Fred^{(\underline{2}, 0)}_*, \fr_\Xi))$ & $(U^{(\underline{1},0)}_\cpt, \fr)$ \rule[0mm]{0mm}{4mm}\\
    $\AII$ & $KO_4$ & $KO_3$ & $(\Fred^{(0,\underline{1})}_*, \fq_\Theta)$ & \hspace{-2mm} $(U_\cpt, \fq \circ*)$ \hspace{-2mm} \rule[0mm]{0mm}{4mm}\\
    $\CII$ & $KO_5$ & $KO_4$ & $(\Fred^{(0,\underline{2})}_*, \fq_\Theta) \ (\cong (\Fred^{(\underline{1},1)}_*, \fq_\Xi))$ & $(U^{(0,\underline{1})}_\cpt, \fq)$  \rule[0mm]{0mm}{4mm}\\
    $\CC$ & $KO_6$ & $KO_5$ & $(\Fred^{(\underline{1},0)}_*, \fq_\Xi)$ & $(U_\cpt, \fq)$ \rule[0mm]{0mm}{4mm}\\
    $\CI$ & $KO_{-1}$ & $KO_6$ & $(\Fred^{(1,\underline{1})}_*, \fr_\Theta) \ (\cong (\Fred^{(\underline{2},0)}_*, \fq_\Xi))$ & $(U^{(\underline{1},0)}_\cpt, \fq)$  \rule[0mm]{0mm}{4mm}\\ \hline
  \end{tabular}
  }
\end{table}

\subsection{Correspondence}\label{Sect.5.4}
By taking a tensor product of the extension (\ref{seq1}) and $(C(\T^{n-2}),\tau_\T)$, we have the following short exact sequence of \TA s,
\begin{equation*}
0 \to(\K \otimes C(\T^{n-2}), \tau) \to (\TTab \otimes C(\T^{n-2}), \tau) \to (\Sab \otimes C(\T^{n-2}), \tau) \to 0.
\end{equation*}
Let consider the following diagram containing the boundary map of $24$-term exact sequence for $KO$-theory associated with this sequence:
\[\xymatrix{
KO_{i(\spadesuit)}(\Sab \otimes C(\T^{n-2}), \tau) \ar[r]^{\hat{\partial}_{i(\spadesuit)} \hspace{3mm}} \ar[d]_{L} & KO_{i(\spadesuit)-1}(\K(\HHab) \otimes C(\T^{n-2}), \tau)
\\
[(\T^{n-2},\zeta), \Fred^{\spadesuit}]_{\Z_2} \ar[r]^{\cong} & [(\T^{n-2},\zeta), F^{\spadesuit}]_{\Z_2} \ar[u]^{\cong}_{\exp}
}\]
where $F^{\spadesuit}$ is the $\Z_2$-subspace of $\Fred^{\spadesuit}$ as in Appendix~\ref{Sect.A}, whose inclusion $F^{\spadesuit} \hookrightarrow \Fred^{\spadesuit}$ is the $\Z_2$-homotopy equivalence.
Maps $L$ and $\exp$ are as follows.
\begin{itemize}
\item When $i(\spadesuit)$ is odd, for $[u] \in KO_{i(\spadesuit)}(\Sab \otimes C(\T^{n-2}),\tau)$, we take a lift $a$ of $u$ and consider the matrix $A$ as in Definition~\ref{boundarymap}, and set $L([u]) = [A]$. The map $\exp$ is defined as in Definition~\ref{boundarymap}.
\item When $i(\spadesuit)=0,4$, for $[u] \in KO_{i(\spadesuit)}(\Sab \otimes C(\T^{n-2}),\tau)$, we take a self-adjoint lift $a$ of $u$ as in Definition~8.3 of \cite{BL16} and set $L([u]) = [a]$. The map $\exp$ is defined by $\exp([a']) = [-\exp(\pi i a')]$.
\item When $i(\spadesuit)=2,6$, for $[u] \in KO_{i(\spadesuit)}(\Sab \otimes C(\T^{n-2}),\tau)$, we take a self-adjoint lift $a$ of $u$ as in Definition~8.3 of \cite{BL16} and set $L([u]) = [ia]$. The map $\exp$ is defined by $\exp([a']) = [-\exp(\pi a')]$.
\end{itemize}
In each case, the map $\exp$ is an isomorphism by Proposition~\ref{htpyequiv} and Sect.~\ref{A2}.
For boundary maps $\hat{\partial}_{i(\spadesuit)}$, we use its expressions through exponentials (see \cite{BL16} for even $i(\spadesuit)$ and Appendix~\ref{Sect.B} for odd $i(\spadesuit)$) and the diagram commutes.
Note that, by Proposition~\ref{surjectivity1}, the boundary map $\hat{\partial}_{i(\spadesuit)}$ is surjective.
The following is the main result of this section.
\begin{theorem}\label{theorem1}
$\hat{\partial}_{i(\spadesuit)}(\I_\BE^{n,2,\spadesuit}(H)) = \I_\Corner^{n,2, \spadesuit}(H)$.
\end{theorem}
\begin{proof}
The operator $\hat{H}^{\alpha,\beta}$ is a self-adjoint lift of $(H^\alpha,H^\beta)$ and preserves the symmetries of the class $\spadesuit$.
Therefore, we have $L(\I_\BE^{n,2,\spadesuit}(H)) = [c(\spadesuit) \hat{H}^{\alpha,\beta}]$ and the results follows from the commutativity of the above diagram.
\end{proof}

\begin{remark}[Relation with bulk weak invariants]\label{relweak1}
Under Assumption~\ref{sgc1}, the bulk Hamiltonian $H$ is also invertible.
When we take $H$ in place of $h = (H^\alpha, H^\beta)$ and define the unitary $u'$ as in Sect.~\ref{Sect.5.2}, this unitary defines an element $[u']$ in $KO_{i(\spadesuit)}(C(\T^n), \tau_\T)$, which classifies bulk invariants in class $\spadesuit$.
A relation between our gapped invariants $\I_\BE^{n,2,\spadesuit}(H)$ and these bulk invariants can be discussed through the map $(\sigma \otimes 1)_* \colon KO_{i(\spadesuit)}(\Sab \otimes C(\T^{n-2}), \tau) \to KO_{i(\spadesuit)}(C(\T^{n-2}), \tau_\T)$, which maps $[u]$ to $[u']$ and we briefly mention its consequences here.
Our gapped invariant $\I_\BE^{n,2,\spadesuit}(H)$ has no relation with bulk invariants in the sense that, under Assumption~\ref{sgc1}, bulk invariants are trivial except for a component corresponding to $KO$-groups of a point\footnote{This component maps to zero by the boundary map $\hat{\partial}_{i(\spadesuit)}$ and has no relation with gapless corner invariants.} and for the cases when $\alpha$, $\beta$ are both rational (or $\pm \infty$) and $t = -ps+qr$ is even.
By Remark~\ref{Remark4.11}, when $\alpha$ and $\beta$ takes these values, some bulk weak invariants can be nontrivial, though they have no relation with $\I_\Corner^{n,2,\spadesuit}(H)$, which can be seen by comparing the above map $(\sigma \otimes 1)_*$ and the boundary map $\hat{\partial}_{i(\spadesuit)}$.
\end{remark}

\begin{remark}[Convex and concave corners]
When we fix $\alpha$ and $\beta$, there exist two models of corners: convex and concave corners ($\HHab$ and $\HHs$).
We have discussed convex corners though, as in \cite{Hayashi3}, similar results also hold for concave corners by using (\ref{seq2}) in place of (\ref{seq1}) in our discussion.
By Remark~\ref{remconcave}, the gapless invariants of these two are related by the factor $-1$.
\end{remark}

\subsection{Higher-Codimensional Cases}\label{Sect.5.5}
Let $n$ and $k$ be positive integers satisfying $3 \leq k \leq n$.
In this subsection, we consider $n$-D system with a codimension-$k$ corner.
Let $d=n-k$.
We consider a continuous map $\T^n \to \mathrm{Herm}(V)$ and the bounded linear self-adjoint operator $H$ on $l^2(\Z^n)$ generated by this map, which is our model of the bulk.
We next introduce models of corners of codimension $k-1$ whose intersection makes a codimension $k$ corner.
For this, we choose $d$ variables $t_{j_1}, \ldots, t_{j_{d}}$ in $t_1, t_2, \ldots, t_n$ and consider the partial Fourier transform in these $d$ variables to obtain a continuous family of self-adjoint operators $\{ H(\bt)\}_{\bt \in \T^{d}}$ on $l^2(\Z^k; V)$.
On the Hilbert space $l^2(\Z^k;V) \cong (l^2(\Z) \otimes \cdots \otimes l^2(\Z)) \otimes V$,
we consider projections $P_k = (P_{\geq 0} \otimes \cdots \otimes P_{\geq 0}) \otimes 1_V$, and $P_{k,i} = (P_{\geq 0} \otimes \cdots \otimes 1 \otimes \cdots \otimes P_{\geq 0}) \otimes 1_V$ for $1 \leq i \leq k$ where inside the brackets is the tensor products of $P_{\geq 0}$ except for the $i$-th tensor product replaced by the identity.
By using these projections, we define the following operators:
\begin{equation*}
	H^c(\bt) = P_k H(\bt) P_k, \ \
	H_i(\bt) = P_{k,i} H(\bt) P_{k,i},
\end{equation*}
for $1 \leq i \leq k$ and for $\bt \in \T^{d}$.
These two are our model for a codimension $k$ corner and codimension $k-1$ corners, respectively.
When we fix a basis on $V$, we have $(H_1(\bt), \ldots, H_k(\bt)) \in M_N(\mathcal{S}^k)$ by the construction.
We assume the following condition in this subsection.

\begin{assumption}[Spectral Gap Condition]\label{sgc2}
We assume that our models for codimension $k-1$ corners $H_1, \ldots, H_k$ are invertible.
\end{assumption}

Under this assumption, the model for the bulk, surfaces and corners of codimension less than $k$, whose intersection makes our codimension-$k$ corner, are invertible.
As in Sect.~\ref{Sect.5.1}, let $h = (H_1, \ldots, H_k)$.
\begin{definition}\label{gappdinv2}
For a Hamiltonian in class $\spadesuit =$ $\AI$, $\BDI$, $\DD$, $\DIII$, $\AII$, $\CII$, $\CC$ or $\CI$ satisfying Assumption~\ref{sgc2},
let $u$ be a unitary defined by using this $h$ in place of that in Sect.~\ref{Sect.5.2}.
As in Sect.~\ref{Sect.5.2}, this unitary $u$ satisfies the relation $\mathcal{R}_{i(\spadesuit)}$ where $i(\spadesuit)$ is as indicated in Table~\ref{label}.
We denote its class $[ u ]$ in the $KO$-group $KO_{i(\spadesuit)}(\mathcal{S}^k \otimes C(\T^{d}), \tau)$ by $\I_\gapped^{n,k, \spadesuit}(H)$.
\end{definition}
The $KO$-groups $KO_{i}(\mathcal{S}^k \otimes C(\T^{d}), \tau)$ are computed by using Proposition~\ref{prop2}.
For each $\bt  \in \T^{d}$, the operator $H^c(\bt)$ is Fredholm by Corollary~\ref{nFredholm}.
\begin{definition}\label{gapless2}
For $\spadesuit=$ AI, BDI, D, DIII, AII, CII, C or CI, let $i(\spadesuit)$, $c(\spadesuit)$ and $\Fred^{\spadesuit}$ be numbers and the $\Z_2$-space as in Tables~\ref{label} and \ref{Table1}.
For a Hamiltonian $H$ in class $\spadesuit$ satisfying Assumption~\ref{sgc2}, we denote $\I_\mathrm{Gapless}^{n,k, \spadesuit}(H)$ for the class $[ c(\spadesuit)  H^c]$ in $KO_{i(\spadesuit)-1}(C(\T^{d}), \tau_\T)$.
We call $\I_\mathrm{Gapless}^{n,k, \spadesuit}(H)$ the {\em gapless corner invariant}.
\end{definition}
We next discuss a relation between these two topological invariants.
As in Sect.~\ref{Sect.5.4}, we consider a tensor product of the extension (\ref{multimainreal}) and $(C(\T^{d}), \tau_\T)$ and let $\partial_{i(\spadesuit)} \colon KO_{i(\spadesuit)}(\mathcal{S}^k \otimes C(\T^{d}), \tau) \to KO_{i(\spadesuit)-1}(\K^k \otimes C(\T^{d}), \tau)$ be the boundary map associated with it expressed through exponentials.
Since $H^c$ is a self-adjoint lift of $(H_1, \ldots, H_k)$, the following relation holds, as in Theorem~\ref{theorem1}.
\begin{theorem}
\label{theorem2}
$\partial_{i(\spadesuit)}(\I_\gapped^{n,k,\spadesuit}(H)) = \I_{\mathrm{Gapless}}^{n,k, \spadesuit}(H).$
\end{theorem}

\begin{remark}
As in Remark~\ref{relweak1}, under Assumption~\ref{sgc2}, some gapped invariants related to corner states for corners of codimension $< k$ are also defined, though, by Remark~\ref{relationhigher}, they are trivial except for a component corresponding to $KO$-groups of a point.
\end{remark}
\begin{remark}
Gapless corner invariants for each systems are elements of the group $KO_i(C(\T^{d}), \tau_\T) \cong \bigoplus_{j=0}^d \binom{d}{j}KO_{i-j}(\C,\id)$.
As in the case of (first-order) topological insulators \cite{Kit09}, we call the component $KO_{i-d}(\C,\id)$ {\em strong} and others {\em weak}.

Complex cases can also be discussed in a similar way\footnote{In \cite{Hayashi2,Hayashi3}, there is a mistake in the computations of the group $K_0(\Sab)$ in the case where $\alpha$ and $\beta$ are rational numbers (there is a torsion part in general, as in $KO_0(\Sab,\tau_{\mathcal{S}})$ computed in Sect.~$4$), which is correctly stated in \cite{Parkthesis}. The author would like to thank Guo Chuan Thiang for pointing this mistake out.}.
For class A systems with a codimension $\geq 3$ corner, under the Assumption~\ref{sgc2}, we define gapped and gapless invariants as elements of $K_0(\mathcal{S}^k \otimes C(\T^d))$ and $K_1(C(\T^d))$, respectively, and the boundary map $\partial_0 \colon K_0(\mathcal{S}^k \otimes C(\T^d)) \to K_1(\K^k \otimes C(\T^d))$ associated with (\ref{multimainreal}) relates these two, which is surjective by Proposition~\ref{surjectivity2}.
In class AIII systems, we use $\partial_1 \colon K_1(\mathcal{S}^k \otimes C(\T^d)) \to K_0(\K^k \otimes C(\T^d))$ instead.
Gapless corner invariants takes value in $K_{i}(C(\T^{d})) \cong \bigoplus_{j=0}^d \binom{d}{j}K_{i-j}(\C)$, and we call the component $K_{i-d}(\C)$ {\em strong} and others {\em weak}.
\end{remark}
Strong invariants for each system are classified in Table~\ref{ptHOTI}.

\subsection{Numerical Corner Invariants}\label{Sect.5.6}
Our gapless corner invariants are defined as elements of some $KO$-group.
In this subsection, we introduce $\Z$- or $\Z_2$-valued numerical corner invariants for our systems in cases where $k=n$ and $k=n-1$ to make the relation between our gapped invariants and corner states more explicit.
From Table~\ref{Table1}, we discuss Hamiltonians in classes $\BDI$, $\DD$, $\DIII$, and $\CII$ when $k=n$ and $\DD$, $\DIII$, $\AII$ and $\CII$ when $k=n-1$ satisfying our spectral gap condition.

\subsubsection{Case of $k=n$}
In this case, our model of the corner  $H^c$ is a self-adjoint Fredholm operator which has some symmetry corresponding to its Altland--Zirnbauer class\footnote{In what follows, we also write $H^c$ for $\hat{H}^{\alpha,\beta}$ in $k=2$ case.}.
An appropriate definition of numerical topological invariants is introduced in \cite{AS69} and we put them in our framework.

In class BDI, the operator $H^c$ is an element of the fixed point set $(\Fred^{(0,\underline{2})}_*)^{\fr_\Theta}$ of the involution $\fr_\Theta$,
where the Clifford action of $\Cl_{0,1}$ on the Hilbert space is given by $\epsilon_1 = \Pi$ (see also Lemma~\ref{htpyspecial} and Remark~\ref{RemarkA12}).
We express $\Pi$ and $\Theta$ as in (\ref{symBDI}) and express $H^c$ as follows.
\begin{equation}\label{decomposition}
H^c = \left(
\begin{array}{cc}
           0&(h^c)^*\\
           h^c& 0
\end{array}
\right).
\end{equation}
The operator $h^c$ is a Fredholm operator that commutes with $\fC$ and thus is a real Fredholm operator.
Its Fredholm index is
\begin{equation*}
	\ind(h^c) = \rank_\C \Ker(h^c) - \rank_\C \Coker(h^c) = \mathrm{Tr} (\Pi |_{\Ker(H^c)}),
\end{equation*}
where the right-hand side is the trace of $\Pi$ restricted to $\Ker(H^c)$.
The Fredholm index induces an isomorphism $\ind^\BDI \colon [(\pt, \id), (\Fred^{(0,\underline{2})}_*, \fr_\Theta)]_{\Z_2} \to \Z$.

In class D, $i H^c$ commutes with the real structure $\Xi$ and is a real skew-adjoint Fredholm operator.
Its mod $2$ index \cite{AS69} is
\begin{equation*}
\ind_1(i H^c) = \rank_\C \Ker (H^c) \mod 2,
\end{equation*}
which induces the isomorphism
$\ind^\DD \colon [(\pt, \id), (\Fred^{(\underline{1}, 0)}_*, \fr_\Xi)]_{\Z_2} \to \Z_2.$

In class $\DIII$, $H^c$ is an element of $(\Fred^{(1,\underline{1})}_*)^{\fq_\Theta}$, where the action of $\Cl_{1,0}$ is given by $e_1 = i\Pi$.
The operator $iH^c$ and $e_1$ commute with the real structure $\Xi$; thus, $iH^c$ is a real skew-adjoint Fredholm operator that anticommutes with $e_1$.
Its mod $2$ index \cite{AS69} is
\begin{equation*}
\ind_2(iH^c) = \frac{1}{2} \rank_\C \Ker (H^c) \mod 2,
\end{equation*}
which induces the isomorphism
$\ind^\DIII \colon [(\pt, \id),(\Fred^{(1,\underline{1})}_*, \fq_\Theta)]_{\Z_2} \to \Z_2.$

In class $\CII$, the operator $H^c$ is an element of $(\Fred^{(0,\underline{2})}_*)^{\fq_\Theta}$, where the Clifford action of $\Cl_{0,1}$ is given by $\epsilon_1 = \Pi$.
We express $\Theta$ and $\Pi$ as in (\ref{symCII}) and express $H^c$ as in (\ref{decomposition}).
The operator $h^c$ commutes with $\fJ$ and is a quaternionic Fredholm operator.
Its Fredholm index $\ind(h^c)$ is an even integer that induces an isomorphism
$\ind^\CII \colon [(\pt, \id), (\Fred^{(0,\underline{2})}_*, \fq_\Theta)]_{\Z_2} \to 2\Z$.

\begin{definition}
For $n$-D systems with codimension-$n$ corners in classes $\BDI$, $\DD$, $\DIII$ and $\CII$, we define the {\em numerical corner invariant} as follows.
\begin{itemize}
\item In class $\BDI$, let $\cN^{n,n,\BDI}_\mathrm{Gapless}(H) = \ind(h^c) \in \Z$.
\item In class $\DD$, let $\cN^{n,n,\DD}_\mathrm{Gapless}(H) = \ind_1(i H^c) \in \Z_2$.
\item In class $\DIII$, let $\cN^{n,n,\DIII}_\mathrm{Gapless}(H) = \ind_2(i H^c) \in \Z_2$.
\item In class $\CII$, let $\cN^{n,n,\CII}_\mathrm{Gapless}(H) = \ind(h^c) \in 2\Z$.
\end{itemize}
\end{definition}
Note that by these definitions, they are images of gapless corner invariants $\I^{n,n,\spadesuit}_\mathrm{Gapless}(H)$ for each class $\spadesuit = \BDI$, $\DD$, $\DIII$ and $\CII$ through the isomorphism $\ind^\spadesuit$.
In each case, the numerical corner invariant is computed through $\Ker(H^c)$ and is related to the number of corner states.

\subsubsection{Case of $k=n-1$}
In this case, $\{ H^c(\bt)\}_{\bt \in \T}$ is a continuous family of self-adjoint Fredholm operators preserving some symmetry.
The numerical corner invariants are given by using ($\Z$-valued) spectral flow \cite{APS3} and its $\Z_2$-valued variants \cite{NSB15,CPSB19}.
We first review $\Z$- and $\Z_2$-valued spectral flow.

Spectral flow is, roughly speaking, the net number of crossing points of eigenvalues of a continuous family of self-adjoint Fredholm operators with zero \cite{APS3}.
The following definition of spectral flow is due to Phillips~\cite{Ph96}.
\begin{definition}[Spectral flow]
Let $A \colon [-1,1] \rightarrow \Fred^{(0,\underline{1})}_*$ be a continuous map.
We choose a partition $-1 = s_0 < s_1 < \cdots < s_n = 1$ and positive numbers $c_1, c_2, \ldots, c_n$ so that for each $i = 1, 2, \ldots, n$, the function $t \mapsto \chi_{[-c_i, c_i]}(A_s)$ is continuous and finite rank on $[s_{i-1}, s_i]$, where $\chi_{[a, b]}$ is the characteristic function of $[a, b]$.
We define the {\em spectral flow} of $A$ as follows.
\begin{equation*}
	\csf(A) = \sum^{n}_{i=1} ( \rank_\C(\chi_{[0, c_i]}(A_{s_i})) - \rank_\C(\chi_{[0, c_i]}(A_{s_{i-1}}) ) \in \Z.
\end{equation*}
\end{definition}
Spectral flow is independent of the choice made and depends only on the homotopy class of the path $A$ leaving the endpoints fixed.
Thus the spectral flow induces a map $\csf \colon [\T, \Fred^{(0,\underline{1})}_*] \to \Z$ which is a group isomorphism.

We next discuss $\Z_2$-valued spectral flow.
Let $\zeta_0$ be an involution on the interval $[-1,1]$ given by $\zeta_0(s) = -s$.
Let $A$ be a $\Z_2$-map from $([-1,1], \zeta_0)$ to $(\Fred^{(0,\underline{1})}_*, \fq)$.
Then, the spectrum $\mathrm{sp}(A_s)$ of $A_s$ is symmetric with respect to $\zeta_0$, and roughly speaking, $\Z_2$-valued spectral flow counts the mod $2$ of the net number of pairs of crossing points of $\mathrm{sp}(A_s)$ with zero.
$\Z_2$-valued spectral flow is studied in \cite{NSB15,CPSB19,NHN19,BCLR20} and we give one definition following \cite{Ph96,NSB15}.

\begin{definition}[$\Z_2$-Valued Spectral Flow] \label{qsf}
Let $A \colon ([-1,1], \zeta_0) \rightarrow (\Fred^{(0,\underline{1})}_*, \fq)$ be a $\Z_2$-map.
We choose a partition $0 = s_0 < s_1 < \cdots < s_n = 1$ of $[0,1]$ and positive numbers $c_1, c_2, \ldots, c_n$ so that for each $i = 1, 2, \ldots, n$, the map $t \mapsto \chi_{[-c_i, c_i]}(A_s)$ is continuous and finite rank on $[s_{i-1}, s_i]$.
We define the {\em $\Z_2$-valued spectral flow} $\qsf(A)$ of $A$ as follows.
\begin{equation*}
	\qsf(A) = \sum^{n}_{i=1} ( \rank_\C(\chi_{[0, c_k]}(A_{s_i})) + \rank_\C(\chi_{[0, c_i]}(A_{s_{i-1}})))
	\hspace{2mm} \bmod{2} \ \in \Z_2.
\end{equation*}
\end{definition}
$\Z_2$-valued spectral flow is independent of the choice made and depends only on the $\Z_2$-homotopy class of the $\Z_2$-map $A$ leaving the endpoints fixed or leaving these points in the $\Z_2$-fixed point set $(\Fred^{(0,\widehat{1})}_*)^\fq$. Thus $\Z_2$-valued spectral flow induces a group homomorphism $\qsf \colon [(\T, \zeta_0), (\Fred^{(0,\underline{1})}_*, \fq)]_{\Z_2} \to \Z_2$.
By Appendix~\ref{Sect.A}, the $\Z_2$-homotopy classes $[(\T, \zeta_0), (\Fred^{(0,\underline{1})}_*, \fq)]_{\Z_2}$ is isomorphic to
$KO_{3}(C(\T),\tau_\T) \cong \Z_2$.
\begin{example}\label{examsf2}
On $\C^2$, let consider a family of self-adjoint operators given by $B_s = \diag(s,-s)$ for $s \in [-1, 1]$, and an antiunitary $j$ given by $j(x,y) = (-\bar{y},\bar{x})$.
Then, we have a $\Z_2$-map $B \colon ([-1,1],\zeta_0) \to (M_2(\C), \Ad_{j})$ whose $\Z_2$-valued spectral flow $\qsf(B)$ is one.
We extend this finite-dimensional example to an infinite-dimensional one to give an example of a family parametrized by the circle of nontrivial $\Z_2$-valued spectral flow.
Let $\mathcal{V}$ be a separable infinite-dimensional complex Hilbert space equipped with a quaternionic structure $q$.
On $\mathcal{V}' = \C^2 \oplus \mathcal{V} \oplus \mathcal{V}$, we consider a quaternionic structure $q' = j \oplus q \oplus q$ and a family self-adjoint Fredholm operators given by $C_s = \diag(B_s,1_\mathcal{V},-1_{\mathcal{V}})$.
Let $U^{(0,\underline{1})}_*(\mathcal{V}')$ the space of unitaries on $\mathcal{V}'$ whose spectrum is $\{\pm 1\}$ equipped with the norm topology.
Then, its endpoints $C_{\pm 1}$ are contained in $U^{(0,\underline{1})}_*(\mathcal{V}')$.
Through an identification $(\mathcal{V} \oplus \mathcal{V}, q \oplus q) \cong (\mathcal{V}', q')$, the operator $\diag(1_{\mathcal{V}}, -1_{\mathcal{V}})$ gives an element $v_0 \in U^{(0,\underline{1})}_*(\mathcal{V}')$ which satisfies $\Ad_q'(v_0) = v_0$.
The space $U^{(0,\underline{1})}_*(\mathcal{V}')$ is homeomorphic to the homogeneous space $U(\mathcal{V}')/ (U(\C \oplus \mathcal{V}) \times U(\C \oplus \mathcal{V}))$, which is contractible by Kuiper's theorem \cite{Kui65}.
Thus, there is a path $l \colon [0,1] \to U^{(0,\underline{1})}_*(\mathcal{V}')$ whose endpoints are $l(0) = v_0$ and $l(1)=C_1$.
We extend $l$ to a $\Z_2$-map $l' \colon ([-1, 1],\zeta_0) \to (U^{(0,\underline{1})}_*(\mathcal{V}), \Ad_{q'})$ by $l'(s) = \Ad_{q'}(l(-s))$ for $s \in [-1,0]$.
Since $l'(\pm 1) = C_{\pm 1}$, we connect the endpoints of $C$ and $l'$ to construct a $\Z_2$-map $C' \colon (\T,\zeta) \to (\Fred^{(0,\widehat{1})}_*, \fq')$, where $\fq' = \Ad_{q'}$.
Then, $\qsf(C') = \qsf(B) = 1$.
\end{example}

In class $\DD$, we have a $\Z_2$-map $i H^c \colon (\T, \zeta) \to (\Fred^{(\underline{1},0)}_*, \fr_\Xi)$.
The $\Z_2$-homotopy classes $[(\T,\zeta), (\Fred^{(\underline{1}, 0)}_*, \fr_\Xi)]_{\Z_2}$ is isomorphic to $KO_{1}(C(\T),\tau_\T)\cong \Z_2 \oplus \Z.$
By forgetting the $\Z_2$-actions and multiplying by $-i$, there is a map from
$[(\T,\zeta), (\Fred^{(\underline{1}, 0)}_*, \fr_\Xi)]_{\Z_2}$ to $[\T, \Fred^{(0,\underline{1})}_*]$.
Combined with this map and the spectral flow $\csf \colon [\T, \Fred^{(0,\underline{1})}_*] \to \Z$, we obtain a homomorphism
\begin{equation*}
	\csf^\DD \colon [(\T,\zeta), (\Fred^{(\underline{1}, 0)}_*, \fr_\Xi)]_{\Z_2} \to \Z, \ \ [A] \mapsto \csf(-iA).
\end{equation*}
\begin{example}\label{examD}
Let $B' \hspace{-0.5mm} \colon \hspace{-0.5mm} ([-1,1], \zeta_0) \hspace{-0.5mm} \to \hspace{-0.5mm} (\C, \Ad_c) $ be a $\Z_2$-map defined by $B'_s = is$.
Then, $\csf^\DD(B')$ is defined and $\csf^\DD(B') = 1$.
\end{example}

In class $\DIII$, $H^c \colon (\T, \zeta) \to (\Fred^{(1,\underline{1})}_*, \fq_\Theta)$ is a $\Z_2$-map, where the action of the Clifford algebra on the right-hand side is given by $e_1 = i\Pi$.
The $\Z_2$-homotopy classes $[(\T,\zeta), (\Fred^{(1,\underline{1})}_*, \fq_\Theta)]_{\Z_2}$ is isomorphic to $KO_{2}(C(\T),\tau_\T)
	\cong \Z_2 \oplus \Z_2$.
Since $(\Fred^{(1,\underline{1})}_*, \fq_\Theta)$ is a $\Z_2$-subspace of $(\Fred^{(0,\underline{1})}_*, \fq_\Theta)$, the inclusion induces a map from
$[(\T,\zeta), (\Fred^{(1,\underline{1})}_*, \fq_\Theta)]_{\Z_2}$ to $[(\T,\zeta), (\Fred^{(0,\underline{1})}_*, \fq_\Theta)]_{\Z_2}$.
Combined with the $\Z_2$-valued spectral flow, we obtain the following map:
\begin{equation*}
	\csf^\DIII \colon [(\T,\zeta), (\Fred^{(1,\underline{1})}_*, \fq_\Theta)]_{\Z_2} \to \Z_2, \ \ [A] \mapsto \qsf(A).
\end{equation*}
For $b = 1$ or $-1$, let $i_b$ be the inclusion $\{ b\} \hookrightarrow \T$, and let $w_b$ be the composite of the following maps:
\begin{equation*}
	[(\T,\zeta), (\Fred^{(1,\underline{1})}_*, \fq_\Theta)]_{\Z_2} \overset{i_b^*}{\longrightarrow} [(\{ \pm 1 \}, \id) , (\Fred^{(1,\underline{1})}_*, \fq_\Theta)]_{\Z_2} \overset{\ind_2}{\longrightarrow} \Z_2.
\end{equation*}
\begin{example}\label{examDIII}
Let $j$, $\mathcal{V}$, $\mathcal{V}'$, $q$, $q'$, $B_s$ and $C_s$ be as in Example~\ref{examsf2}.
Let
$e_1 = \left(
\begin{array}{cc}
           0&i\\
           i& 0
\end{array}
\right)$
and
$e'_1 = \left(
\begin{array}{cc}
           0& \hspace{-1mm} -1_\mathcal{V}\\
           1_\mathcal{V}& 0
\end{array}
\right)$
which gives a $\Cl_{1,0}$-module structure on $\C^2$ and $\mathcal{V} \oplus \mathcal{V}$, respectively.
Then, $C_s = \diag(B_s, 1, -1)$ gives a $\Z_2$-map from $([-1,1],\zeta_0)$ to $(\Fred^{(1,\underline{1})}_*, \fq')$.
The operator $C_1$ is contained in the space of self-adjoint unitaries on $\mathcal{V}'$ that anticommutes with $e_1 \oplus e_1'$.
As in \cite{AS69}, this space of unitaries is contractible by Kuiper's theorem.
We embed $[-1,1]$ into $\T$ by $s \mapsto \exp(\frac{\pi i s}{2})$ and, as in Example~\ref{examsf2}, extend $C$ onto $\T$ through this contractible space of unitaries to obtain a $\Z_2$-map $D \colon (\T, \zeta) \to (\Fred^{(1,\underline{1})}_*, \fq')$.
For this example, we have $\csf^\DIII(D) = \csf^\DIII(B) = 1$, $w_+(D)=1$ and $w_-(D) = 0$.
If we take $D'$ as $D'_s = D_{-s}$, then $D'$ is also such a $\Z_2$-map and its invariants are $\csf^\DIII(D') = 1$, $w_+(D')=0$ and $w_-(D') = 1$.
\end{example}

In class $\AII$, $H^c \colon (\T, \zeta) \to (\Fred^{(0,\underline{1})}_*, \fq_\Theta)$ is a $\Z_2$-map and its $\Z_2$-valued spectral flow is defined.
We denote $\csf^{\AII}$ for $\csf_2$.

In class $\CC$, we have a $\Z_2$-map $i H^c \colon (\T, \zeta) \to (\Fred^{(\underline{1},0)}_*, \fq_\Xi)$.
Note that the $\Z_2$-homotopy classes
$[(\T,\zeta), (\Fred^{(\underline{1},0)}_*, \fq_\Xi)]_{\Z_2}$ is isomorphic to $KO_{5}(C(\T),\tau_\zeta) \cong \Z$.
By forgetting the $\Z_2$-actions and multiplying by $-i$, there is a map from $[(\T,\zeta), (\Fred^{(\underline{1},0)}_*, \fq_\Xi)]_{\Z_2}$ to $[\T, \Fred^{(0,\underline{1})}_*]$.
Combined with the spectral flow, we obtain a homomorphism
\begin{equation*}
	\csf^\CC \colon [(\T,\zeta),  (\Fred^{(\underline{1},0)}_*, \fq_\Xi)]_{\Z_2} \to 2\Z, \ \ [A] \mapsto \csf(-iA).
\end{equation*}
	Note that  image the image of $\csf^\CC$ are even integers since each eigenspace corresponding to the crossing points of the spectrum of $-iA_t$ with zero has a quaternionic vector space structure given by $\Xi$.
\begin{example}\label{examC}
For $s \in [-1,1]$, let $B''_s = \diag(is,is)$, and let $j$ be the quaternionic structure in Example~\ref{examsf2}.
Then, $B'' \colon ([-1,1], \zeta) \to (M_2(\C), \Ad_{j})$ is a $\Z_2$-map, and we have $\csf^\CC(B'') = \csf(-iB'') = 2$.
\end{example}
\begin{lemma}
\begin{enumerate}
\renewcommand{\labelenumi}{(\arabic{enumi})}
\item $\csf^\AII \hspace{-0.5mm} \colon [(\T, \zeta), (\Fred^{(0,\underline{1})}_*, \fq)]_{\Z_2} \to \Z_2$ is an isomorphism.
\item $\csf^\DD \colon [(\T,\zeta), (\Fred^{(\underline{1}, 0)}_*, \fr)]_{\Z_2} \to \Z$ is surjective.
\item $\csf^\DIII, \ w_+, \ w_- \colon [(\T,\zeta), (\Fred^{(1,\underline{1})}_*, \fq)]_{\Z_2} \to \Z_2$ are surjective
\item $\csf^\CC \colon [(\T,\zeta), (\Fred^{(\underline{1},0)}_*, \fq)]_{\Z_2} \to 2\Z$ is an isomorphism.
\end{enumerate}
\end{lemma}
\begin{proof}
It is sufficient to find examples of $\Z_2$-maps which maps to generators of $\Z$, $\Z_2$ and $2\Z$.
Therefore, (1) and (3) follows from Example~\ref{examsf2} and Example~\ref{examDIII}.
For (2) and (4), we can construct such examples from Example~\ref{examD} and Example~\ref{examC}, as in Example~\ref{examsf2}.
\end{proof}
In class DIII cases, we have three surjections $\csf^\DIII$, $w_+$ and $w_-$ from $\Z_2 \oplus \Z_2$ to $\Z_2$.
There is the following relation between them.
\begin{lemma}
$\csf^\DIII = w_+ + w_-$.
\end{lemma}
\begin{proof}
Let $D$ and $D'$ be $\Z_2$-maps in Example~\ref{examDIII}.
Let $D'' = D \oplus D'$, then we have $\csf^\DIII(D'') = 0$, $w_+(D'')=1$ and $w_-(D'') = 1$.
Invariants $\csf^\DIII$, $w_-$ and $w_+$ for $D$, $D'$ and $D''$ tell that non-trivial three elements in the group $[(\T,\zeta), (\Fred^{(1,\underline{1})}_*, \fq)]_{\Z_2}$ consists of classes of $D$, $D'$ and $D''$.
Therefore, we computed three maps $\csf^\DIII$, $w_-$ and $w_+$, from which the result follows.
\end{proof}

\begin{remark}\label{remweakDIII}
For our class $\DIII$ systems, $\Z_2$-valued spectral flow counts the strong invariant.
This corresponds to one direct summand of $\Z_2 \oplus \Z_2$, while the other corresponds to a weak invariant.
When $w_+ \neq w_-$, the strong invariant is nonzero.
When $w_+ = w_- = 1$, the strong invariant is zero and the weak invariant is nonzero.
When $w_+ = w_- = 0$, both of them are zero.
\end{remark}

\begin{definition}\label{nglinv}
For $n$-D systems with codimension $n-1$ corners in classes $\DD$, $\DIII$, $\AII$ and $\CC$, we define the {\em numerical corner invariant} as follows.
\begin{itemize}
\item In class $\DD$, let $\cN^{n,n-1,\DD}_\mathrm{Gapless}(H) = \csf(H^c) \in \Z$.
\item In class $\DIII$, let $\cN^{n,n-1,\DIII}_\mathrm{Gapless}(H) = \qsf(H^c) \in \Z_2$.
\item In class $\AII$, let $\cN^{n,n-1,\AII}_\mathrm{Gapless}(H) = \qsf(H^c) \in \Z_2$.
\item In class $\CC$, let $\cN^{n,n-1,\CC}_\mathrm{Gapless}(H) = \csf(H^c) \in 2\Z$.
\end{itemize}
\end{definition}
For each of the above classes $\spadesuit$, the numerical invariant $\cN^{n,n-1,\spadesuit}_\mathrm{Gapless}(H)$ is the image of the gapless corner invariant $\I^{n,n-1,\spadesuit}_\mathrm{Gapless}(H)$ through the map $\csf^\spadesuit$.
These numerical invariants account for strong invariants.
\begin{remark}
In Definition~\ref{nglinv}, the numerical corner invariants for both class $\DIII$ and class $\AII$ are defined by using $\Z_2$-valued spectral flow, though these two $\Z_2$ are different from the viewpoint of index theory in the sense that they sit in different Bott clock. A similar remark holds for, e.g., cases of $n=k$ in classes $\BDI$ and $\CII$, where both of these numerical corner invariants are defined as Fredholm indices.
\end{remark}

\subsection{Product Formula}\label{Sect.5.7}
In Sect.~$4$ of \cite{Hayashi2}, a construction of the second-order topological insulators of $3$-D class A systems is proposed, which is given by using the Hamiltonians of $2$-D class A and $1$-D class AIII topological insulators.
In this subsection, we generalize this construction to other pairs in the Altland--Zirnbauer classification.
This provides a way to construct nontrivial examples of each entry in Table~\ref{ptHOTI} from the Hamiltonians of (first-order) topological insulators\footnote{For the case of $k=2$, the construction is restricted to $\alpha = 0$ and $\beta = \infty$ case.}.
For this purpose, we use an exterior product of topological $KR$-groups \cite{At66}.

For $j=1,2$, let $H_j$ be a bulk Hamiltonian of an $n_j$-D $k_j$-th order topological insulator\footnote{that in Sect.~\ref{Sect.5.1} satisfying Assumption~\ref{sgc1} when $k_j=2$ or that in Sect.~\ref{Sect.5.5} satisfying Assumption~\ref{sgc2} when $k_j \geq 3$.
When $k_j = 1$, the bulk Hamiltonian is assumed to be gapped.
When $k_j=2$, we consider the case of $\alpha = 0$ and $\beta = \infty$.} in real Altland--Zirnbauer class $\spadesuit_j$ (AI, BDI, D, DIII, AII, CII, C or CI).
Let $n = n_1+ n_2$, $k=k_1 + k_2$ and $d_j = n_j - k_j$ for $j=1,2$, and let $d = d_1+d_2$.
Corresponding to the class in the Altland--Zirnbauer classification (for which we write $\spadesuit_j$) to which the Hamiltonian belongs, it preserves the symmetries as (even/odd) TRS, (even/odd) PHS or chiral symmetry.
We write $\Theta_j$, $\Xi_j$ and $\Pi_j$ for the symmetry operator for $H_j$.
As in Sect.~\ref{Sect.5}, the models of corners $H_i^c$ lead to a continuous family of self-adjoint or skew-adjoint Fredholm operators and defines an element of the $KO$-group $KO_{i'(\spadesuit_j)}(C(\T^{d_j}), \tau)$ where $i'(\spadesuit_j) = i(\spadesuit_j) -1$.
As in Appendix~\ref{KOproduct}, we have an exterior product of $KO$-groups
\begin{equation*}
	KO_{i'(\spadesuit_1)}(C(\T^{d_1}), \tau_\T) \times KO_{i'(\spadesuit_2)}(C(\T^{d_2}), \tau_\T)
		\to
			KO_{i'(\spadesuit_1) + i'(\spadesuit_2)}(C(\T^{d}), \tau_\T),
\end{equation*}
described through these Fredholm operators.
By using this form of the product, we obtain an explicit form of the product of the gapless invariants of $H_1$ and $H_2$.
As a result, we can write down a bulk Hamiltonian $H$ of an $n$-D $k$-th order topological insulator of class $\spadesuit$.
The lattice on which we consider $H^c$ as a model of the codimension-$k$ corner is introduced as the product of that\footnote{When $k_j=1$, the lattice is $\Z_{\geq 0} \times \Z^{d_j}$, where $H^c_j$ is the compression of the bulk Hamiltonian onto this half-space. Topological invariants for them are the one discussed in topological insulators. To clarify our sign choices, we mention that they are obtained by applying the discussion in Sect.~\ref{Sect.5} to the Toeplitz extension (\ref{Toeplitz}) in place of (\ref{seq1}) or (\ref{multimain}).} of $H^c_1$ and that of $H^c_2$.
By this construction, we have the following relation between gapless invariants.
\begin{theorem}\label{product}
For the Hamiltonian $H$ indicated in Table~\ref{producttalbe}, we have
\begin{equation*}
		\I_\mathrm{Gapless}^{n_1,k_1, \spadesuit_1}(H_1) \cdot \I_\mathrm{Gapless}^{n_2,k_2, \spadesuit_2}(H_2)
		=
		\I_\mathrm{Gapless}^{n,k, \spadesuit}(H),
\end{equation*}
	where $\cdot$ denotes the exterior product of elements of the $KO$-groups.
\end{theorem}
Note that Theorem~\ref{product} is the product formula at the level of $KO$-group elements and accounts for both strong and weak invariants.
In order to show this theorem, we need to write down the explicit form of $H$.
In the following, we discuss them for some classes.

Let us consider the case where $\spadesuit_1 = \BDI$ and $\spadesuit_2 = \BDI$.
In this case, each $H_j$ has even TRS, even PHS and chiral symmetry.
We now consider the following $n$-dimensional Hamiltonian:
\begin{equation}\label{prodham}
	H = H_1 \otimes 1 + \Pi_1 \otimes H_2,
\end{equation}
which satisfies even TRS given by $\Theta = \Theta_1 \otimes \Theta_2$, even PHS given by $\Xi = \Xi_1 \otimes \Xi_2$ and the chiral symmetry given by $\Pi = \Pi_1 \otimes \Pi_2$.
Thus, the Hamiltonian $H$ belongs to the class $\spadesuit = \BDI$.
The model of the codimension-$k$ corner $H^c$ of $H$ is written by using the model $H^c_j$ of the codimension-$k_j$ corner as follows:
\begin{equation*}
	H^c(\bt_1,\bt_2) = H^c_1(\bt_1) \otimes 1 + \Pi_1 \otimes H^c_2(\bt_2),
\end{equation*}
where $\bt_j$ is an element of the $d_j$-dimensional torus (momentum space) corresponding to a direction parallel to the corner of $H^c_j$ for $j=1,2$.
Note that $(\bt_1, \bt_2)$ constitute the parameter of the $d$-dimensional momentum space in a direction parallel to the corner of $H^c$.
By our assumption, $H^c_j(\bt_j)$ is an element of the space $\Fred^{(0,\underline{2})}_*$ and gives a $\Z_2$-map $(\T^{d_j}, \zeta) \to (\Fred^{(0,\underline{2})}_*, \fr_{\Theta_j})$.
The operator $H^c(\bt_1,\bt_2)$ is the image of the pair $(H^c_1, H^c_2)$ through the map,
\begin{equation*}
	(\Fred^{(0, \underline{2})}_*, \fr_{\Theta_1}) \times (\Fred^{(0, \underline{2})}_*, \fr_{\Theta_2})  \to (\Fred^{(0, \underline{2})}_*, \fr_{\Theta}),
\end{equation*}
in (\ref{FredProdsa}), where the action of $\Cl_{0,1}$ to define the left-hand side is given by $\epsilon_j = \Pi_j$ and that for the right-hand side is given by $\epsilon = \epsilon_1 \otimes \epsilon_2 = \Pi$.
Since this map induces the exterior product of $KO$-groups (Appendix~\ref{KOproduct}),
\begin{equation*}
	KO_{0}(C(\T^{d_1}), \tau_\T) \times KO_{0}(C(\T^{d_2}), \tau_\T) \to KO_{0}(C(\T^{d}), \tau_\T),
\end{equation*}
we obtain Theorem~\ref{product} in this case.

We next consider the case where $\spadesuit_1 = \DIII$ and $\spadesuit_2 = \DD$.
In this case, $H_1$ has odd TRS, even PHS and the chiral symmetry, and $H_2$ has even PHS.
As in Sect.~\ref{Sect.5.3}, $H^c_1(\bt_1)$ belongs to $(\Fred^{(1, \underline{1})}_*, \fq_{\Theta_1})$, and $iH^c_2(\bt_2)$ belongs to $(\Fred^{(\underline{1},0)}_*, \fr_{\Xi_2})$.
By using Proposition~\ref{htpyhomeo}, we identify $(\Fred^{(1, \underline{1})}_*, \fq_{\Theta_1})$ with $(\Fred^{(2, \underline{2})}_*, \fq_{\Theta_1 \oplus \Theta_1})$ and $(\Fred^{(\underline{1},0)}_*, \fr_{\Xi_2})$ with $(\Fred^{(0,\underline{3})}_*, \fr_{\Xi_2 \oplus \Xi_2})$.
We then use the map (\ref{FredProdsa}) of the form
\begin{equation*}
	(\Fred^{(2, \underline{2})}_*, \fq_{\Theta_1 \oplus \Theta_1}) \times (\Fred^{(0, \underline{3})}_*, \fr_{\Xi_2 \oplus \Xi_2}) \to (\Fred^{(2, \underline{3})}_*, \fq'),
\end{equation*}
where $\fq'$ is the conjugation of the fourfold direct sum of $\Theta_1 \otimes \Xi_2$.
By Proposition~\ref{htpyhomeo}, we have the $\Z_2$-homeomorphism $(\Fred^{(2, \underline{3})}_*, \fq') \cong (\Fred^{(0, \underline{1})}_*, \fq_{\Theta_1 \otimes \Xi_2})$.
Thus, we obtain a $\Z_2$-map $H^c \colon (\T^d, \zeta) \to (\Fred^{(0, \underline{1})}_*, \fq_{\Theta_1 \otimes \Xi_2})$ from $H_1^c$ and $H_2^c$ which is a model for the codimension-$k$ corner in class $\spadesuit = \AII$.
Its bulk Hamiltonian $H$ and (odd) TRS operator $\Theta$ is expressed as (\ref{prodham}) and $\Theta = \Theta_1 \otimes \Xi_2$, respectively.
Note that $\Pi_1 = i \Theta_1 \Xi_1$ in class DIII and $\Theta$ commutes with $H$.
Since the map (\ref{FredProdsa}) induces the exterior product of $KO$-groups, Theorem~\ref{product} holds for this class AII Hamiltonian $H$.

The other cases are computed in a similar way, and the results are summarized in Table~\ref{producttalbe}, where we write
\begin{equation*}
H_\bigstar = \left( \begin{array}{cc}
	0 & \hspace{-1.5mm} H_1 \otimes 1 - i \otimes H_2 \hspace{-1.5mm} \\
	\hspace{-1.5mm} H_1 \otimes 1 + i \otimes H_2 \hspace{-1mm} & 0 \end{array} \right), \
\Theta_\clubsuit = \left( \begin{array}{cc}
	\hspace{-1mm} \Theta_1 \otimes \Xi_2 \hspace{-1mm} & 0 \\
	0 & \hspace{-1mm} \Theta_1 \otimes \Xi_2 \end{array} \hspace{-1mm} \right),
\end{equation*}
\begin{equation*}
H_\Box = \left( \begin{array}{cc}
	0 & \hspace{-2mm} -H_1 \otimes i - 1 \otimes H_2 \hspace{-1.5mm} \\
	\hspace{-1.5mm} H_1 \otimes i - 1 \otimes H_2 \hspace{-1mm} & 0 \end{array} \right)\hspace{-0.5mm},
\Theta_\triangle = \left( \begin{array}{cc}
	\hspace{-1.5mm} \Xi_1 \otimes \Theta_2 \hspace{-1mm} & 0 \\
	0 & \hspace{-1mm} \Xi_1 \otimes \Theta_2 \end{array} \hspace{-1.5mm} \right) \hspace{-0.5mm},
\end{equation*}
\begin{equation*}
\Theta_\diamondsuit = \left( \begin{array}{cc} 0&-\Xi_1 \otimes \Xi_2\\ \Xi_1 \otimes \Xi_2&0 \end{array} \right)
\ \ \text{and} \ \ \Theta_\heartsuit = \left( \begin{array}{cc} 0 & \Theta_1 \otimes \Theta_2 \\ \Theta_1 \otimes \Theta_2 & 0 \end{array} \right).
\end{equation*}

\begin{longtable}{|c|c||c|c|c|c|c|}
  \caption{The forms of the Hamiltonians and symmetry operators in class $\spadesuit$ constructed from two pairs of Hamiltonians and symmetry operators in classes $\spadesuit_1$ and $\spadesuit_2$.
  Complex cases are also included \cite{Hayashi2,Hayashi3}.}
  \label{producttalbe}
\\
\hline
  $\spadesuit_1$ & $\spadesuit_2$ & $\spadesuit$ & Hamiltonian ($H$) & TRS ($\Theta$) & PHS ($\Xi$) & Chiral ($\Pi$) \\ \hline
  \endfirsthead
$\AI$ & $\AI$ & $\CI$ & $H_\bigstar$ & $\Theta_\heartsuit$ & $i \Theta \Pi$ & $\diag(1,-1)$  \\ \hline
$\AI$ & $\BDI$ & $\AI$ & $H_1 \otimes  \Pi_2 +1 \otimes H_2$ & $\Theta_1 \otimes \Theta_2$ & --- & ---  \\ \hline
$\AI$ & $\DD$ & $\BDI$ &$H_\bigstar$ & $\Theta_\clubsuit$ & $\Xi= \Theta \Pi$ & $\diag(1,-1)$ \\ \hline
$\AI$ & $\DIII$ & $\DD$ & $H_1 \otimes \Pi_2 + 1 \otimes H_2$ & --- & $\Theta_1 \otimes \Theta_2 \Pi_2$ & ---  \\ \hline
$\AI$ & $\AII$ & $\DIII$ & $H_\bigstar$ & $\Theta_\heartsuit$ & $i \Pi \Theta$ & $\diag(1,-1)$  \\ \hline
$\AI$ & $\CII$ & $\AII$ & $H_1 \otimes  \Pi_2 +1 \otimes H_2$ & $\Theta_1 \otimes \Theta_2$ & --- & ---  \\ \hline
$\AI$ & $\CC$ & $\CII$ & $H_\bigstar$ & $\Theta_\clubsuit$ & $\Xi= -\Theta \Pi$ & $\diag(1,-1)$ \\ \hline
$\AI$ & $\CI$ & $\CC$ & $H_1 \otimes \Pi_2 + 1 \otimes H_2$ & --- & $\Theta_1 \otimes \Theta_2 \Pi_2$ & ---  \\ \hline
$\BDI$ & $\AI$ & $\AI$ & $H_1 \otimes 1 + \Pi_1 \otimes H_2$ & $\Theta_1 \otimes \Theta_2$ & --- & ---  \\ \hline
$\BDI$ & $\BDI$ & $\BDI$ & $H_1 \otimes 1 + \Pi_1 \otimes H_2$ & $\Theta_1 \otimes \Theta_2$ & $\Xi_1 \otimes \Xi_2$ & $\Pi_1 \otimes \Pi_2$  \\ \hline
$\BDI$ & $\DD$ & $\DD$ & $H_1 \otimes 1 + \Pi_1 \otimes H_2$ & --- & $\Xi_1 \otimes \Xi_2$ & ---  \\ \hline
$\BDI$ & $\DIII$ & $\DIII$ & $H_1 \otimes 1 + \Pi_1 \otimes H_2$ & $\Theta_1 \otimes \Theta_2$ & $\Xi_1 \otimes \Xi_2$ & $\Pi_1 \otimes \Pi_2$  \\ \hline
$\BDI$ & $\AII$ & $\AII$ & $H_1 \otimes 1 + \Pi_1 \otimes H_2$ & $\Theta_1 \otimes \Theta_2$ & --- & ---  \\ \hline
$\BDI$ & $\CII$ & $\CII$ & $H_1 \otimes 1 + \Pi_1 \otimes H_2$ & $\Theta_1 \otimes \Theta_2$ & $\Xi_1 \otimes \Xi_2$ & $\Pi_1 \otimes \Pi_2$  \\ \hline
$\BDI$ & $\CC$ & $\CC$ & $H_1 \otimes 1 + \Pi_1 \otimes H_2$ & --- & $\Xi_1 \otimes \Xi_2$ & ---  \\ \hline
$\BDI$ & $\CI$ & $\CI$ & $H_1 \otimes 1 + \Pi_1 \otimes H_2$ & $\Theta_1 \otimes \Theta_2$ & $\Xi_1 \otimes \Xi_2$ & $\Pi_1 \otimes \Pi_2$  \\ \hline
$\DD$ & $\AI$ & $\BDI$ & $H_\Box$ & $\Theta_\triangle$ & $\Xi= \Theta \Pi$ & $\diag(1,-1)$ \\ \hline
$\DD$ & $\BDI$ & $\DD$ & $H_1 \otimes \Pi_2 + 1 \otimes H_2$ & --- & $\Xi_1 \otimes \Xi_2$ & ---  \\ \hline
$\DD$ & $\DD$ & $\DIII$ & $H_\bigstar$ & $\Theta_\diamondsuit$ & $i \Pi \Theta$ & $\diag(1,-1)$ \\ \hline
$\DD$ & $\DIII$ & $\AII$ & $H_1 \otimes \Pi_2 + 1 \otimes H_2$ & $\Xi_1 \otimes \Theta_2$ & --- & ---  \\ \hline
$\DD$ & $\AII$ & $\CII$ & $H_\Box$ & $\Theta_\triangle$ & $\Xi= -\Theta \Pi$ & $\diag(1,-1)$ \\ \hline
$\DD$ & $\CII$ & $\CC$ & $H_1 \otimes \Pi_2 + 1 \otimes H_2$ & --- & $\Xi_1 \otimes \Xi_2$ & ---  \\ \hline
$\DD$ & $\CC$ & $\CI$ & $H_\bigstar$ & $\Theta_\diamondsuit$ & $i \Theta \Pi$ & $\diag(1,-1)$ \\ \hline
$\DD$ & $\CI$ & $\AI$ & $H_1 \otimes \Pi_2 + 1 \otimes H_2$ & $\Xi_1 \otimes \Theta_2$ & --- & ---  \\ \hline
$\DIII$ & $\AI$ & $\DD$ & $H_1 \otimes 1 + \Pi_1 \otimes H_2$ & --- & $\Theta_1 \Pi_1 \otimes \Theta_2$ & ---  \\ \hline
$\DIII$ & $\BDI$ & $\DIII$ & $H_1 \otimes \Pi_2 + 1 \otimes H_2$ & $\Theta_1 \otimes \Theta_2$ & $\Xi_1 \otimes \Xi_2$ & $\Pi_1 \otimes \Pi_2$  \\ \hline
$\DIII$ & $\DD$ & $\AII$ & $H_1 \otimes 1 + \Pi_1 \otimes H_2$ & $\Theta_1 \otimes \Xi_2$ & --- & ---  \\ \hline
$\DIII$ & $\DIII$ & $\CII$ & $H_1 \otimes 1 + \Pi_1 \otimes H_2$ & $\Theta_1\otimes \Theta_2 \Pi_2$ & $\Pi_1 \Theta_1 \otimes \Theta_2$ & $\Pi_1 \otimes \Pi_2$  \\ \hline
$\DIII$ & $\AII$ & $\CC$ & $H_1 \otimes 1 + \Pi_1 \otimes H_2$ & --- & $\Theta_1 \Pi_1 \otimes \Theta_2$ & ---  \\ \hline
$\DIII$ & $\CII$ & $\CI$ & $H_1 \otimes \Pi_2 + 1 \otimes H_2$ & $\Theta_1 \otimes \Theta_2$ & $\Xi_1 \otimes \Xi_2$ & $\Pi_1 \otimes \Pi_2$  \\ \hline
$\DIII$ & $\CC$ & $\AI$ &  $H_1 \otimes 1 + \Pi_1 \otimes H_2$ & $\Theta_1 \otimes \Xi_2$ & --- & ---  \\ \hline
$\DIII$ & $\CI$ & $\BDI$ & $H_1 \otimes 1 + \Pi_1 \otimes H_2$ & $\Theta_1 \otimes \Theta_2 \Pi_2$ & $\Theta_1 \Pi_1 \otimes \Theta_2$ & $\Pi_1 \otimes \Pi_2$  \\ \hline
$\AII$ & $\AI$ & $\DIII$ & $H_\bigstar$ & $\Theta_\heartsuit$ & $i \Pi \Theta$ & $\diag(1,-1)$  \\ \hline
$\AII$ & $\BDI$ & $\AII$ & $H_1 \otimes  \Pi_2 +1 \otimes H_2$ & $\Theta_1 \otimes \Theta_2$ & --- & ---  \\ \hline
$\AII$ & $\DD$ & $\CII$ & $H_\bigstar$ & $\Theta_\clubsuit$ & $\Xi= -\Theta \Pi$ & $\diag(1,-1)$ \\ \hline
$\AII$ & $\DIII$ & $\CC$ & $H_1 \otimes \Pi_2 + 1 \otimes H_2$ & --- & $\Theta_1 \otimes \Theta_2 \Pi_2$ & ---  \\ \hline
$\AII$ & $\AII$ & $\CI$ & $H_\bigstar$ & $\Theta_\heartsuit$ & $i \Theta \Pi$ & $\diag(1,-1)$  \\ \hline
$\AII$ & $\CII$ & $\AI$ & $H_1 \otimes  \Pi_2 +1 \otimes H_2$ & $\Theta_1 \otimes \Theta_2$ & --- & ---  \\ \hline
$\AII$ & $\CC$ & $\BDI$ & $H_\bigstar$ & $\Theta_\clubsuit$ & $\Xi= \Theta \Pi$ & $\diag(1,-1)$ \\ \hline
$\AII$ & $\CI$ & $\DD$ & $H_1 \otimes \Pi_2 + 1 \otimes H_2$ & --- & $\Theta_1 \otimes \Theta_2 \Pi_2$ & ---  \\ \hline
$\CII$ & $\AI$ & $\AII$ & $H_1 \otimes 1 + \Pi_1 \otimes H_2$ & $\Theta_1 \otimes \Theta_2$ & --- & ---  \\ \hline
$\CII$ & $\BDI$ & $\CII$ & $H_1 \otimes 1 + \Pi_1 \otimes H_2$ & $\Theta_1 \otimes \Theta_2$ & $\Xi_1 \otimes \Xi_2$ & $\Pi_1 \otimes \Pi_2$  \\ \hline
$\CII$ & $\DD$ & $\CC$ & $H_1 \otimes 1 + \Pi_1 \otimes H_2$ & --- & $\Xi_1 \otimes \Xi_2$ & ---  \\ \hline
$\CII$ & $\DIII$ & $\CI$ & $H_1 \otimes 1 + \Pi_1 \otimes H_2$ & $\Theta_1 \otimes \Theta_2$ & $\Xi_1 \otimes \Xi_2$ & $\Pi_1 \otimes \Pi_2$  \\ \hline
$\CII$ & $\AII$ & $\AI$ & $H_1 \otimes 1 + \Pi_1 \otimes H_2$ & $\Theta_1 \otimes \Theta_2$ & --- & ---  \\ \hline
$\CII$ & $\CII$ & $\BDI$ & $H_1 \otimes 1 + \Pi_1 \otimes H_2$ & $\Theta_1 \otimes \Theta_2$ & $\Xi_1 \otimes \Xi_2$ & $\Pi_1 \otimes \Pi_2$  \\ \hline
$\CII$ & $\CC$ & $\DD$ & $H_1 \otimes 1 + \Pi_1 \otimes H_2$ & --- & $\Xi_1 \otimes \Xi_2$ & ---  \\ \hline
$\CII$ & $\CI$ & $\DIII$ & $H_1 \otimes 1 + \Pi_1 \otimes H_2$ & $\Theta_1 \otimes \Theta_2$ & $\Xi_1 \otimes \Xi_2$ & $\Pi_1 \otimes \Pi_2$  \\ \hline
$\CC$ & $\AI$ & $\CII$ & $H_\Box$ & $\Theta_\triangle$ & $\Xi= -\Theta \Pi$ & $\diag(1,-1)$ \\ \hline
$\CC$ & $\BDI$ & $\CC$ & $H_1 \otimes \Pi_2 + 1 \otimes H_2$ & --- & $\Xi_1 \otimes \Xi_2$ & ---  \\ \hline
$\CC$ & $\DD$ & $\CI$ & $H_\bigstar$ & $\Theta_\diamondsuit$ & $i \Theta \Pi$ & $\diag(1,-1)$ \\ \hline
$\CC$ & $\DIII$ & $\AI$ & $H_1 \otimes \Pi_2 + 1 \otimes H_2$ & $\Xi_1 \otimes \Theta_2$ & --- & ---  \\ \hline
$\CC$ & $\AII$ & $\BDI$ & $H_\Box$ & $\Theta_\triangle$ & $\Xi= \Theta \Pi$ & $\diag(1,-1)$ \\ \hline
$\CC$ & $\CII$ & $\DD$ & $H_1 \otimes \Pi_2 + 1 \otimes H_2$ & --- & $\Xi_1 \otimes \Xi_2$ & ---  \\ \hline
$\CC$ & $\CC$ & $\DIII$ & $H_\bigstar$ & $\Theta_\diamondsuit$ & $i \Pi \Theta$ & $\diag(1,-1)$ \\ \hline
$\CC$ & $\CI$ & $\AII$ & $H_1 \otimes \Pi_2 + 1 \otimes H_2$ & $\Xi_1 \otimes \Theta_2$ & --- & ---  \\ \hline
$\CI$ & $\AI$ & $\CC$ & $H_1 \otimes 1 + \Pi_1 \otimes H_2$ & --- & $\Theta_1 \Pi_1 \otimes \Theta_2$ & ---  \\ \hline
$\CI$ & $\BDI$ & $\CI$ & $H_1 \otimes \Pi_2 + 1 \otimes H_2$ & $\Theta_1 \otimes \Theta_2$ & $\Xi_1 \otimes \Xi_2$ & $\Pi_1 \otimes \Pi_2$  \\ \hline
$\CI$ & $\DD$ & $\AI$ & $H_1 \otimes 1 + \Pi_1 \otimes H_2$ & $\Theta_1 \otimes \Xi_2$ & --- & ---  \\ \hline
$\CI$ & $\DIII$ & $\BDI$ & $H_1 \otimes 1 + \Pi_1 \otimes H_2$ & $\Theta_1 \otimes \Theta_2 \Pi_2$ & $\Theta_1 \Pi_1 \otimes \Theta_2$ & $\Pi_1 \otimes \Pi_2$  \\ \hline
$\CI$ & $\AII$ & $\DD$ & $H_1 \otimes 1 + \Pi_1 \otimes H_2$ & --- & $\Theta_1 \Pi_1 \otimes \Theta_2$ & ---  \\ \hline
$\CI$ & $\CII$ & $\DIII$ & $H_1 \otimes \Pi_2 + 1 \otimes H_2$ & $\Theta_1 \otimes \Theta_2$ & $\Xi_1 \otimes \Xi_2$ & $\Pi_1 \otimes \Pi_2$  \\ \hline
$\CI$ & $\CC$ & $\AII$ & $H_1 \otimes 1 + \Pi_1 \otimes H_2$ & $\Theta_1 \otimes \Xi_2$ & --- & ---  \\ \hline
$\CI$ & $\CI$ & $\CII$ & $H_1 \otimes 1 + \Pi_1 \otimes H_2$ & $\Theta_1 \otimes \Theta_2 \Pi_2$ & $\Pi_1 \Theta_1 \otimes \Theta_2$ & $\Pi_1 \otimes \Pi_2$  \\ \hline \hline
$\A$ & $\A$ & $\AIII$ & $H_\bigstar$ & --- & --- & $\diag(1,-1)$  \\ \hline
$\A$ & $\AIII$ & $\A$ & $H_1 \otimes \Pi_2 + 1 \otimes H_2$ & --- & --- & --- \\ \hline
$\AIII$ & $\A$ & $\A$ & $H_1 \otimes 1 + \Pi_1 \otimes H_2$ & --- & --- & --- \\ \hline
$\AIII$ & $\AIII$ & $\AIII$ & $H_1 \otimes 1 + \Pi_1 \otimes H_2$ & --- & --- & $\Pi_1 \otimes \Pi_2$  \\ \hline
\end{longtable}

Our product formula (Theorem~\ref{product}) and the graded ring structure of $KO_*(\C,\id)$ (Theorem~6.9 of \cite{ABS64}) lead to the following product formula for numerical corner invariants.
We collect the results here where the form of $H$ is as indicated in Table~\ref{producttalbe}.

\begin{corollary}[Cases of $n=k$]\label{n=k}
The case of $k_1 = n_1$ and $k_2 = n_2$.
\begin{itemize}
\item $\BDI \times \BDI \to \BDI$, \
	$\cN^{n_1,n_1,\BDI}_\mathrm{Gapless}(H_1) \cdot \cN^{n_2,n_2,\BDI}_\mathrm{Gapless}(H_2)
		= \cN^{n,n,\BDI}_\mathrm{Gapless}(H)$.
\item $\BDI \times \DD \to \DD$, \
	$(\cN^{n_1,n_1,\BDI}_\mathrm{Gapless}(H_1) \bmod 2) \cdot \cN^{n_2,n_2,\DD}_\mathrm{Gapless}(H_2)
		= \cN^{n,n,\DD}_\mathrm{Gapless}(H)$.
\item $\BDI \times \DIII \to \DIII$, \\
	$(\cN^{n_1,n_1,\BDI}_\mathrm{Gapless}(H_1) \bmod 2) \cdot \cN^{n_2,n_2,\DIII}_\mathrm{Gapless}(H_2)
		= \cN^{n,n,\DIII}_\mathrm{Gapless}(H)$.
\item $\BDI \times \CII \to \CII$, \
	$\cN^{n_1,n_1,\BDI}_\mathrm{Gapless}(H_1) \cdot \cN^{n_2,n_2,\CII}_\mathrm{Gapless}(H_2)
		= \cN^{n,n,\CII}_\mathrm{Gapless}(H)$.
\item $\DD \times \DD \to \DIII$, \
	$\cN^{n_1,n_1,\DD}_\mathrm{Gapless}(H_1) \cdot \cN^{n_2,n_2,\DD}_\mathrm{Gapless}(H_2)
		= \cN^{n,n,\DIII}_\mathrm{Gapless}(H)$.
\item $\CII \times \CII \to \BDI$, \
	$\cN^{n_1,n_1,\CII}_\mathrm{Gapless}(H_1) \cdot \cN^{n_2,n_2,\CII}_\mathrm{Gapless}(H_2)
		= \cN^{n,n,\BDI}_\mathrm{Gapless}(H)$.
\end{itemize}
\end{corollary}
\begin{corollary}[Cases of $n=k-1$]\label{n=k-1}
The case of $k_1 = n_1$ and $k_2 = n_2-1$.
\begin{itemize}
\item $\BDI \times \DD \to \DD$, \
	$\cN^{n_1,n_1,\BDI}_\mathrm{Gapless}(H_1) \cdot \cN^{n_2,n_2-1,\DD}_\mathrm{Gapless}(H_2)
		= \cN^{n,n-1,\DD}_\mathrm{Gapless}(H)$.
\item $\BDI \times \DIII \to \DIII$, \\
\hspace{5mm} $(\cN^{n_1,n_1,\BDI}_\mathrm{Gapless}(H_1) \bmod 2) \cdot \cN^{n_2,n_2-1,\DIII}_\mathrm{Gapless}(H_2)
		= \cN^{n,n-1,\DIII}_\mathrm{Gapless}(H)$.
\item $\BDI \times \AII \to \AII$, \\
\hspace{5mm} 	$(\cN^{n_1,n_1,\BDI}_\mathrm{Gapless}(H_1) \bmod 2) \cdot \cN^{n_2,n_2-1,\AII}_\mathrm{Gapless}(H_2)
		= \cN^{n,n-1,\AII}_\mathrm{Gapless}(H)$.
\item $\BDI \times \CC \to \CC$, \
	$\cN^{n_1,n_1,\BDI}_\mathrm{Gapless}(H_1) \cdot \cN^{n_2,n_2-1,\CC}_\mathrm{Gapless}(H_2)
		= \cN^{n,n-1,\CC}_\mathrm{Gapless}(H)$.
\item $\DD \times \DD \to \DIII$, \\
	$\cN^{n_1,n_1,\DD}_\mathrm{Gapless}(H_1) \cdot (\cN^{n_2,n_2-1,\DD}_\mathrm{Gapless}(H_2) \bmod 2)
		= \cN^{n,n-1,\DIII}_\mathrm{Gapless}(H)$.
\item $\DD \times \DIII \to \AII$, \
	$\cN^{n_1,n_1,\DD}_\mathrm{Gapless}(H_1) \cdot \cN^{n_2,n_2-1,\DIII}_\mathrm{Gapless}(H_2)
		= \cN^{n,n-1,\AII}_\mathrm{Gapless}(H)$.
\item $\DIII \times \DD \to \AII$, \\
\hspace{5mm} 	$\cN^{n_1,n_1,\DIII}_\mathrm{Gapless}(H_1) \cdot (\cN^{n_2,n_2-1,\DD}_\mathrm{Gapless}(H_2) \bmod 2)
		= \cN^{n,n-1,\AII}_\mathrm{Gapless}(H)$.
\item $\CII \times \DD \to \CC$, \
	$\cN^{n_1,n_1,\CII}_\mathrm{Gapless}(H_1) \cdot \cN^{n_2,n_2-1,\DD}_\mathrm{Gapless}(H_2)
		= \cN^{n,n-1,\CC}_\mathrm{Gapless}(H)$.
\item $\CII \times \CC \to \DD$, \
	$\cN^{n_1,n_1,\CII}_\mathrm{Gapless}(H_1) \cdot \cN^{n_2,n_2-1,\CC}_\mathrm{Gapless}(H_2)
		= \cN^{n,n-1,\DD}_\mathrm{Gapless}(H)$.
\end{itemize}
\end{corollary}
We also have a similar formula by exchanging $H_1$ and $H_2$ (e.g. pairs like $\DD \times \BDI \to \DD$).
Note that in the case of $\CII \times \CII \to \BDI$ in Corollary~\ref{n=k}, we take the product of two even integers, which is necessarily a multiple of four.
A similar remark also holds in the case of $\CII \times \CC \to \DD$ in Corollary~\ref{n=k-1}.

\appendix
\section{$\Z_2$-Spaces of Self-Adjoint/Skew-Adjoint Fredholm Operators and Boersema--Loring's $K$-theory}\label{Sect.A}
In this Appendix, we collect necessary results and notations used in this paper.
The results have been developed in much generality \cite{At66,Wood,AS69,Ka70,Ku16,Gomi17,BCLR20}, and we contain minimal background for this paper focusing on their relation with Boersema--Loring's $K$-theory \cite{BL16}.
In Appendix~\ref{KOproduct}, we introduce some $\Z_2$-spaces of self-adjoint and skew-adjoint Fredholm operators following \cite{AS69}.
Some proofs for known results are contained simply to fix isomorphisms used in this paper (e.g. the derivation of Table~\ref{producttalbe}).
In Appendix~\ref{A2}, we discuss its relation with Boersema--Loring's $K$-theory.
In Appendix~\ref{Sect.B}, inspired by exponential maps in \cite{Wood,AS69}, we write the boundary maps of the $24$-term exact sequence of $KO$-theory in Boersema--Loring's unitary picture through exponentials.
Some of them are already expressed by exponentials in \cite{BL16}; thus, we consider the remaining cases.
This form of boundary maps is useful when we discuss a relation between our gapped invariants and gapless invariants through boundary maps \cite{BL16,Ku15}.

\subsection{$\Z_2$-Spaces of Self-Adjoint/Skew-Adjoint Fredholm Operators}
\label{KOproduct}

For non-negative integers $k$ and $l$, let $\Cl_{k,l}$ be the Clifford algebra that is an associative algebra with unit over $\R$ generated by $k+l$ elements $e_1, \ldots , e_k$ and $\epsilon_1, \ldots , \epsilon_l$, which satisfy $e_i^2= -1$ $(i = 1, \ldots, k)$ and $\epsilon_j^2=1$ $(j = 1, \ldots, l)$ and anticommute with each other.
The following are well-known Clifford algebra isomorphisms \cite{LM89}.
\begin{lemma}\label{Cliff1}
\begin{enumerate}
\renewcommand{\labelenumi}{(\arabic{enumi})}
\item	$\Cl_{k,l+1} \cong \Cl_{l, k+1}$.
\item $\Cl_{k,l} \otimes \Cl_{1,1} \cong \Cl_{k+1,l+1}$.
\item $\Cl_{k,l} \otimes \Cl_{4,0} \cong \Cl_{k+4,l}$ and $\Cl_{k,l} \otimes \Cl_{0,4} \cong \Cl_{k,l+4}$.
\item $\Cl_{k,l} \otimes \Cl_{8,0} \cong \Cl_{k+8,l}$ and $\Cl_{k,l} \otimes \Cl_{0,8} \cong \Cl_{k,l+8}$.
\end{enumerate}
\end{lemma}
\begin{proof}
\begin{enumerate}
\item[(1)] Let $e_1, \ldots, e_k$ and $\epsilon_1, \ldots, \epsilon_{l+1}$ be generators of the Clifford algebra $\Cl_{k,l+1}$.
Let $\tilde{e}_i = \epsilon_{i+1} \epsilon_{1}$ $(i = 1, \ldots, l)$, $\tilde{\epsilon}_{1} = \epsilon_{1}$ and $\tilde{\epsilon}_{i} = e_{i-1} \epsilon_{1}$ $(i = 2, \ldots, k+1)$.
Then, $\tilde{e}_1, \ldots, \tilde{e}_l$ and $\tilde{\epsilon}_{1}, \ldots, \tilde{\epsilon}_{k+1}$ correspond to generators of the Clifford algebra $\Cl_{l, k+1}$.

\item[(2)] Let $e_1, \ldots, e_k$ and $\epsilon_1, \ldots, \epsilon_{l}$ be generators of the Clifford algebra $\Cl_{k,l}$, and let $e_1$ and $\epsilon_1$ be those of $\Cl_{1,1}$.
We write $\omega_{1,1}$ for $e'_1 \epsilon'_1 \in \Cl_{1,1}$.
Then, $\tilde{e}_{i} = e_i \otimes \omega_{1,1}$ $(i = 1, \ldots, k)$, $\tilde{e}_{k+1} = 1 \otimes e'_1$, $\tilde{\epsilon}_{i} =  \epsilon_i \otimes \omega_{1,1}$ $(i = 1, \ldots, l)$ and $\tilde{\epsilon}_{l} = 1 \otimes \epsilon'_1$ correspond to generators of the Clifford algebra $\Cl_{k+1,l+1}$.

\item[(3)] We show that $\Cl_{k,l} \otimes \Cl_{0,4} \cong \Cl_{k,l+4}$; the other is proved similarly.
Let $e_1, \ldots, e_k$ and $\epsilon_1, \ldots, \epsilon_{l}$ be generators of the Clifford algebra $\Cl_{k,l}$, and let
$\epsilon'_1$, $\epsilon'_2$, $\epsilon'_3$, and $\epsilon'_4$ be those of $\Cl_{0,4}$.
We write $\omega_{0,4}$ for $-\epsilon'_1 \epsilon'_2 \epsilon'_3 \epsilon'_4 \in \Cl_{0,4}$.
Then, $\tilde{e}_{i} = e_i \otimes \omega_{0,4}$ $(i = 1, \ldots, k)$,  $\tilde{\epsilon}_i = \epsilon_i \otimes \omega_{0,4}$ $(i = 1, \ldots, l)$ and $\tilde{\epsilon}_{i} = 1 \otimes \epsilon'_i$ $(i = 1, \ldots, 4)$ correspond to generators of the algebra $\Cl_{k,l+4}$.

\item[(4)] We show that $\Cl_{k,l} \otimes \Cl_{0,8} \cong \Cl_{k,l+8}$; the other is proved similarly.
Let $e_1, \ldots, e_k$ and $\epsilon_1, \ldots, \epsilon_{l}$ be generators of $\Cl_{k,l}$, and let
$\epsilon'_1, \ldots, \epsilon'_8$ be those of $\Cl_{0,8}$.
We write $\omega_{0,8}$ for $\epsilon'_1 \cdots \epsilon'_8 \in \Cl_{0,8}$.
Then, $\tilde{e}_{i} = e_i \otimes \omega_{0,8}$ $(i = 1, \ldots, k)$, $\tilde{\epsilon}_i = \epsilon_i \otimes \omega_{0,8}$ $(i = 1, \ldots, l)$ and $\tilde{\epsilon}_{i} = 1 \otimes \epsilon'_i$ $(i = 1, \ldots, 8)$ correspond to generators of the algebra $\Cl_{k,l+8}$.\qedhere
\end{enumerate}
\end{proof}

Let $W$ be a (ungraded) complex left $\Cl_{k,l}$-module.
We say that $W$ is a {\em (ungraded) real (resp. quaternionic) $\Cl_{k,l}$-module}\footnote{Note that the ``real $\Z_2$-graded $\mathrm{Cliff}(\R^{k,l})$-module'' introduced in \cite{At66} is the same as the (ungraded) real $\Cl_{l,k+1}$-module discussed in this paper.}
if $W$ is equipped with an antilinear map $r \colon W \to W$ (resp. $q \colon W \to W$), which commutes with the $\Cl_{k,l}$-action and satisfies $r^2 = 1$ (resp. $q^2=-1$).
We call this $r$ (resp. $q$) the {\em real (resp. quaternionic) structure} on the Clifford module.
Since a real (resp. quaternionic) $\Cl_{k,l}$-module is the same thing as a module of $\Cl_{k,l} \otimes \Cl_{1,1} \cong \Cl_{k+1,l+1}$ (resp. $\Cl_{k,l} \otimes \Cl_{2,0} \cong \Cl_{l+2,k}$) over $\R$,
the algebra $\Cl_{k,l}$ has one inequivalent irreducible real or quaternionic module when $k - l \not\equiv 3 \bmod 4$ and has two when $k - l \equiv 3 \bmod 4$.
\begin{lemma}\label{Cliff3}\
\begin{enumerate}
\renewcommand{\labelenumi}{(\arabic{enumi})}
\item Let $\Delta_{1,1}$ be a complex irreducible representation of $\Cl_{1,1}$. There exists a real structure $r_{1,1}$ on $\Delta_{1,1}$ that commutes with the Clifford action.
\item Let $\Delta_{0,4}$ (resp. $\Delta_{4,0}$) be a complex irreducible representation of $\Cl_{0,4}$ (resp. $\Cl_{4,0}$). There exists a quaternionic structure $q_{0,4}$ (resp. $q_{4,0}$) on $\Delta_{0,4}$ (resp. $\Delta_{4,0}$) that commutes with the Clifford action.
\item Let $\Delta_{0,8}$ (resp. $\Delta_{8,0}$) be a complex irreducible representation of $\Cl_{0,8}$ (resp. $\Cl_{8,0}$). There exists a real structure $r_{0,8}$ (resp. $r_{8,0}$) on $\Delta_{0,8}$ (resp. $\Delta_{8,0}$) that commutes with the Clifford action.
\end{enumerate}
\end{lemma}
For the proof of this lemma, see \cite{Fr00}, for example.
For a $\Z_2$-space $(X, \zeta)$ with two $\Z_2$-fixed points $x_0, x_1 \in X^\zeta$, we write $P(X; x_0, x_1)$ for the path space starting from $x_0$ and ending at $x_1$, that is, the space of continuous maps $f \colon [0,1] \to X$ satisfying $f(0) = x_0$ and $f(1) = x_1$ equipped with the compact-open topology.
On this space, we consider an involution, for which we also write $\zeta$ by abuse of notation, defined as $(\zeta(f))(t) = \zeta(f(t))$ for $t$ in $[0,1]$, and obtain a $\Z_2$-space $(P(X; x_0, x_1), \zeta)$.
When $x_0=x_1$, we write $\Omega_{x_0}X$ for $P(X; x_0, x_0)$, which is the based loop space of $X$ with the base point $x_0$.
\begin{remark}\label{CW}
Banach $\Z_2$-spaces and its open $\Z_2$-subspaces are $\Z_2$-absolute neighborhood retracts \cite{Anto87}, and have the homotopy type of $\Z_2$-CW complexes \cite{Kwa80}.
The path spaces and loop spaces we discuss in the following also have the homotopy type of $\Z_2$-CW complexes \cite{Wan80a}.
By the equivariant Whitehead theorem, weak $\Z_2$-homotopy equivalences between these spaces are $\Z_2$-homotopy equivalences \cite{Mat71a,AE09}.
\end{remark}

Let $\mathcal{V}$ be a separable infinite-dimensional complex Hilbert space.
Let $\B(\mathcal{V})$ be the space of bounded complex linear operators on $\mathcal{V}$ equipped with the norm topology.
Let $GL(\mathcal{V})$, $U(\mathcal{V})$, $\Fred(\mathcal{V})$ and $\K(\mathcal{V})$ be subspaces of $\B(\mathcal{V})$ consisting of invertible, unitary, Fredholm and compact operators on $\mathcal{V}$, respectively.
We assume that our Hilbert space $\mathcal{V}$ has a real structure $r$ or a quaternionic structure $q$, that is, an antiunitary operator on $\mathcal{V}$ satisfying $r^2=1$ or $q^2 = -1$, respectively.
Correspondingly, the space $\B(\mathcal{V})$ has an (antilinear) involution $\fr = \Ad_r$ or $\fq = \Ad_q$.
These involutions induce involutions on $GL(\mathcal{V})$, $U(\mathcal{V})$, $\Fred(\mathcal{V})$ and $\K(\mathcal{V})$, for which we also write $\fr$ or $\fq$.
We write $a$ for $r$ or $q$ and $\fa$ for $\fr$ or $\fq$.
We also assume that there is a complex linear action of the Clifford algebra $\Cl_{k,l}$ on the Hilbert space $\mathcal{V}$ that commutes with the real or the quaternionic structure.
For an element $v \in \Cl_{k,l}$, we also write $v$ for its action on $\mathcal{V}$, for simplicity.
When $k-l \equiv 3 \bmod 4$, we further assume that each of the two inequivalent irreducible real or quaternionic representations of $Cl_{k,l}$ has infinite multiplicity.
In the following, we discuss the subspaces of $\B(\mathcal{V})$; we may abbreviate the Hilbert space $\mathcal{V}$ from its notation when it is clear from the context.
When the Hilbert space $\mathcal{V}$ is such a $\Cl_{k,l}$-module,
let $\B^{(\underline{k+1},l)}_\ska$ (resp. $\B^{(k,\underline{l+1})}_\sa$) be the subspace of $\B(\mathcal{V})$ consisting of skew-adjoint (resp. self-adjoint) operators $A$ on $\mathcal{V}$ satisfying $e_i A = -A e_i$ for $i = 1, \ldots, k$ and $\epsilon_j A = - A \epsilon_j$ for $j=1, \ldots, l$.
Let $\Fred^{(\underline{k},l)}_\ska = \Fred \cap \B^{(\underline{k},l)}_\ska$ and $\Fred^{(k,\underline{l})}_\sa = \Fred \cap \B^{(k,\underline{l})}_\sa$.
The involution $\fa$ on $\B(\mathcal{V})$ induces involutions on $\Fred^{(\underline{k},l)}_\ska$ and $\Fred^{(k,\underline{l})}_\sa$ for which we also write $\fa$.
Consider the space $\Fred^{(\underline{k},l)}_\ska$, and let $\Upsilon = e_1 \cdots e_{k-1} \epsilon_1 \cdots \epsilon_l$.
When $k-l$ is odd, the space $\Fred^{(\underline{k},l)}_\ska$ is decomposed into three components $\Fred^{(\underline{k},l)}_+$, $\Fred^{(\underline{k},l)}_-$ and $\Fred^{(\underline{k},l)}_*$ corresponding to whether the following element is essentially positive, essentially negative or neither:
$i^{-1}\Upsilon A$ when $k-l \equiv 1 \bmod 4$ and $\Upsilon A$ when $k-l \equiv 3 \bmod 4$ for $A \in \Fred^{(\underline{k},l)}_\ska$.
As in \cite{AS69}, each of these three components is nonempty.
When $k-l \equiv 1 \bmod 4$, the involution $\fa$ maps $\Fred^{(\underline{k},l)}_\pm$ to $\Fred^{(\underline{k},l)}_\mp$ (double-sign corresponds), and $\Fred^{(\underline{k},l)}_*$ is closed under the action of $\fa$.
When $k-l \equiv 3 \bmod 4$, each of the three components is closed under the action of $\fa$.
The space $\Fred^{(k,\underline{l})}_\sa$ is also decomposed into three components in the same way, except that we take $e_1 \cdots e_k \epsilon_1 \cdots \epsilon_{l-1}$ for $\Upsilon$ in this case, and we define the space $\Fred^{(k,\underline{l})}_*$ when $k-l$ is odd.
When $k-l$ is even, we set $\Fred^{(\underline{k},l)}_* = \Fred^{(\underline{k},l)}_\ska$ and $\Fred^{(k,\underline{l})}_* = \Fred^{(k,\underline{l})}_\sa$.
Summarizing, we have the following $\Z_2$-spaces:
\begin{equation}\label{modelFred}
	(\Fred^{(\underline{k},l)}_*, \fr), \
	(\Fred^{(k,\underline{l})}_*, \fr), \
	(\Fred^{(\underline{k},l)}_*, \fq), \
	(\Fred^{(k,\underline{l})}_*, \fq).
\end{equation}
\begin{proposition}\label{htpyhomeo}
The following $\Z_2$-homeomorphisms exist.
\begin{enumerate}
\renewcommand{\labelenumi}{(\arabic{enumi})}
\item $(\Fred^{(\underline{k},{l})}_*, \fa) \cong (\Fred^{(\underline{k+1},{l+1})}_*, \fa)$ and
		$(\Fred^{(k,\underline{l})}_*, \fa) \cong (\Fred^{(k+1,\underline{l+1})}_*, \fa)$,
\item $(\Fred^{(\underline{k},{l})}_*, \fa) \cong (\Fred^{(\underline{k+4},{l})}_*, \widetilde{\fa})$ and
		$(\Fred^{(k,\underline{l})}_*, \fa) \cong (\Fred^{(k+4,\underline{l})}_*,  \widetilde{\fa})$,
\item $(\Fred^{(\underline{k},{l})}_*, \fa) \cong (\Fred^{(\underline{k},{l+4})}_*,  \widetilde{\fa})$ and
		$(\Fred^{(k,\underline{l})}_*, \fa) \cong (\Fred^{(k,\underline{l+4})}_*,  \widetilde{\fa})$,
\item $(\Fred^{(\underline{k},{l})}_*, \fa) \cong (\Fred^{(\underline{k+8},{l})}_*, \fa)$ and
		$(\Fred^{(k,\underline{l})}_*, \fa) \cong (\Fred^{(k+8,\underline{l})}_*, \fa)$,
\item $(\Fred^{(\underline{k},{l})}_*, \fa) \cong (\Fred^{(\underline{k},{l+8})}_*, \fa)$ and
		$(\Fred^{(k,\underline{l})}_*, \fa) \cong (\Fred^{(k,\underline{l+8})}_*, \fa)$,
\item $(\Fred^{(\underline{k+1},{l+1})}_*, \fa) \cong (\Fred^{(l,\underline{k+2})}_*, \fa)$,
\end{enumerate}
where $\widetilde{\fa} = \fq$ when $\fa = \fr$ and $ \widetilde{\fa} = \fr$ when $\fa = \fq$.
\end{proposition}
\begin{proof}
Once the Clifford module structure on the left-hand side of these homeomorphisms is fixed, that on the right-hand side is given following the isomorphisms of Clifford algebras in Lemma~\ref{Cliff1}.
By using Lemma~\ref{Cliff3}, the $\Z_2$-homeomorphisms are given as follows.
\begin{enumerate}
\item[(1)] The map $(\Fred^{(k,\underline{l})}_*(\mathcal{V}), \Ad_a) \to (\Fred^{(k+1,\underline{l+1})}_*(\mathcal{V} \otimes \Delta_{1,1} ), \Ad_{a \otimes r_{1,1} })$ given by $A \mapsto A \otimes \omega_{1,1}$ is a $\Z_2$-homeomorphism
The other one is proved similarly.

\item[(3)] The map $(\Fred^{(k,\underline{l})}_*(\mathcal{V}), \Ad_a) \to (\Fred^{(k,\underline{l+4})}_*(\mathcal{V} \otimes \Delta_{0,4}),  \Ad_{a \otimes q_{0,4}})$ given by $A \mapsto A \otimes \omega_{0,4}$ is a $\Z_2$-homeomorphism. The other one and (2), (4) and (5) follow in a similar way.

\item[(6)] The map $(\Fred^{(\underline{k+1},{l+1})}_*(\mathcal{V}), \fa) \to (\Fred^{(l,\underline{k+2})}_*(\mathcal{V}), \fa)$ given by $A \mapsto A \epsilon_{1}$ is a $\Z_2$-homeomorphism. \qedhere
    \end{enumerate}
\end{proof}

\begin{proposition}\label{htpymain}
The following $\Z_2$-homotopy equivalences exist.
\begin{enumerate}
\renewcommand{\labelenumi}{(\arabic{enumi})}
\item $(\Fred^{(\underline{k+1},l)}_*, \fa) \simeq (\Omega_{e_k}\Fred^{(\underline{k},l)}_*, \fa)$, for $k \geq 1$ and $l \geq 0$.
\item $(\Fred^{(k,\underline{l+1})}_*, \fa) \simeq (\Omega_{\epsilon_l}\Fred^{(k,\underline{l})}_*, \fa)$, for $k \geq 0$ and $l \geq 1$.
\item $(\Fred^{(\underline{1},0)}_*, \fa) \simeq (\Omega_{1}\Fred, \fa)$.
\end{enumerate}
\end{proposition}
Proposition~\ref{htpymain} is proved as in \cite{AS69}.
In what follows, we outline its proof since some spaces introduced there are of our interest.

\begin{proposition}\label{htpymain2}
The following maps are $\Z_2$-homotopy equivalences.
\begin{enumerate}
\renewcommand{\labelenumi}{(\arabic{enumi})}
\item $\alpha_1 \colon (\Fred^{(\underline{k+1},l)}_*, \fa) \to (P(\Fred^{(\underline{k},l)}_*; e_k, -e_k), \fa)$, where\\
	$\alpha_1(A)(t) = e_k \cos(\pi t) + A \sin(\pi t)$ for $0 \leq t \leq 1$.
\item $\alpha_2 \colon (\Fred^{(k,\underline{l+1})}_*, \fa) \to (P(\Fred^{(k,\underline{l})}_*; \epsilon_l, -\epsilon_l), \fa)$, where\\
	$\alpha_2(A)(t) = \epsilon_l \cos(\pi t) - A \sin(\pi t)$ for $0 \leq t \leq 1$.
\item $\alpha_3 \colon (\Fred^{(\underline{1},0)}_*, \fa) \to (P(\Fred; 1, -1), \fa)$, where\\
	$\alpha_3(A)(t) = \cos(\pi t) + A \sin(\pi t)$ for $0 \leq t \leq 1$.
\end{enumerate}
\end{proposition}
Proposition~\ref{htpymain} follows from Proposition~\ref{htpymain2} since, in each case, there is a path connecting the endpoints of each path space in the unitaries preserving the Clifford action and the $\Z_2$-action.
As in \cite{AS69}, the proof of Proposition~\ref{htpymain2} reduces to showing the $\Z_2$-homotopy equivalences between some spaces of Fredholm operators and some spaces of unitary operators (Proposition~\ref{htpyequiv}).
Let $F^{(\underline{k},l)}_*$ (resp. $F^{(k,\underline{l})}_*$) be the subspace of $\Fred^{(\underline{k},l)}_*$ (resp. $\Fred^{(k,\underline{l})}_*$ ) consisting of those operators whose essential spectra are $\{ i, -i \}$ (resp. $\{ 1, -1 \}$) and whose operator norms are $1$.
The spaces $F^{(\underline{k},l)}_*$ and $F^{(k,\underline{l})}_*$ are closed under the action of $\fa$; thus, we have $\Z_2$-spaces $(F^{(\underline{k},l)}_*, \fa)$ and $(F^{(k,\underline{l})}_*, \fa)$.
Inclusions $(F^{(\underline{k},l)}_*, \fa) \hookrightarrow (\Fred^{(\underline{k},l)}_*, \fa)$ and $(F^{(k,\underline{l})}_*, \fa) \hookrightarrow (\Fred^{(k,\underline{l})}_*, \fa)$ are $\Z_2$-homotopy equivalences.
Let $U_\cpt$ be the subspace of $U(\mathcal{V})$ consisting of unitary operators of the form $1 + T$, where $T \in \K(\mathcal{V})$.
When the Hilbert space $\mathcal{V}$ is a $\Cl_{k,l}$-module, let $U^{(\underline{k},l)}_\cpt$ (resp. $U^{(k,\underline{l})}_\cpt$) be the subspace of $U(\mathcal{V}) \cap \B^{(\underline{k},l)}_\ska$ (resp. $U(\mathcal{V}) \cap \B^{(k,\underline{l})}_\sa$) consisting of a unitary $u$ satisfying $u^2=-1$ (resp. $u^2=1$) and $u \equiv e_k$ (resp. $u \equiv \epsilon_l$) modulo compact operators.
If the Hilbert space has a real or quaternionic structure, these spaces of unitaries are closed under the action of $\fa$, and we obtain $\Z_2$-spaces.

\begin{proposition}\label{htpyequiv}
The following maps are $\Z_2$-homotopy equivalences:
\begin{enumerate}
\renewcommand{\labelenumi}{(\arabic{enumi})}
\item $p_1 \colon (F^{(\underline{k+1},l)}_*, \fa) \to (-U^{(\underline{k},l)}_\cpt, \fa)$,
	$p_1(A) = e_{k} \exp(\pi A e_{k})$, for $k \geq 1$, $l \geq 0$.
\item $p_2 \colon (F^{(k,\underline{l+1})}_*, \fa) \to (-U^{(k,\underline{l})}_\cpt, \fa)$, $p_2(A) = \epsilon_{l} \exp(\pi A \epsilon_{l})$, for $k \geq 0$, $l \geq 1$.
\item $p_3 \colon (F^{(\underline{1},0)}_*, \fa) \to (-U_\cpt, \fa)$, $p_3(A) = \exp(\pi A)$.
\item $p_4 \colon (F^{(0,\underline{1})}_*, \fa) \to (-U_\cpt, \fa \circ*)$, $p_4(A) = \exp(\pi i A)$.
\end{enumerate}
\end{proposition}
\noindent
\begin{proof}
By Remark~\ref{CW}, it is sufficient to show that these maps are weak $\Z_2$-homotopy equivalences.
Equivalently, to show that $p_i$ and its restriction to the $\Z_2$-fixed point sets (the map $p_1^{\Z_2} \colon (F^{(\underline{k+1},l)}_*)^\fa \to (-U^{(\underline{k},l)}_\cpt)^\fa$ in the case of (1)) are weak homotopy equivalences.
They are proved by using quasifibrations on some dense subspaces of contractible fibers as in \cite{AS69}.
\end{proof}

\begin{lemma}\label{lemma1}
There is a $\Z_2$-homeomorphism
$(\Fred, \fa) \cong (\Fred_*^{(0,\underline{2})}, \fa)$.
\end{lemma}
\begin{proof}
This is given by a $\Z_2$-map $(\Fred(\mathcal{V}), \fa) \to (\Fred_*^{(0,\underline{2})}(\mathcal{V} \oplus \mathcal{V}), \fa \oplus \fa)$, $ A \mapsto
\left(\hspace{-1mm}
    \begin{array}{cc}
           0 & \hspace{-1mm} A^*\\
           A & 0
    \end{array}\hspace{-1mm}
\right)$,
where the action of $\Cl_{0,1}$ on $\mathcal{V} \oplus \mathcal{V}$ is given by $\epsilon_1 = \diag(1,-1)$.
\end{proof}
Proposition~\ref{htpyhomeo}, Proposition~\ref{htpymain} and Lemma~\ref{lemma1} lead to the following.
\begin{corollary}\label{htpyomega}
The following $\Z_2$-homotopy equivalences exist.
\begin{enumerate}
\renewcommand{\labelenumi}{(\arabic{enumi})}
\item $(\Fred^{(\underline{k},l)}_*, \fr) \simeq (\Omega^{k-l} \Fred, \fr)$.
\item $(\Fred^{(\underline{k},l)}_*, \fq) \simeq (\Omega^{k-l+4} \Fred, \fr)$.
\item $(\Fred^{(k,\underline{l})}_*, \fr) \simeq (\Omega^{l-k+6} \Fred, \fr)$.
\item $(\Fred^{(k,\underline{l})}_*, \fq) \simeq (\Omega^{l-k+2} \Fred, \fr)$.
\end{enumerate}
When the subscript $m$ on $\Omega^m$ is negative, this should be replaced by $m + 8n$ by taking a sufficiently large integer $n$ to make the subscript non-negative.
\end{corollary}
Note that when $k$ and $l$ are relatively small, we further have the following $\Z_2$-homeomorphisms.
\begin{lemma}\label{htpyspecial}
Multiplication by the imaginary unit $i = \sqrt{-1}$ induces the following $\Z_2$-homeomorphisms:
\begin{enumerate}
\renewcommand{\labelenumi}{(\arabic{enumi})}
\item $(\Fred^{(0,\underline{2})}_*, \Ad_r) \to (\Fred^{(\underline{1},1)}_*, \Ad_{\widetilde{r}})$, where $\tilde{r}= r \epsilon_1$.
\item $(\Fred^{(0,\underline{2})}_*, \Ad_q) \to (\Fred^{(\underline{1},1)}_*, \Ad_{\widetilde{q}})$, where $\tilde{q}= -q \epsilon_1$.
\item $(\Fred^{(1,\underline{1})}_*, \Ad_q) \to (\Fred^{(\underline{2},0)}_*, \Ad_{\widetilde{r}})$, where $\tilde{r} =qe_1$.
\item $(\Fred^{(1,\underline{1})}_*, \Ad_r) \to (\Fred^{(\underline{2},0)}_*, \Ad_{\widetilde{q}})$, where $\tilde{q}= -re_1$.
\end{enumerate}
\end{lemma}
\begin{remark}\label{RemarkA12}
The $\Z_2$-spaces in Lemma~\ref{htpyspecial} appear in the study of topological insulators.
Specifically, Table~\ref{Table1} is obtained by taking the quantum symmetries as real or quaternionic Clifford module structures as follows.
\begin{itemize}
\item In class $\BDI$, we put $r = \Theta$ and $\epsilon_1 = \Pi$ in (1); then, $\tilde{r} = \Xi$.
\item In class $\CII$, we put $q = \Theta$ and $\epsilon_1 = \Pi$ in (2); then, $\tilde{q} = \Xi$.
\item In class $\DIII$, we put $q = \Theta$ and $e_1 = i\Pi$ in (3); then, $\tilde{r} = \Xi$.
\item In class $\CI$, we put $r = \Theta$ and $e_1 = i\Pi$ in (4); then, $\tilde{q} = \Xi$.
\end{itemize}
\end{remark}

For $l \geq 1$, let us consider the map
\begin{equation}\label{FredProdsa}
	(\Fred^{(k, \underline{l+1})}_*, \fa) \times (\Fred^{(k', \underline{l'+1})}_*, \fa')  \to (\Fred^{(k+k', \underline{l+l'})}_*, \fa \otimes \fa')
\end{equation}
defined by $(A, B) \mapsto A \otimes 1 + \epsilon_l \otimes B$, where the Clifford action to define $\Fred^{(k+k',\underline{l+l'})}_*$ is generated by $\tilde{e}_i = e_i \otimes 1$ $(i=1, \ldots, k)$, $\tilde{e}_{k+i} = \epsilon_l \otimes e_i$ $(i = 1, \ldots, k')$, $\tilde{\epsilon}_i = \epsilon_i \otimes 1$ $(i= 1, \ldots, l-1)$ and $\tilde{\epsilon}_{l+i-1} = \epsilon_l \otimes \epsilon_i$ $(i=1, \ldots, l')$.
This map induces the exterior product of topological $KR$-groups as in \cite{AS69}.

\subsection{Relation with Boersema--Loring's Unitary Picture}
\label{A2}
In this subsection, we discuss a relation between these $\Z_2$-spaces of self-adjoint/skew-adjoint Fredholm operators and Boersema--Loring's $K$-theory.

Let $\{W_i\}_{i \in I}$ be the set of mutually inequivalent irreducible real (resp. quaternionic) representations of $\Cl_{k,l}$ with hermitian inner-products which $\{W_i\}_{i \in I}$ consists of one or two elements corresponding to $k$ and $l$.
Let $W = \oplus_{i \in I} W_i$ and $\mathcal{V} = l^2(\Z_{\geq 0}) \otimes W$ which has a real (resp. quaternionic) $\Cl_{k,l}$-module structure induced by that of $\{W_i\}_{i \in I}$.
We take a complete orthonormal basis $\{\delta_j\}_{j \in \Z_{\geq 0}}$ of $l^2(\Z_{\geq 0})$ given by generating functions of each points in $\Z_{\geq 0}$.
Let $\mathcal{V}_n$ be the subspace of $\mathcal{V}$ spanned by $\{ \delta_j \otimes w \ | \ 0 \leq j \leq n, w \in W \}$, which is a real (resp. quaternionic) $\Cl_{k,l}$-module.
Let $GL_\cpt$ be the space of invertible operators on $\mathcal{V}$ of the form $e_k + T$ for some compact operator $T$.
Let $GL^{(\underline{k},l)}_\cpt = GL_\cpt \cap \B_\ska^{(\underline{k},l)}(\mathcal{V})$ and $GL^{(k,\underline{l})}_\cpt = GL_\cpt \cap \B_\sa^{(k,\underline{l})}(\mathcal{V})$.
Let $GL^{(\underline{k},l)}_n$ (resp. $GL^{(k,\underline{l})}_n$) be the subspace of $\B_\ska^{(\underline{k},l)}(\mathcal{V}_n)$ (resp. $\B_\sa^{(k,\underline{l})}(\mathcal{V}_n)$) consisting of invertible operators, and let $U^{(\underline{k},l)}_n$ (resp. $U^{(k,\underline{l})}_n$) be its subspace of unitaries.
We have an injection $GL^{(\underline{k},l)}_n \hookrightarrow GL^{(\underline{k},l)}_{n+1}$ (resp. $GL^{(k,\underline{l})}_n \hookrightarrow GL^{(k,\underline{l})}_{n+1}$) given by mapping $A$ to $A \oplus e_k$ (resp. $A \oplus \epsilon_l$), and let $GL^{(\underline{k},l)}_\infty$ (resp. $GL^{(k,\underline{l})}_\infty$) be its inductive limit  $\colim GL^{(\underline{k},l)}_n$ (resp.  $\colim GL^{(k,\underline{l})}_n$).
We also define $U^{(\underline{k},l)}_\infty$ and $U^{(k,\underline{l})}_\infty$ for unitaries in the same way.
The space $GL^{(\underline{k},l)}_n$ (resp. $GL^{(k,\underline{l})}_n$) is identified with the subspace of $GL^{(\underline{k},l)}_\cpt$ (resp. $GL^{(k,\underline{l})}_\cpt$) consisting of operators of the form $e_k + T$ (resp. $\epsilon_l + T$), where $T \in \B(\mathcal{V}_n)$, 
and we have an injective $\Z_2$-map $(GL^{(\underline{k},l)}_\infty, \fa) \to (GL^{(\underline{k},l)}_\cpt, \fa)$ (resp. $(GL^{(k,\underline{l})}_\infty, \fa) \to (GL^{(k,\underline{l})}_\cpt, \fa)$).
As in \cite{Pal65}, the following holds\footnote{In \cite{Pal65}, an upper semicontinuous function is introduced to show that an injection $GL_\infty \to GL_\cpt$ is a homotopy equivalence. In our setup, this function is $\Z_2$-invariant, and the result follows as in \cite{Pal65}.}.
\begin{proposition}
The map $(GL^{(\underline{k},l)}_\infty, \fa) \to (GL^{(\underline{k},l)}_\cpt, \fa)$ and the map $(GL^{(k,\underline{l})}_\infty, \fa) \to (GL^{(k,\underline{l})}_\cpt, \fa)$ are $\Z_2$-homotopy equivalences.
\end{proposition}
By using a deformation of invertibles to unitaries, $(U^{(\underline{k},l)}_\infty \hspace{-0.7mm}, \fa)$ and $(U^{(k,\underline{l})}_\infty \hspace{-0.7mm}, \fa)$ are $\Z_2$-homotopy equivalent to $(U^{(\underline{k},l)}_\cpt, \fa)$ and $(U^{(k,\underline{l})}_\cpt, \fa)$, respectively.
We denote $U^{\spadesuit}_\infty$ for these subspaces of $U^{\spadesuit}$ as indicated in Table~\ref{Table1}.

These $\Z_2$-spaces of unitaries appears in Boersema--Loring's $KO$-theory \cite{BL16}.
Let $(X, \zeta)$ be a compact Hausdorff $\Z_2$-space, and consider a \TA \ $(C(X), \tau_\zeta)$ of continuous functions on $X$, whose transposition $\tau_\zeta$ is given by $f^{\tau_\zeta}(x) = f(\zeta(x))$.
Then, the $\Z_2$-homotopy classes $[(X, \zeta), U^{\spadesuit}_\infty]_{\Z_2}$ can be identified with the group $KO_{i(\spadesuit)-1}(C(X), \tau_\zeta)$ where $i(\spadesuit)$ is as indicated in Table~\ref{label}.
In the following, we discuss two of eight $KO$-groups and the others are discussed in a similar way.

As for the $KO_{-1}$-group, an element of the set $[(X, \zeta), (U_\infty, \fr \circ *)]_{\Z_2}$ is represented by a $\Z_2$-map $f \colon (X, \zeta) \to (U_n, \fr \circ *)$.
This $f$ is a unitary element of $M_n(C(X))$ satisfying $f(\zeta(x)) = \fr(f(x))^*$ which is the same as the relation $f^{\tau_\zeta} = f$ to define $KO_{-1}$-groups.
Thus, the set $[(X, \zeta), (U_\infty, \fr \circ *)]_{\Z_2}$ is the same as $KO_{-1}(C(X), \tau_\zeta)$ by the definition of Boersema--Loring's $KO_{-1}$-group.

Finally, we discuss the $KO_6$-group.
By the multiplication of $-i$, we have a $\Z_2$-homeomorphism $(U^{(\underline{1},0)}_\infty, \fq) \to (U^{(0,\underline{1})}_\infty, -\fq)$.
A $\Z_2$-continuous map $f \colon (X, \zeta) \to (U^{(0,\underline{1})}_n, -\fq)$ is a self-adjoint unitary in $M_n(C(X))$ satisfying $f^{\sharp \otimes \tau_\zeta} = -f^* = -f$.
The Clifford algebra $\Cl_{1,0}$ has just one irreducible quaternionic representation up to equivalence, which is constructed as follows.
On $W = \C^2$, we consider the action $\rho$ of $\Cl_{1,0} \otimes \Cl_{2,0}$ defined as follows:
\begin{equation*}
\rho(1 \otimes e_1) =
\left(
    \begin{array}{cc}
           i&0\\
           0&i
    \end{array}
\right), \ \
\rho(1 \otimes e_2) =
\left(
    \begin{array}{cc}
           0&-c\\
           c&0
    \end{array}
\right), \ \
\rho(e_1 \otimes 1) =
\left(
    \begin{array}{cc}
           0&- 1\\
           1&0
    \end{array}
\right),
\end{equation*}
where $c$ is the complex conjugation on $\C$.
The space $U^{(\underline{1},0)}_\infty$ is defined as the inductive limit of maps $U^{(\underline{1},0)}_n \to U^{(\underline{1},0)}_{n+1}$, $A \mapsto A \oplus I$ where $I = \rho(e_1 \otimes 1)$ and the space $U^{(0,\underline{1})}_\infty$ is defined as that of maps $A \mapsto A \oplus -iI$ where $-iI = I^{(6)}$.

\subsection{Boersema--Loring's $K$-Theory and Exponential Maps}\label{Sect.B}
We describe boundary maps of the $24$-term exact sequence of $KO$-theory (which we denote as $\partial^{\BL}_i$ in this section) in Boersema--Loring's unitary picture through exponential maps.
The map $\partial^{\BL}_i$ for even $i$ has already been expressed as an exponential map in \cite{BL16}; thus, we focus on $\partial^{\BL}_i$ for odd $i$.
A clue is the exponential maps given in Proposition~\ref{htpyequiv}.
For a short exact sequence of \TA s,
\begin{equation}\label{exgen}
	 0 \to (\mathcal{I}, \tau) \to (\mathcal{A}, \tau) \overset{\varphi}{\to} (\mathcal{B},\tau) \to 0,
\end{equation}
and for each odd $i$, we construct a group homomorphism
\begin{equation}
	\partial^{\exp}_i \colon KO_i(\mathcal{B}, \tau) \to KO_{i-1}(\mathcal{I},\tau)
\end{equation}
and show they coincides with $\partial^{\BL}_i$ up to a factor of $-1$.
Let $W_{2n} \in M_{2n}(\C)$ and $Q_{4n} \in M_{4n}(\C)$ be the following matrices:
\begin{equation*}
W_{2n} =
\frac{1}{\sqrt{2}}
\left(
    \begin{array}{cc}
           i \cdot 1_n & 1_n\\
           1_n & i \cdot 1_n
    \end{array}
\right), \ \
Q_{4n} =
\frac{1}{\sqrt{2}}
\left(
    \begin{array}{cc}
           1_{2n}&-I_{n}^{(2)}\\
           I_{n}^{(2)}&1_{2n}
    \end{array}
\right),
\end{equation*}
and let $V_{2n} \in M_{2n}(\R)$ and $X_{4n} \in M_{4n}(\R)$ be the permutation matrices satisfying\footnote{Matrices $W_{2n}$, $Q_{4n}$, $V_{2n}$ and $X_{4n}$ are what we borrowed from Sect.~$8$ of \cite{BL16}. Some of the basic formulas that they satisfy can be found there.}
\begin{gather*}
V_{2n} \diag(x_1, \ldots, x_{2n}) V_{2n}^*
= \diag(x_1,x_{n+1},x_2,x_{n+2},\ldots,x_n,x_{2n}),\\
X_{4n} \diag(x_1, \ldots, x_{4n}) X_{4n}^* \hspace{6cm} \\
= \diag(x_1,x_2,x_{2n+1},x_{2n+2},x_3,x_4,x_{2n+3},x_{2n+4},\ldots,x_{4n}).
\end{gather*}
As in \cite{BL16}, let $Y_{2n}^{(-1)} \hspace{-1mm}= V_{2n} W_{2n}$,
$Y_{2n}^{(1)} \hspace{-1.5mm}= V_{2n}$,
$Y_{4n}^{(3)} \hspace{-1.5mm}= V_{4n} Q_{4n} W_{4n}$ and
$Y_{4n}^{(5)} \hspace{-1.5mm}= X_{4n}$.

\begin{definition}\label{boundarymap}
Suppose we have a short exact sequence of \TA s as in (\ref{exgen}).
We assume $\mathcal{I} = \Ker(\varphi)$ and identify the unit in $\widetilde{\mathcal{I}}$ with that of $\widetilde{\mathcal{A}}$.
For $i \in \{-1,1,3,5\}$, suppose $[u] \in KO_i(\mathcal{B}, \tau)$, where $u \in M_{n_i \cdot n}(\widetilde{\mathcal{B}})$ is a unitary satisfying the relation $\cR_i$ and $\lambda(u) = I_n^{(i)}$, for which $n_i$, $\cR_i$ and $I^{(i)}$ are as in Table~\ref{BL}.
Let $a$ in $M_{n_i \cdot n}(\widetilde{\mathcal{A}})$ be a lift of $u$ satisfying the relation $\cR_i$ and $\|a\|\leq1$.
Then, define
\begin{equation*}
	\partial^{\exp}_i([u]) =
		\left[ -Y_{2n_i \cdot n}^{(i)}\bigl( \epsilon_1 \exp(\pi A \epsilon_1) \bigl) Y_{2n_i \cdot n}^{(i)*} \right] \in  KO_i(\mathcal{I},\tau),
\end{equation*}
where
$A = \left(
    \begin{array}{cc}
           0&a^*\\
           a&0
    \end{array}
\right)$
and $\epsilon_1 = \diag(1_{n_i \cdot n}, -1_{n_i \cdot n})$.
\end{definition}

\begin{lemma}
$\partial^{\exp}_i$ for odd $i$ are well-defined group homomorphisms.
\end{lemma}
\begin{proof}
We need to show that
(a) the unitaries constructed all satisfy the correct relation,
(b) the choice of lift is not important,
(c) some lift is always available,
(d) homotopy is respected,
(e) compatible with respect to the stabilization by $I^{(i)}$
and
(f) the addition is respected.
(c), (b) and (d) are proved in the same way as in Lemma~$8.2$ of \cite{BL16} and we discuss the other parts.
For convenience, let $C(a) = -\epsilon_1 \exp(\pi A \epsilon_1)$ and $C'(a)= Y_{2n_i \cdot n}^{(i)} C(a) Y_{2n_i \cdot n}^{(i)*}$.
\begin{enumerate}
\item[(1)] We first consider the case of $i=1$.
Let $u \in M_n(\widetilde{\mathcal{B}})$ be a unitary satisfying $u^\tau = u^*$ and $\lambda_n(u) = I_n^{(1)}$.
We take a lift $a \in M_n(\widetilde{\mathcal{A}})$ of $u$ such that $\| a \| \leq 1$ and $a^\tau = a^*$.
Since
$\varphi(C'(a)) = V_{2n} \epsilon_1 V_{2n}^* = I^{(0)}_n$,
we have $C'(a) \in M_n(\widetilde{I})$ and $\lambda(C'(a)) = I^{(0)}_n$.
Since $A^\tau = A$, we have,
\begin{equation*}
	C(a)^\tau = -\exp (\pi A \epsilon_1)^\tau \epsilon_1^\tau
			= - \epsilon_1^2 \exp(\pi \epsilon_1^\tau A^\tau) \epsilon_1^*
			= -\epsilon_1 \exp(\pi \epsilon_1^2 A \epsilon_1^*)
			= C(a).
\end{equation*}
Since $Y_{2n}^{(1)}= V_{2n}$ is the orthogonal matrix, $(Y_{2n}^{(1)})^\tau = Y_{2n}^{(1)*}$, and thus, $C'(a)^\tau = C'(a)$ holds.
When $u = 1$, we can take $a=1$ and $C'(1)= I^{(1)}$ in this case.
Combined with this, the proof is completed once we have checked that $\partial^{\exp}_1$ preserves the addition.
Let $u \in M_m(\widetilde{\mathcal{B}})$ and $v \in M_n(\widetilde{\mathcal{B}})$.
We take their lift $a$ and $b$ such that $a^\tau = a^*$ and $b^\tau = b^*$.
Then, we have $C'(\diag(a,b)) = \diag(C'(a),C'(b))$ since
\small
\begin{equation*}
C\left(\hspace{-0.5mm}\left(
    \begin{array}{cc}
           a&0\\
           0&b
    \end{array}
\right) \hspace{-0.5mm} \right)
	=
- \left(
    \begin{array}{cc}
           \hspace{-0.5mm} 1_{m+n} \hspace{-0.5mm} \hspace{-2mm} & 0 \\
           0 & \hspace{-2mm} -1_{m+n} \hspace{-0.5mm}
    \end{array}
\right) \exp \Biggl( \hspace{-0.5mm} \pi \hspace{-0.5mm}
 \left(\begin{array}{cc}
 0 & \hspace{-1mm} \diag(-a^*, -b^*)\\
 \diag(a, b) \hspace{-2mm}  & 0
\end{array}\right)\hspace{-1mm}
\Biggl),
\end{equation*}
\normalsize
\begin{align*}
&
\Ad_{V_{2m+2n}}
\exp \Biggl( \pi
 \left(\begin{array}{cc}
 0 & \diag(-a^*, -b^*) \\
 \diag(a, b) & 0
\end{array}\right)
\Biggl)\\
&=
	\exp \left(\pi \cdot
\diag \left(
\Ad_{V_{2m}}
 \left(\begin{array}{cc}
         0 & -a^* \\
         a & 0 \\
  \end{array}\right),
\Ad_{V_{2n}} \left(\begin{array}{cc}
                       0 & -b^* \\
                       b & 0 \\
                      \end{array}\right)
\right)
\right),
\end{align*}
which shows that $\partial^{\exp}_1([u] + [v]) = \partial^{\exp}_1([u]) + \partial^{\exp}_1([v])$.

\item[(2)] We next consider the case of $i=-1$.
Let $u \in M_n(\widetilde{\mathcal{B}})$ be a unitary satisfying $u^\tau = u$ and $\lambda_n(u) = I_n^{(-1)}$.
We take a lift $a \in M_n(\widetilde{\mathcal{A}})$ of $u$ such that $\| a \| \leq 1$ and $a^\tau = a$.
Since $\varphi(C'(a)) =I^{(6)}_n$,
we have $C'(a) \in M_n(\widetilde{\mathcal{I}})$ and $\lambda(C'(a)) = I^{(6)}_n$.
Since $A^{\widetilde{\sharp}\otimes \tau} = -A$, we have
\begin{align*}
C(a)^{\widetilde{\sharp}\otimes \tau} &= -\exp(\pi A \epsilon_1)^{\widetilde{\sharp}\otimes \tau} \epsilon_1^{\widetilde{\sharp}\otimes \tau}
	= \exp(\pi \epsilon_1^{\widetilde{\sharp}\otimes \tau} A^{\widetilde{\sharp}\otimes \tau}) \epsilon_1\\
	&= \exp(\pi \epsilon_1 A) \epsilon_1
	= \epsilon_1 \exp(\pi A \epsilon_1)
	= -C(a).
\end{align*}
Since $(V_{2n} x V_{2n}^*)^{\sharp\otimes \tau} = V_{2n} x^{\widetilde{\sharp} \otimes \tau} V_{2n}^*$ and $W_{2n}^{\widetilde{\sharp} \otimes \tau} = - W_{2n}^*$, we have $C'(a)^{\sharp\otimes \tau} = - C'(a)$.
For $u = 1$, we take $a = 1$ and $C'(1) = I^{(6)}$ holds.
Therefore, as in (1), all we have to show is the additivity of $\partial_{-1}^{\exp}$.
Let $a \in M_m(\widetilde{\cA})$ and $b \in M_n(\widetilde{\cA})$ be lifts of the unitaries $u$ and $v$.
Then, $C'(\diag(a,b)) = \diag(C'(a), C'(b))$ follows from
\begin{align*}
& V_{2m+2n} W_{2m+2n}
 \left(\begin{array}{cc}
 0 & \diag(-a^*, -b^*) \\
 \diag(a, b) & 0
\end{array}\right)
W_{2m+2n}^* V_{2m+2n}^*\\
	&=
\diag \left(
V_{2m} W_{2m}
\left(\begin{array}{cc}
         0 & -a^* \\
         a & 0 \\
\end{array}\right) W_{2m}^* V_{2m}^*,
V_{2n} W_{2n} \left(\begin{array}{cc}
                       0 & -b^* \\
                       b & 0 \\
                      \end{array}\right)
W_{2n}^* V_{2n}^*
\right).
\end{align*}

\item[(3)] Let us consider the case of $i=5$.
Let $u \in M_{2n}(\widetilde{\mathcal{B}})$ be a unitary satisfying $u^{\sharp \otimes \tau} = u^*$ and $\lambda_{2n}(u) = I_n^{(5)}=1_{2n}$.
We take a lift $a \in M_{2n}(\widetilde{\mathcal{A}})$ of $u$ such that $\| a \| \leq 1$ and $a^{\sharp \otimes \tau} = a^*$.
Since $(X_{4n} x X_{4n})^{\sharp \otimes \tau} = X_{4n} x^{\sharp \otimes \tau} X_{4n}$, $A^{\sharp \otimes \tau} = A$ and $\epsilon_1^{\sharp \otimes \tau} =\epsilon_1$, the relation
\begin{equation*}
C'(a)^{\sharp \otimes \tau} \hspace{-1mm}= -X_{4n}\exp (\pi \epsilon_1 A^{\sharp \otimes \tau}) \epsilon_1 X_{4n}^*
= -X_{4n}\exp(\pi \epsilon_1 A)\epsilon_1 X_{4n}^* = C'(a)
\end{equation*}
holds.
We have $C'(1_2) = I^{(4)}$, and for the additivity of $\partial^{\exp}_{5}$, note that
\begin{align*}
& X_{4m+4n}
 \left(\begin{array}{cc}
 0 & \diag(-a^*, -b^*) \\
 \diag(a, b) & 0
\end{array}\right)
X_{4m+4n}^*\\
	&=
\diag \left( X_{4m}
\left(\begin{array}{cc}
         0 & -a^* \\
         a & 0 \\
\end{array}\right) X_{4m}^*,
X_{4n} \left(\begin{array}{cc}
                       0 & -b^* \\
                       b & 0 \\
                      \end{array}\right)
				X_{4n}^* \right).
\end{align*}

\item[(4)] Consider the case of $i=3$.
Let $u \in M_{2n}(\widetilde{\mathcal{B}})$ be a unitary satisfying $u^{\sharp \otimes \tau} = u$ and $\lambda_{2n}(u) = I_n^{(3)}=1_{2n}$.
We take a lift $a \in M_{2n}(\widetilde{\mathcal{A}})$ of $u$ such that $\| a \| \leq 1$ and $a^{\sharp \otimes \tau} = a$.
Since $A^{\widetilde{\sharp} \otimes \sharp \otimes \tau} = -A$ and $\epsilon_1^{\widetilde{\sharp} \otimes \sharp \otimes \tau} = -\epsilon_1$, the relation
$C(a)^{\widetilde{\sharp} \otimes \sharp \otimes \tau} = -C(a)$ holds.
Since $(Q_{4n} x Q_{4n}^*)^\tau = Q_{4n} x^{\widetilde{\sharp} \otimes \sharp \otimes \tau} Q_{4n}^*$ and $W_{4n}^{\widetilde{\sharp} \otimes \sharp} = -W_{4n}^*$, we have $C'(a)^\tau = -C'(a)$.
For the remaining part, we note that $C'(1_2) = I^{(2)}_2$ and
\begin{align*}
& Y_{4m+4n}^{(3)}
 \left(\begin{array}{cc}
 0 & \diag(-a^*, -b^*) \\
 \diag(a, b) & 0
\end{array}\right)
Y_{4m+4n}^{(3)*}\\
	&=
\diag \left(
Y_{4m}^{(3)}
 \left(\begin{array}{cc}
         0 & -a^* \\
         a & 0 \\
\end{array}\right) Y_{4m}^{(3)*},
Y_{4n}^{(3)} \left(\begin{array}{cc}
                       0 & -b^* \\
                       b & 0 \\
                      \end{array}\right)
				Y_{4n}^{(3)*}
\right).
\qedhere
\end{align*}
\end{enumerate}
\end{proof}
\begin{lemma}
Each $\partial_i^{\exp}$ is natural with respect to the morphisms of short exact sequences of \CA s.
That is, suppose we have the following commutative diagram of exact lows:
\[\xymatrix{
0 \ar[r] & (\mathcal{I}_1,\tau) \ar[r] \ar[d]_\iota & (\mathcal{A}_1, \tau) \ar[r]^{\varphi_1} \ar[d]_\alpha & (\mathcal{B}_1, \tau) \ar[r] \ar[d]_\beta & 0\\
0 \ar[r] & (\mathcal{I}_2,\tau) \ar[r] & (\mathcal{A}_2, \tau) \ar[r]^{\varphi_2} & (\mathcal{B}_2, \tau) \ar[r] & 0
}\]
Then, we have $\iota_* \circ \partial_i^{\exp} = \partial_i^{\exp} \circ \beta_*$.
\end{lemma}
\begin{proof}
As in Lemma~$8.5$ of \cite{BL16}, this lemma is proved by following the definition of each map.
We assume $\Ker(\varphi_j) = \mathcal{I}_j$ $(j=1,2)$ for simplicity.
Let $[u_1] \in KO_i^u(\mathcal{B}_1,\tau)$ be an element represented by a unitary $u_1 \in M_{n_i \cdot n}(\widetilde{\mathcal{B}}_1)$ satisfying the symmetry relation $\cR_i$.
Let $a_1 \in M_{n_i \cdot n}(\widetilde{\mathcal{A}})$ be a lift of $u_1$ such that $||a_1||\leq 1$, and satisfy the relation $\cR_i$.
Then, $a_2 = \alpha(a_1)$ is a lift of $u_2$ satisfying the symmetry, and thus, $\partial_i^{\exp}([u_2]) = [C'(a_2)]$ holds. Since $\alpha(C'(a_1)) = C'(\alpha(a_1)) = C'(\alpha_2)$, we have
\begin{equation*}
	\iota_* \circ \partial_i^{\exp}([u_1]) = \iota_*[C'(a_1)]
		= [\alpha(C'(a_1)]\\
		= \partial_i^{\exp}([u_2])
		= \partial_i^{\exp} \circ \beta_* [u_1]. \qedhere
\end{equation*}
\end{proof}
\begin{proposition}\label{brelation}
	$\partial_i^{\BL} = - \partial_i^{\exp}$ for odd $i$.
\end{proposition}

As in the proof of Theorem~$8.9$ of \cite{BL16}, we can reduce the proof to the complex case, and
Proposition~\ref{brelation} follows from the lemma below.
In the complex case, the boundary map $\partial_1^{\exp} \colon K_1(\mathcal{B}) \to K_0(\mathcal{I})$ is defined by forgetting the real structure in the case of $i=1$ of Definition~\ref{boundarymap}.
\begin{lemma}\label{lemma}
The boundary maps $\partial_1^{\BL}$ and $\partial_1^{\exp}$ from $K_1(\mathcal{B})$ to $K_0(\mathcal{I})$ satisfy the relation $\partial_1^{\BL} = - \partial_1^{\exp}$.
\end{lemma}
\begin{proof}
Suppose that $[u] \in K_1(\mathcal{B})$ where $u \in M_n(\tilde{\mathcal{B}})$ and $\lambda(u)=1_n$.
We take a lift $a$ of $u$ in $M_n(\tilde{A})$ satisfying $||a|| \leq 1$.
Consider the partial isometry
$v = \left(
    \begin{array}{cc}
           a&0\\
           \sqrt{1-a^*a} \hspace{-1mm}&0
    \end{array}
\right)$
and let
$V = \left(
    \begin{array}{cc}
           0&v^* \hspace{-1mm} \\
           v&0
    \end{array}
\right)$.
$\partial_1^{\exp}([u])$ is computed as
\begin{align*}
	\partial_1^{\exp}([u]) &= \left[ -Y_{4n}^{(1)}\bigl( \epsilon_1 \exp(\pi V \epsilon_1) \bigl) Y_{4n}^{(1)*} \right]\\
	&=
	[Y_{4n}^{(1)}(1-2vv^*)Y_{4n}^{(1)^*}] - [Y_{4n}^{(1)}(1-2v^*v)Y_{4n}^{(1)^*}].
\end{align*}
As in \cite{RLL00}, $\partial_1^{\BL}([u])$ is also expressed by using $v$, which is $-\partial_1^{\exp}([u])$.
\end{proof}

\subsection*{Acknowledgments}
The author thanks Takeshi Nakanishi, Ryo Okugawa, Guo Chuan Thiang and Yukinori Yoshimura for the many discussions.
This work was supported by JSPS KAKENHI (Grant Nos. JP17H06461, JP19K14545) and JST PRESTO (Grant No. JPMJPR19L7).

\end{document}